\documentclass[a4,11pt]{article}
\usepackage{amsthm} 
\usepackage{a4}
\usepackage{graphics}
\usepackage{amssymb}
\usepackage{latexsym}
\usepackage{amsmath}
\usepackage[all]{xy}
\usepackage{hyperref}
\usepackage{ifthen}
\usepackage{calc}
\usepackage[utf8]{inputenc}
\usepackage{enumitem}
\newlist{ealph}{enumerate}{1}
\setlist[ealph]{label=(\Alph*)}

\usepackage{tikz}
\newcommand{\nc}{\newcommand}
\nc{\look}{\marginpar{\makebox(0,0){\raisebox{22pt}{$\bullet$}}}}
\newtheorem{theo}{Theorem}[section]
\newtheorem{ddef}[theo]{Definition}
\newtheorem{llem}[theo]{Lemma} 
\newtheorem{oobs}[theo]{Observation} 
\newtheorem{rrem}[theo]{Remark} 
\newtheorem{prop}[theo]{Proposition} 
\newtheorem{ccor}[theo]{Corollary} 
\newtheorem{claim}[theo]{Claim} 
\newtheorem{fact}[theo]{Fact} 
\newtheorem{pprov}[theo]{Proviso}
\newtheorem{qquest}[theo]{Question}
\newtheorem{eexam}[theo]{Example} 
\newtheorem{exercise}{Exercise}[section] 

\nc{\bT}{\begin{theo}}
\nc{\eT}{\end{theo}}
\nc{\bD}{\begin{ddef} \rm }
\nc{\eD}{\end{ddef}}
\nc{\bExe}{\begin{exercise} \rm }
\nc{\eExe}{\end{exercise}}
\nc{\bC}{\begin{ccor}}
\nc{\eC}{\end{ccor}}
\nc{\bCl}{\begin{claim}}
\nc{\eCl}{\end{claim}}
\nc{\bL}{\begin{llem}}
\nc{\eL}{\end{llem}}
\nc{\bP}{\begin{prop}}
\nc{\eP}{\end{prop}}
\nc{\bR}{\begin{rrem}}
\nc{\eR}{\end{rrem}}
\nc{\bO}{\begin{oobs}}
\nc{\eO}{\end{oobs}}
\nc{\bF}{\begin{fact}}
\nc{\eF}{\end{fact}}
\nc{\bProv}{\begin{pprov}}
\nc{\eProv}{\end{pprov}}
\nc{\bQ}{\begin{qquest}}
\nc{\eQ}{\end{qquest}}
\nc{\bE}{\begin{eexam} \rm }
\nc{\eE}{\end{eexam}}

\nc{\prf}{\begin{proof}}
\nc{\eprf}{\end{proof}}

\renewcommand{\phi}{\varphi}
\renewcommand{\geq}{\geqslant}
\renewcommand{\leq}{\leqslant}

\renewcommand{\subset}{\subseteq}

\nc{\cutout}[1]{\mbox{}}
\nc{\seelabel}[1]{\protect\label{#1}}

\newenvironment{romanenumerate}%
{\begin{list}{{\rm (\roman{enumi})}}{\usecounter{enumi}
\setlength{\labelwidth}{2cm}
\setlength{\itemindent}{0pt}
\setlength{\itemsep}{0.5\itemsep}
\setlength{\topsep}{\itemsep}
\setlength{\parsep}{0pt}
}}{\end{list}}
\nc{\bre}{\begin{romanenumerate}}
\nc{\ere}{\end{romanenumerate}}
\newenvironment{alphaenumerate}%
{\begin{list}{{\rm (\alph{enumii})}}{\usecounter{enumii}
\setlength{\labelwidth}{2cm}
\setlength{\itemindent}{0pt}
\setlength{\itemsep}{0.5\itemsep}
\setlength{\topsep}{\itemsep}
\setlength{\parsep}{0pt}
}}{\end{list}}
\nc{\bae}{\begin{alphaenumerate}}
\nc{\eae}{\end{alphaenumerate}}
\newenvironment{numenumerate}%
{\begin{list}{{\rm (\arabic{enumiii})}}{\usecounter{enumiii}
\setlength{\labelwidth}{2cm}
\setlength{\itemindent}{0pt}
\setlength{\itemsep}{0.5\itemsep}
\setlength{\topsep}{\itemsep}
\setlength{\parsep}{0pt}
}}{\end{list}}
\nc{\bne}{\begin{numenumerate}}
\nc{\ene}{\end{numenumerate}}

\newenvironment{fll}%
{\btab{@{}l}}%
{\etab}%
\newenvironment{tfll}%
{\btab[t]{@{}l}}%
{\etab}%
\nc{\barr}{\begin{array}}
\nc{\earr}{\end{array}}
\nc{\btab}{\begin{tabular}}
\nc{\etab}{\end{tabular}}
\nc{\bfll}{\begin{fll}}
\nc{\efll}{\end{fll}}
\nc{\btfll}{\begin{tfll}}
\nc{\etfll}{\end{tfll}}

\nc{\LLL}{\mbox{{$\cal L$}}}
\nc{\PL}{\mbox{\rm PL}}
\nc{\MSOL}{\mbox{\rm MSOL}}
\nc{\msol}{{\rm MSOL}}
\nc{\LTL}{\mbox{\rm LTL}}
\nc{\CTL}{\mbox{\rm CTL}}
\nc{\Lmu}{\mbox{${\rm L}_\mu$}}
\nc{\PML}{\mbox{\rm PML}}

\nc{\FOL}{\textsf{\upshape FO}}
\nc{\LL}{\textsf{\upshape L}}
\nc{\LC}{\textsf{\upshape C}}
\nc{\ISO}[1][]{\textsc{ISO}\ifthenelse{\equal{#1}{}}{}{(#1)}}
\nc{\IL}[1][]{\textsc{IL}\ifthenelse{\equal{#1}{}}{}{(#1)}}
\nc{\CONT}[1]{\textsc{Cont}(#1)}
\nc{\COMP}[1]{\textsc{Comp}(#1)}
\nc{\comp}{\textsc{Comp}}
\nc{\MATCH}[1]{\textsc{Match}(#1)}
\nc{\BISO}[1][]{\textsc{B-ISO}\ifthenelse{\equal{#1}{}}{}{(#1)}}

\nc{\NE}{{\mathbb X}}
\nc{\EV}{{\mathbb F}}
\nc{\AL}{{\mathbb G}}
\nc{\UN}{{\mathbb U}}

\nc{\PI}{{\bf I}}
\nc{\PII}{{\bf II}}

\nc{\str}[1]{{\mathcal{#1}}}

\nc{\FTS}{\mbox{\rm FTS}}
\nc{\TS}{\mbox{\rm TS}}
\nc{\WS}{\mbox{\rm WS}}
\nc{\FTSY}{\mbox{\rm FS}}
\nc{\TSY}{\mbox{\rm S}}
\nc{\T}{\mbox{\rm T}}

\nc{\Ptime}{\mbox{\sc P}}
\nc{\NP}{\mbox{\sc NP}}
\renewcommand{\Bar}[1]{{\bf #1}}

\nc{\abar}{\Bar{a}}
\nc{\bbar}{\Bar{b}}
\nc{\cbar}{\Bar{c}}
\nc{\dbar}{\Bar{d}}

\nc{\nbar}{\Bar{n}}
\nc{\mbar}{\Bar{m}}

\nc{\rbar}{\Bar{r}}
\nc{\sbar}{\Bar{s}}
\nc{\tbar}{\Bar{t}}
\nc{\ubar}{\Bar{u}}
\nc{\vbar}{\Bar{v}}
\nc{\wbar}{\Bar{w}}
\nc{\xbar}{\Bar{x}}
\nc{\ybar}{\Bar{y}}
\nc{\zbar}{\Bar{z}}

\nc{\Ibar}{\Bar{I}}

\nc{\Pbar}{\Bar{P}}
\nc{\Ubar}{\Bar{U}}
\nc{\Vbar}{\Bar{V}}
\nc{\Wbar}{\Bar{W}}
\nc{\Xbar}{\Bar{X}}
\nc{\Ybar}{\Bar{Y}}
\nc{\Zbar}{\Bar{Z}}

\nc{\starred}[1]{#1^{\ast}}
\nc{\ustarred}[1]{#1_{\ast}}
\nc{\strich}{\rule{0em}{1ex}'}
\nc{\nothing}{\rule{0em}{1ex}}
\nc{\nt}{\nothing}
\nc{\highnothing}{\rule{0em}{3ex}}
\nc{\hnt}{\highnothing}
\nc{\vhnt}{\rule{0em}{4ex}}
\nc{\vvhnt}{\rule{0em}{6ex}}
\nc{\s}[1]{\mbox{\boldmath $#1$}}

\nc{\restr}{\!\restriction\!}
\nc{\ssc}{\scriptscriptstyle}
\nc{\dom}{\mathrm{dom}}
\nc{\range}{{\rm im}}
\nc{\tp}{\mathrm{tp}}
\nc{\sktp}{\mathrm{tp}}
\nc{\atp}{\mathrm{atp}}
\nc{\etp}{\mathrm{etp}}
\nc{\bequiv}{\equiv_{\ssc \mathrm{bool}}}

\nc{\brck}[1]{[\![#1]\!]}
\nc{\rep}[2]{[\![#1^{\ssc \str{#2}}]\!]}

\nc{\unt}[1]{\underline{#1}}

\nc{\integers}{{\mathbb Z}}

\nc{\card}[2]{\#_{#1} \bigl( #2 \bigr)}
\nc{\Card}[2]{\#_{#1} \Bigl( #2 \Bigr)}

\nc{\calA}{\mbox{$\cal A$}}
\nc{\calO}{\mbox{$\cal O$}}
\nc{\calQ}{\mbox{$\cal Q$}}
\nc{\calI}{\mbox{$\cal I$}}
\nc{\calF}{\mbox{$\cal F$}}
\nc{\calK}{\mbox{$\cal K$}}
\nc{\calC}{\mbox{$\cal C$}}
\nc{\calP}{\mbox{$\cal P$}}

\nc{\ealm}{\exists^{\klgeq m}}
\nc{\eals}{\exists^{\klgeq s}}
\nc{\eem}{\exists^{\kleq m}}
\nc{\ees}{\exists^{\kleq s}}

\nc{\diag}{D}

\nc{\fracsimeq}{\approx}
\nc{\fracsimeqbool}{\approx_{\ssc \mathrm{bool}}}

\nc{\ins}[1]{\begin{quotation} \sloppy \noindent \bf #1 \end{quotation}}

\newcounter{menum}

\nc{\rightreach}{\rightsquigarrow}
\nc{\leftrightreach}{\leftrightsquigarrow}

\nc{\E}{\exists}
\nc{\A}{\forall}
\nc{\N}{{\mathbb N}}
\nc{\B}{{\mathbb B}}
\nc{\Q}{{\mathbb Q}}
\nc{\R}{{\mathbb R}}
\nc{\Z}{{\mathbb Z}}

\renewcommand{\epsilon}{\varepsilon}

\usepackage{color}

\newcounter{rbcounter}
\renewcommand{\vec}[1]{\mathbf{#1}}

\newcommand{\CA}{\str{A}}
\newcommand{\CB}{\str{B}}
\newcommand{\CG}{\str{G}}
\newcommand{\CH}{\str{H}}
\newcommand{\CK}{\str{K}}
\newcommand{\CT}{\str{T}}
\newcommand{\CX}{\str{X}}
\newcommand{\hCX}{\widehat{\str{X}}}
\newcommand{\hX}{\widehat{X}}
\newcommand{\conc}{\,\widehat{\ }\,}
\newcommand{\IFF}{\,\Leftrightarrow\,}

\begin{document}

\title{Pebble Games and Linear Equations}

\author{Martin Grohe\\
\normalsize RWTH Aachen University
\and 
Martin Otto\\
\normalsize Technische Universit\"at Darmstadt}

\maketitle

\begin{abstract}
We give a new, simplified and detailed account of the 
correspondence between levels of the Sherali--Adams relaxation of 
graph isomorphism and levels of pebble-game equivalence with
counting (higher-dimensional Weisfeiler--Lehman colour refinement).
The correspondence between basic colour refinement and fractional
isomorphism,
due to Tinhofer~\cite{tin86,tin91} and Ramana, Scheinerman and Ullman~\cite{RamanaScheinermanUllman}, 
is re-interpreted as the base level of Sherali--Adams  
and generalised to higher levels in this sense by Atserias and 
Maneva~\cite{AtseriasManeva} and Malkin~\cite{mal14}, who prove that the two resulting 
hierarchies interleave. 
In carrying this analysis further, we here give
(a) a precise 
characterisation of the level~$k$ Sherali--Adams relaxation in terms
of a modified counting pebble game;
(b) a variant of the Sherali--Adams levels that precisely match
the $k$-pebble counting game; 
(c) a proof that the interleaving between 
these two hierarchies is strict.
We also investigate the variation based on boolean arithmetic instead 
of real/rational arithmetic and obtain 
analogous correspondences and separations for 
plain $k$-pebble equivalence (without counting). Our
results are driven by considerably simplified accounts of the
underlying combinatorics and linear algebra.
\end{abstract}

\tableofcontents 

\section{Introduction}
We study a surprising connection between equivalence in
finite variable logics and a linear programming approach to the graph
isomorphism problem. This connection has recently been uncovered by
Atserias and Maneva~\cite{AtseriasManeva} and, independently,
Malkin~\cite{mal14},
building on earlier work 
of Tinhofer~\cite{tin86,tin91} and Ramana, Scheinerman and Ullman~\cite{RamanaScheinermanUllman} that just
concerns the 2-variable case.

Finite variable logics play a central role in finite model theory. Most
important for this paper are finite variable logics with counting, which have
been specifically studied in connection with the question for a logical
characterisation of polynomial time and in connection with the graph
isomorphism problem (e.g.\
\cite{caifurimm92,graott93,gro10,immlan90,lau10,ott97}). Equivalence in finite variable logics can be
characterised in terms of simple combinatorial games known as pebble games.
Specifically, $\LC^k$-equivalence can be characterised by the bijective
  $k$-pebble game introduced by Hella~\cite{hel92}. Cai, F\"urer and
Immerman \cite{caifurimm92} observed that $\LC^k$-equivalence exactly
corresponds to indistinguishability by the $k$-dimensional Weisfeiler-Lehman
(WL) algorithm,\footnote{The dimensions of the WL algorithm are counted
  differently in the literature; what we call ``$k$-dimensional'' here is
  sometimes called ``$(k-1)$-dimensional''.} a combinatorial graph
isomorphism algorithm that goes back to work of
Weisfeiler and Lehman in the 1970s (for example, \cite{wei76}; see
\cite{caifurimm92} for an account of the history of the algorithm).
The 2-dimensional version of the WL
algorithm precisely corresponds to an even simpler isomorphism algorithm known
as colour refinement.  

The isomorphisms between two graphs can be described by the integral
solutions of a system of linear equations and inequalities. If we have two graphs with
adjacency matrices $A$ and $B$, then each isomorphism from the first
to the second corresponds to a permutation matrix $X$ such that
$X^tAX=B$, or equivalently
\begin{equation}
  \label{eq:iso-lp}
  AX=XB.
\end{equation}
If we view the entries of $X$ as variables, this equation
corresponds to a system of linear equations. We can add inequalities
that force $X$ to be a permutation matrix and obtain a system $\ISO$ of linear
equations and inequalities whose integral solutions correspond to the
isomorphisms between the two graphs.
In particular, the system $\ISO$ has
an integral solution if and only if the two graphs are isomorphic.

What happens if we drop the integrality constraints, that is, if we admit
arbitrary real solutions of the system $\ISO$? 
We can ask for doubly stochastic matrices $X$ satisfying 
equation \eqref{eq:iso-lp}. (A real matrix is
\emph{doubly stochastic} if its entries are non-negative and all row sums
and column sums are one.) Tinhofer~\cite{tin86,tin91} proved a
beautiful result that establishes a connection between linear algebra
and logic: the
system $\ISO$ has a real solution if, and only if, the colour refinement
algorithm does not distinguish the two graphs with
adjacency matrices $A$ and $B$. Recall that the latter is equivalent
to the two graphs being $\LC^2$-equivalent.

To bridge the gap between integer linear programs and their
LP-relaxations, researchers in combinatorial optimisation often add
additional constraints to the linear programs to bring them closer to
their integer counterparts. The Sherali--Adams hierarchy
\cite{sheada90} of relaxations
gives a systematic way of doing this. For every integer linear program
$\IL$ in $n$ variables and every positive integer $k$, there is
a \emph{rank-$k$
  Sherali--Adams relaxation} $\IL[k]$ of $\IL$, such that $\IL[1]$ is the
standard LP-relaxation of $\IL$ where all integrality constraints are
dropped and $\IL[n]$ is equivalent to $\IL$. There is a considerable body
of research studying the strength of the various levels of this and
related 
hierarchies (e.g.~\cite{bieozb04,burgalhoo+03,charmakmak09,matsin09,schtretul07,sch08}). 

Quite surprisingly, Atserias and Maneva~\cite{AtseriasManeva} and
Malkin~\cite{mal14} were
able to lift Tinhofer's result, which we may now restate as an
equivalence between $\ISO[1]$ and $\LC^2$-equivalence, to a close
correspondence between the higher levels of the Sherali--Adams
hierarchy for $\ISO$ and the logics $\LC^k$. They proved for every $k\ge 2$:
\begin{enumerate}
\item if $\ISO[k]$ has a (real) solution, then the two graphs 
are $\LC^k$-equivalent;
\item if the two graphs are $\LC^k$-equivalent, then $\ISO[k-1]$ has a solution.
\end{enumerate}
Atserias and Maneva~\cite{AtseriasManeva} used these results to transfer results about the
logics $\LC^k$ to the world of polyhedral combinatorics and combinatorial optimisation, and
conversely, results about the Sherali--Adams hierarchy to logic.

Atserias and Maneva~\cite{AtseriasManeva} left open the question whether the
interleaving between the levels of the Sherali--Adams hierarchy and the
finite-variable-logic hierarchy is strict or whether either the correspondence
between $\LC^k$-equivalence and $\ISO[k]$ or the correspondence between
$\LC^{k}$-equivalence and $\ISO[k-1]$ is exact. Note that for $k=2$ the
correspondence between $\LC^{k}$-equivalence and $\ISO[k-1]$ is exact by
Tinhofer's theorem. We prove that for all $k\ge3$ the
interleaving is strict. However, we can prove an exact correspondence between
$\ISO[k-1]$ and a variant of the bijective $k$-pebble game that
characterises $\LC^k$-equivalence. This variant, which we call the weak
bijective $k$-pebble game, is actually equivalent to a game called
$(k-1)$-sliding game by Atserias and Maneva.

Furthermore, we prove that a natural combination of
equalities from $\ISO[k]$ and $\ISO[k-1]$ gives a linear
program $\ISO[k-1/2]$ that characterises $\LC^k$-equivalence
exactly. Malkin~\cite{mal14} gives an alternative characterisation of
$\LC^k$-equivalence in terms of the Sherali-Adams relaxations of a
different polytope. Building on our work, yet another algebraic characterisation of 
$\LC^k$-equivalence has recently been given in
\cite{bergro15}.

To obtain these results, we give simple new, and arguably simpler proofs of the theorems of
Tinhofer and of Atserias and Maneva. 
In fact, the linear algebra we use is so simple that much of it can be
carried out not only over the field of real numbers, but over
arbitrary semirings. By using similar algebraic arguments over the
boolean semiring (with disjunction as addition and conjunction as
multiplication), we obtain analogous results to those for
$\LC^k$-equivalence for the ordinary $k$-variable logic $\LL^k$, characterising
$\LL^k$-equivalence, i.e., $k$-pebble game equivalence without counting, 
by systems of `linear' equations over the boolean semiring.

\section{Finite variable logics and pebble games}
\label{sec:logic}
We assume the reader to be familiar with the basics of first-order logic
\FOL. We almost exclusively consider first-order logic over finite
graphs, which we view as finite relational structures with one binary
relation. We
assume graphs to be undirected and loop-free.
For every positive integer $k$, we let $\LL^k$ be the fragment of
$\FOL$ consisting of all formulae that contain at most $k$ distinct
variables. 

We write $\str{A}\equiv_\LL^k \str{B}$ to denote that two structures
$\str{A},\str{B}$ are \emph{$\LL^k$-equivalent}, that is, satisfy the same
$\LL^k$-sentences. $\LL^k$-equivalence can be characterised in terms
of the \emph{$k$-pebble game}, played by two players on a pair $\str{A},\str{B}$
of structures. A \emph{play} of the game consists of a (possibly
infinite) sequence of \emph{rounds}. In each round, player \PI\
picks up one of his pebbles and places it on an element of one of the
structures $\str{A},\str{B}$.
Player \PII\ answers by picking up her pebble with the same label
and placing it on an element of the other structure.

Note that after each round $r$ there is a subset $p\subseteq \str{A}\times
\str{B}$ consisting of the at most $k$ pairs of elements on which the pairs
of corresponding pebbles are placed. We call $p$ the \emph{position}
after round $r$. Player \PII\ wins the play if every position that
occurs is a local isomorphism, that is, a local mapping from $\str{A}$
to $\str{B}$ that is injective and preserves membership and non-membership
in all relations (adjacency and non-adjacency if $\str{A}$ and $\str{B}$
are graphs). 

\begin{theo}[Barwise~\cite{bar77}, Immerman~\cite{imm82}]
  $\str{A}\equiv_\LL^k \str{B}$ if, and only 
if, player \PII\ has a winning strategy
  for the $k$-pebble game on $\str{A},\str{B}$.
\end{theo}

We extend
$\LL^k$-equivalence to structures with distinguished
elements. For tuples $\vec a$ and $\vec b$ of the same length $\ell\le
k$ we let $\str{A},\vec a\equiv_\LL^k\str{B},\vec b$ if $\str A,\vec
a$ and $\str{B},\vec b$ satisfy the same $\LL^k$-formulae $\phi(\vec
x)$ with $\ell$ free variables $\vec x$. The pebble game
characterisation extends: $\str{A},\vec a\equiv_\LL^k\str{B},\vec b$
if, and only if, player \PII\ has a winning strategy for the
$k$-pebble game on $\str{A},\str{B}$ starting with pebbles on $\vec a$
and the corresponding pebbles on $\vec b$. The \emph{$\LL^k$-type} of
a tuple $\vec a$ in a structure $\str A$ is the $\equiv_\LL^k$-equivalence
class of $\str A,\vec a$. More syntactically, we may also view the $\LL^k$-type of $\vec
a$ as the set of all $\LL^k$-formulae $\phi(\vec
x)$ satisfied by $\str A,\vec a$.

Let us turn to the $k$-variable counting logics. It is convenient to start with the (syntactical) extension $\LC$ of $\FOL$ by \emph{counting quantifiers}
$\exists^{\ge n}$. The semantics of these counting quantifiers is the
obvious one: $\exists^{\ge n}x\,\phi$ means that there are at least $n$
elements $x$ such that $\phi$ is satisfied. Of course this can be
expressed in $\FOL$, but only by a formula that uses at least $n$
variables. For all positive integers $k$, we let $\LC^k$ denote the
$k$-variable fragment of $\LC$. Whereas $\LC$ and $\FOL$ have the same
expressive power, $\LC^k$ is strictly more expressive than
$\LL^k$. 

We write $\str{A}\equiv_\LC^k \str{B}$ to indicate that structures $\str{A}$ and $\str{B}$ are
$\LC^k$-equivalent. $\LC^k$-equivalence can be characterised in terms
of the \emph{bijective $k$-pebble game},\label{bijective-game} which, like the $k$-pebble
game, is played by two players by placing $k$ pairs of pebbles on a
pair of structures $\str{A},\str{B}$. The rounds of the bijective game are as
follows. Player \PI\ picks up one of his pebbles, and player \PII\
picks up her corresponding pebble. Then player \PII\ chooses a
bijection $f$ between $\str{A}$ and $\str{B}$ (if no such bijection exists, that
is, if the structures have different cardinalities, player \PII\
immediately loses). Then player \PI\ places his pebble on an element
$a$ of $\str{A}$, and player \PII\ places her pebble on $f(a)$.
Again, player \PII\ wins a play if all positions are local
isomorphisms.

\begin{theo}[Hella~\cite{hel92}]
  \label{theo:hel}
  $\str{A}\equiv_\LC^k \str{B}$ if, and only if, player \PII\ has a winning strategy
  for the bijective $k$-pebble game on $\str{A},\str{B}$.
\end{theo}

As with $\LL^k$-equivalence, we extend $\LC^k$-equivalence to structures
with distinguished elements, writing $\str A,\vec a\equiv_\LC^k\str
B,\vec b$. Again, the pebble-game characterisation of the equivalence
extends. We define $\LC^k$-types analogously to $\LL^k$-types.

The \emph{colour refinement} algorithm is a simple combinatorial
heuristic for testing whether two graphs are isomorphic. Given two
graphs $\str{A}$ and $\str{B}$, which we assume to be disjoint, it
computes a colouring of their vertices by the following iterative
procedure: Initially, all vertices have the same colour. Then in each
round, the colouring is refined by assigning different colours to
vertices that have a different number of neighbours of at least one
colour assigned in the previous round. Thus after the first round, two
vertices have the same colour if, and only if, they have the same
degree. After the second round, two vertices have the same colour if,
and only if, they have the same degree and for each $d$ the same
number of neighbours of degree $d$. The algorithm stops if no further
refinement is achieved; this happens after at most
$|\str{A}|+|\str{B}|$ rounds. We call the resulting colouring of
$\str{A}\cup \str{B}$ the \emph{stable colouring} of
$\str{A},\str{B}$. If the stable colouring differs on the two graphs,
that is, for some colour $c$ the graphs have a different number of
vertices of colour $c$, then we say that colour refinement
\emph{distinguishes} the graphs.

\begin{theo}[Immerman and Lander~\cite{immlan90}]\label{theo:immlan}
  $\str{A}\equiv_\LC^2 \str{B}$ if, any only if, 
colour refinement does not distinguish
  $\str{A}$ and $\str{B}$.
\end{theo}

The \emph{$k$-dimensional Weisfeiler-Lehman algorithm} (for short:
$k$-WL) is a generalisation of the colour
refinement algorithm, which instead of vertices colours $k$-tuples of vertices.
Given two structures $\str{A}$ and $\str{B}$, which we assume to be disjoint, $k$-WL
iteratively computes a colouring of $\str{A}^k\cup \str{B}^k$.  Initially,
two tuples $\vec a=(a_1,\ldots,a_k),\vec b=(b_1,\ldots,b_k)\in \str{A}^k\cup \str{B}^k$ get the same colour if the
mapping defined by $p(a_i)=b_i$ 
is a local isomorphism. In each round of the algorithm, the
colouring is refined by assigning different colours to tuples that for
some $j\in[k]$ and some colour $c$ have different numbers of
$j$-neighbours of colour $c$ in their respective graphs. Here we call
two $k$-tuples \emph{$j$-neighbours} if they differ only in their
$j$th component.
The algorithm
stops if no further refinement is achieved; this happens after at most
$|\str{A}|^k+|\str{B}|^k$ rounds. If after the refinement
process the colourings of the two graphs differ, that is, for some
colour $c$ the graphs have a different number of $k$-tuples of colour $c$, then
we say that $k$-WL \emph{distinguishes} the
graphs. 

\begin{theo}[Cai, F\"urer, and Immerman~\cite{caifurimm92}]
  $\str{A}\equiv_\LC^k \str{B}$ if, and only if, 
$k$-WL does not distinguish $\str{A}$ and $\str{B}$.
\end{theo}

More significantly, Cai, F\"urer, and Immerman~\cite{caifurimm92}
proved that for all $k$ there are nonisomorphic graphs $\str{A}_k,\str{B}_k$ of size $O(k)$ such
that $\str{A}\equiv_\LC^k \str{B}$.

Note that the previous two theorems imply that colour refinement and
$2$-WL distinguish the same graphs.

There are also `boolean'
versions of the two algorithms 
characterising $\LL^k$-equivalence (see \cite{ott97}).

\section{Basic combinatorics and linear algebra}%
\label{comlinalgsec}
We consider matrices with entries 
in $\B = \{0,1\}$, $\Q$ or $\R$. A matrix $X \in \R^{m,n}$ with
$m$ rows and $n$ columns has entry $X_{ij}$ in
row $i \in [m] = \{ 1,\ldots, m\}$ and 
column $j \in [n] = \{ 1,\ldots, n \}$.
We write $E_n$ for the $n$-dimensional unit matrix.

We write $X \geq 0$ to
say that (the real or rational) matrix $X$ has only non-negative entries,
and $X > 0$ to say that all entries are strictly positive. We also
speak of \emph{non-negative} or \emph{strictly positive matrices} 
in this sense. For a boolean matrix $X$, strict positivity $X > 0$ means
that all entries are $1$.

A square $n\!\times\!n$-matrix 
is \emph{doubly stochastic} 
if its entries are non-negative and 
if the sum of entries across every row and column is 
$1$. \emph{Permutation matrices} are doubly stochastic matrices 
over $\{0,1\}$, with precisely one $1$ in every row and in every
column. Permutation matrices are orthogonal, 
i.e., $P P^t = P^t P = E_n$ for every permutation matrix $P$. 
The permutation $p \in S_n$ associated with a
permutation matrix $P \in \R^{n,n}$ is such that 
$P \mathbf{e}_j = \mathbf{e}_{p(j)}$, i.e., it describes the  
permutation of the standard basis vectors $\mathbf{e}_j$ that is effected by
$P$. 
We also say that $P$ represents $p$. 
The permutation matrices form a subgroup 
of the general linear groups. The doubly stochastic matrices 
do not form a subgroup, but are closed under transpose and product.

\medskip
It will be useful to have the shorthand notation 
\[
X_{D_1D_2} = 0
\]
for the assertion that $X_{d_1d_2} = 0$ for all $d_1 \in D_1$, $d_2 \in D_2$.
If $p$ and $q$ are permutations in $S_n$ represented by permutation 
matrices $P$ and $Q$, then 
\[
(P^tXQ)_{D_1D_2} = 0 \quad\mbox{ iff }\quad 
X_{p(D_1)q(D_2)} = 0.  
\]
So, if $X_{D_1D_2} = 0$ and $P$ and $Q$ are chosen such that $p^{-1}(D_1)$ and 
$q^{-1}(D_2)$ are final and initial segments of $[n]$, respectively, then 
$P^t X Q$ has a null block of dimensions $|D_1| \times |D_2|$ in the 
upper right-hand corner. 

\subsection{Decomposition into irreducible blocks}

\bD
\label{irreddef}
With $X \in \R^{n,n}$ %
associate the directed graph 
\[
\mathrm{G}(X) := ([n], \{ (i,j) \colon X_{ij} \not= 0\}).
\] 
The strongly connected components of $\mathrm{G}(X)$ induce a 
partition of the set $[n] = \{ 1,\ldots, n\}$ of rows/columns of $X$. 
$X$ is called \emph{irreducible} if this partition has just 
the set $[n]$ itself. 
\eD

Note that $X$ is irreducible iff
$P^t X P$ is irreducible for every permutation matrix $P$. 

\bO
\label{posobs}
Let $X \in \R^{n,n} \geq 0$ with 
strictly positive diagonal entries. 
If $X$ is irreducible, then all powers
$X^\ell$ for $\ell \geq n-1$ have non-zero 
entries throughout.
Moreover, if $X$ is irreducible, then so is $X^\ell$ for all $\ell \geq 1$.
\eO

\prf
It is easily proved by induction on  $\ell \geq 1$
that $(X^\ell)_{ij} \not= 0$ if, and only if 
there is a directed path of length $\ell$ from vertex $i$ to vertex $j$
in $G(X)$. For $X$ with positive diagonal entries, $G(X)$ 
has loops in every vertex, and therefore there is a  
path of length $\ell$ from vertex $i$ to vertex $j$ if, and only if, there is  
path of length $M'$ for every $\ell' \geq \ell$ from $i$ to $j$.  
If $G(X)$ is also strongly connected, then any two vertices 
are linked by a path of length up to $n-1$. 
\eprf

Let us call two matrices $Z,Z' \in \R^{n,n}$ \emph{permutation-similar}
or \emph{$S_n$-similar}, $Z \sim_{S_n} Z'$, 
if $Z' = P^t Z P$ for some permutation matrix $P$, i.e., if one is obtained 
from the other by simultaneously permuting rows and columns with the
same permutation. %

\bL
\label{symdecomplem}
Every symmetric $Z \in \R^{n,n} \geq 0$ 
is per\-mu\-ta\-tion-similar 
to some block diagonal matrix $\mathrm{diag}(Z_1,\ldots,Z_s)$ 
with irreducible blocks $Z_i \in \R^{n_i,n_i}$.

The permutation matrix $P$ corresponding to the row- and
column-permutation $p \in S_n$ that puts $Z$ into block diagonal 
form $P^t Z P = \mathrm{diag}(Z_1,\ldots,Z_s)$  with irreducible blocks, 
is unique up to an outer permutation that 
re-arranges the block intervals $([k_i+1,k_i+n_i])_{1\leq i \leq s}$
where $k_i = \sum_{j < i} n_j$,
and a product of inner permutations within each one of these 
$s$ blocks. 

The underlying partition 
$[n] = \dot{\bigcup}_{1 \leq i \leq s} D_i$ where $D_i :=
p([k_i+1,k_i+n_i])$
for $k_i = \sum_{j < i} n_j$, is uniquely determined by $Z$.%
\footnote{Here we regard two partitions as identical if they have 
the same partition sets, i.e., we ignore their indexing/enumeration.}
\eL

In the following we refer to the \emph{partition induced by 
a symmetric matrix $Z$}.

\prf
Obvious, based on the partition of the vertex set $[n]$ of $G(Z)$ 
into connected components (note that symmetry of $Z$ is preserved under 
similarity, and strong connectivity is plain connectivity in $G(Z)$ 
for symmetric $Z$).  
\eprf

\bO
\label{decomppartobs}
In the situation of Lemma~\ref{symdecomplem},
the partition $[n] = \dot{\bigcup}_i D_i$
induced by the symmetric matrix $Z$ is the 
partition of $[n]$ into the vertex sets of the connected components of $G(Z)$. 
Then, for every pair $i \not= j$, $Z_{D_i D_{j}} = 0$, 
while all the minors $Z_{D_iD_i}$ are irreducible.%
\footnote{Note that this does not depend on the enumeration of the 
partition set $D_i$, because irreducibility is invariant under 
permutation-similarity.}

If, moreover, $Z$ has strictly positive 
diagonal entries, then the partition 
induced by $Z$ is the same as that induced by $Z^\ell$, for any $\ell \geq 1$; 
for $\ell \geq n-1$, 
the diagonal blocks $(Z^\ell)_{D_iD_i}$ have
non-zero entries throughout: $(Z^\ell)_{D_iD_i} > 0$ . 
\eO

The last assertion says that 
for a symmetric $n\!\times\!n$ matrix $Z$ with non-negative entries 
and no zeroes on the diagonal, all powers $Z^\ell$ for $\ell \geq n-1$ 
are \emph{good symmetric} in the sense of the following definition.

\bD
\label{goodsymmatrixdef}
Let $Z \geq 0$ be symmetric with strictly positive diagonal. 
Then $Z$ is called \emph{good symmetric} if 
w.r.t.\ the partition $[n] = \dot{\bigcup}_i D_i$ induced by $Z$,
all $Z_{D_iD_i} > 0$.

More generally, a not necessarily symmetric matrix 
$X \geq 0$ without null rows or columns is \emph{good} 
if $Z = XX^t$ and $Z' = X^tX$ are good in the above sense. 
\eD

The importance of this notion lies in the fact that, as observed
above, for an arbitrary symmetric $n\!\times\!n$ matrix $Z \geq 0$  
without zeroes on the diagonal, the partition induced by $Z$ 
is the same as that induced by the good symmetric matrix $\hat{Z} := Z^{n-1}$;
and, as for any good matrix, this partition  
can simply be read off from $\hat{Z}$: 
$i,j \in [n]$ are in the same partition set if, and only if, 
$\hat{Z}_{ij} \not= 0$.

\bD
\label{relpartdef}
Consider partitions $[n] = \dot{\bigcup}_{i \in I} D_i$ and 
$[m] = \dot{\bigcup}_{i \in I} D_i'$ of the sets $[n]$ and $[m]$
with the same number of partition sets. We say that
these two partitions are \emph{$X$-related}
for some matrix $X \in \R^{n,m}$ if 
\bre
\item
$X \geq 0$ has no null rows or columns, and 
\item
$X_{D_i {D_{j}}\!\!\nt'} = 0$ for every pair of distinct indices 
$i,j \in I$.
\ere
\eD

Note that partitions that are $X$-related 
are $X^t$-related in the opposite direction.
More importantly, each one of the
$X/X^t$-related partitions can be recovered from 
the other one through $X$ according to
\[
\barr{rcl}
D_i' &=& \{ d' \in [m] \colon X_{dd'} > 0 \mbox{ for some } d \in D_i
\},
\\
\hnt
D_i &=& \{ d \in [n] \colon X_{dd'} > 0 
\mbox{ for some } d' \in D_i'\}.
\earr
\]
 
For a more algebraic treatment, we associate 
with the partition sets $D_i$ of a
partition $[n] = \dot{\bigcup}_{i \in I} D_i$ the
\emph{characteristic vectors} $\mathbf{d}_i$
with entries $1$ and $0$ according to whether
the corresponding component belongs to $D_i$:
\[
\textstyle \mathbf{d}_i = \sum_{d \in D_i} \mathbf{e}_d, 
\]
where $\mathbf{e}_d$ is the $d$-th standard basis vector. 
In terms of these characteristic vectors
$\mathbf{d}_i$ for  $[n] = \dot{\bigcup}_{i \in I} D_i$ 
and $\mathbf{d}_i'$  for $[m] = \dot{\bigcup}_{i \in I} D_i'$,
the $X/X^t$-relatedness of these partitions means that 
\[
\barr{rcl}
D_i' 
&=&
\{ d' \in [m] \colon ( X^t \mathbf{d}_i)_{d'} > 0 \},
\\
\hnt
D_i 
&=&
\{ d \in [n] \colon (X \mathbf{d}_i')_{d} > 0 \}.
\earr
\]

\bL
\label{stochrelatedpartlem}
If two partitions $[n] = \dot{\bigcup}_{i \in I} D_i$ and 
$[n] = \dot{\bigcup}_{i \in I} D_i'$ of the same set $[n]$ 
are $X$-related for some doubly stochastic matrix $X \in \R^{n,n}$,
then $|D_i| = |D_i'|$ for all $i \in I$, and for the characteristic
vectors $\mathbf{d}_i$ and $\mathbf{d}_i'$ of the partition sets
$D_i$ and $D_i'$ even
\[
\mathbf{d}_i = X \mathbf{d}_i'
\quad
\mbox{ and }
\quad
\mathbf{d}_i' = X^t \mathbf{d}_i.
\]
\eL

\prf
Observe that for all $d\in[n]$ we have $0\le (X\vec d_i')_d=\sum_{d'\in
  D_i'}X_{dd'}\le 1$. It follows immediately from the
definition of $X$-relatedness that $(X\vec d_i')_d =0$ for all
$d\not\in D_i$.
Therefore, 
\[
|D_i|\geq \sum_{d\in D_i}(X\vec d_i')_d=\sum_{d\in[n]}(X\vec d_i')_d
=\sum_{d'\in D_i'}\sum_{d\in[n]}X_{dd'}=|D_i'|.
\]

Similarly, $0\le (X^t\vec d_i)_{d'}\le 1$ for $d'\in[n]$, and
$|D_i'|\ge \sum_{d'\in D_i'}(X^t\vec d_i)_{d'}=|D_i|$.
Together, we obtain
\[
|D_i|=\sum_{d\in D_i}(X\vec d_i')_d=|D_i'|=\sum_{d'\in D_i'}(X^t\vec
d_i)_{d'}.
\]

As all summands are bounded by $1$, this implies $(X\vec d_i')_d=1$ for all
$d\in D_i$ and $(X^t\vec d_i)_{d'}=1$ for all $d'\in D_i$.  
\eprf

\bL
\label{newdecomplem}
Let $X \geq 0$ be an $m\!\times\!n$ matrix without 
null rows or columns. Then the $m\!\times\!m$ matrix 
$Z := X X^t$ and the $n\!\times\!n$ matrix 
$Z' := X^t X$ are symmetric with positive entries 
on their diagonals. Moreover, the (unique)
partitions of $[m]$ and $[n]$ that are induced 
by $Z$ and $Z'$, respectively, are $X/X^t$-related.%
\footnote{As $X/X^t$-relatedness refers to partitions presented with
an indexing of the partition sets, we need to allow a suitable 
re-indexing for at least one of them, so as to match the other one.}
\eL

\prf
It is obvious that $Z$ and $Z'$ are symmetric with positive diagonal
entries.
Let partitions $[m] = \dot{\bigcup}_{i\in I} D_i$ and 
$[n] = \dot{\bigcup}_{i\in I'} D_i'$ be 
obtained from decompositions of $Z$ and $Z'$ into irreducible blocks.
We need to show that the non-zero entries in $X$ 
give rise to a 
bijection between the index sets 
$I$ and $I'$ of the two partitions, in the sense that 
partition sets $D_i$ and $D_j'$ are related if, and only if, some 
pair of members $d \in D_i$ and $d' \in D_j'$ have a positive entry $X_{dd'}$.
Then a re-numbering of one of these partitions will make them
$X$-related in the sense of Definition~\ref{relpartdef}. 
Recall from Observation~\ref{decomppartobs} 
that the $D_i$ are the vertex sets of the connected components 
of $G(XX^t)$ on $[m]$, while  the $D_i'$ the are the
vertex sets of the connected components 
of $G(X^tX)$ on $[n]$. 

Consider the uniformly directed bipartite graph $G(X)$ 
on $[m] \,\dot{\cup}\, [n]$ with an edge from 
$i \in [m]$ to $j \in [n]$ if $X_{ij} > 0$.
In light of the symmetry of the whole situation w.r.t.\
$X$ and $X^t$, it just remains to argue for instance 
that no $i \in [m]$ can have edges into 
two distinct sets of the partition $[n] = \dot{\bigcup}_{i\in I'} D_i'$.
But any two target nodes of edges from one and the same  
$i \in [n]$ are in the same connected component of $G(X^t X)$, hence 
in the same partition set.  
\eprf

In the situation of Lemma~\ref{newdecomplem}, 
powers of $Z$ induce the same partitions as $Z$, and 
the partitions induced by $(Z^\ell X)(Z^\ell X)^t = Z^{2\ell+1}$
are $X/X^t$-related as well as 
$Z^\ell X/X^tZ^\ell$-related, for all $\ell \geq 1$.

For $\ell \geq n/2 - 1$, the matrix $Z^\ell X$ 
has no null rows or columns: 
else $Z^\ell X (Z^\ell X)^t = Z^{2\ell+1}$ would have to have a zero entry
on the diagonal, contradicting the fact 
that this symmetric matrix is good symmetric 
in the sense of Definition~\ref{goodsymmatrixdef}.
The same reasoning shows that $Z^\ell X$ is itself good  
in the sense of Definition~\ref{goodsymmatrixdef}.

\bC
\label{goodlinkcor}
Let $X \geq 0$ be an $m\!\times\!n$ matrix without null rows or
columns, $Z = XX^t$, $Z' = X^tX$ the associated 
symmetric matrices with non-zero entries on the diagonal. 
Then for $\ell \geq m-1, n-1$, 
the matrix $\hat{X} := Z^\ell X = X (Z')^\ell$ and its 
transpose $\hat{X}^t = X^t Z^\ell = (Z')^\ell X^t$ 
are good and relate the partitions $[m] = \dot{\bigcup}_i D_i$ and 
$[n] = \dot{\bigcup}_i D_i'$ induced by $Z$ and $Z'$, respectively.%
\addtocounter{footnote}{-1}\footnotemark\ 
Moreover, 
\bre
\item
$\hat{X}_{D_iD_i'}>0$ for all $i$, and 
\item
$\hat{X}_{D_iD_j'} = 0$  for all $i \not= j$.
\ere
\eC

\prf
$Z^\ell X$ is good symmetric by the above reasoning. So
$(Z^\ell)_{D_iD_i} > 0$  for all
$i$, while $(Z^m)_{D_i D_j} = 0$ for all $j \not= i$. It follows that  
$(Z^\ell X)_{D_i D_i'} = (Z^\ell)_{D_iD_i} X_{D_i D_i'}$ 
has only non-zero entries because $X_{D_i D_i'}$ 
does not have null columns. This proves~(i). 
Assertion~(ii) is clear as, for $i \not= j$, 
$(Z^\ell X)_{D_i D_j'} = (Z^\ell)_{D_iD_j} X_{D_j D_j'} = 0 \,
X_{D_jD_j'} = 0$.
\eprf

\subsubsection*{Aside: boolean vs.\ real arithmetic}
Looking at matrices with $\{0,1\}$-entries, we may not only treat them as 
matrices over $\R$ as we have done so far, but also over other fields, 
or as matrices over the boolean semiring $\B = \{ 0,1\}$ with the
logical operations of $\vee$ for addition and $\wedge$ for
multiplication. Though not even forming a ring, boolean arithmetic 
yields a very natural interpretation in the context where 
we associate non-negative entries with edges, as we did in passage 
from $X$ to $G(X)$ 
(cf.\ Definition~\ref{irreddef} and 
Observation~\ref{posobs}). 
The `normalisation map'
$\chi \colon \R_{\geq 0} \rightarrow \{ 0, 1\}$, $x \mapsto 1$ iff $x > 0$,
relates the arithmetic of reals $x,y \geq 0$
to boolean arithmetic in
\[
\chi(x + y) = \chi(x) \vee \chi(y) 
\quad
\mbox{ and }
\quad 
\chi(x y) = \chi(x) \wedge \chi(y). 
\]

This is the `logical' arithmetic that supports, for instance, 
arguments used in Observation~\ref{posobs}:
for any real $n\!\times\!n$ matrix $X \geq 0$,  
$(XX)_{ij} = \sum_{k} X_{ik} X_{kj} \not= 0$ iff there is at 
least one $k \in [n]$ for which $X_{ik} \not=0$ and $X_{kj} \not= 0$ 
iff $\bigvee_{k\in[n]} (\chi(X_{ik}) \wedge \chi(X_{kj})) = 1$. 
It is no surprise, therefore, that several of the considerations 
apparently presented for real non-negative matrices above, have
immediate analogues for boolean arithmetic -- in fact, one could
argue, that the boolean interpretation is closer to the combinatorial essence. 
We briefly sum up these analogues with a view to their use in the
analysis of $\LL^k$-equivalence, while the real versions are related to 
$\LC^k$-equivalence. Note also that the boolean analogue of a doubly stochastic
matrix with non-negative real entries is a matrix without null rows or 
columns.

Also note that Definitions~\ref{irreddef} (irreducibility) 
and~\ref{relpartdef}  ($X$-relatedness)
are applicable to boolean matrices without any changes.
Observations~\ref{posobs} and~\ref{decomppartobs} go through 
(as just indicated), and so does Lemma~\ref{symdecomplem}. For 
Lemma~\ref{stochrelatedpartlem}, one may look at $X$-related 
partitions of sets $[m]$ and $[n]$, where not necessarily 
$n=m$, by any boolean matrix $X$ without null rows or columns
and obtains the relationship between the characteristic vectors 
as stated there, now in terms of boolean arithmetic --
but of course we do not get any numerical equalities
between the sizes of the partition sets. 
Lemma~\ref{newdecomplem}, finally,
applies to boolean arithmetic, exactly as stated,
and also Corollary~\ref{goodlinkcor} translates accordingly. 

\bL
\label{summaryboollem} 
In the sense of boolean arithmetic for matrices with entries 
in $\B = \{ 0,1 \}$:
\bae
\item
Any symmetric $Z \in \B^{n,n}$ induces a unique partition of 
$[n]$ for which the diagonal minors induced by the partition sets  
are irreducible and the remaining blocks null; 
$d,d' \in [n]$ are in the same
partition set if, and only if, 
$(Z^\ell)_{dd'} = 1$ for any/all $\ell \geq n-1$. 
\item
If two partitions (not necessarily of the same set) with the same
number of partition sets are related by some boolean 
matrix $X \in \B^{m,n}$, then
the characteristic vectors $(\mathbf{d}_i)_{i \in I}$
and $(\mathbf{d}_i')_{i \in I}$
of the partitions are related by
$\mathbf{d}_i = X \mathbf{d}_i'$ and 
$\mathbf{d}_i' = X^t \mathbf{d}_i$.
\item
For any matrix $X \in \B^{m,n}$ without 
null rows or columns, the symmetric boolean matrices 
$Z = XX^t$ and $Z' = X^tX$ have diagonal entries $1$ 
and induce partitions that are 
$X/X^t$-related, and agree with the partitions induced by 
higher powers of $Z$ and $Z'$ or on the basis of 
$Z^\ell X$ and $X(Z')^\ell$ for any $\ell \in \N$.
For $\ell \geq m-1,n-1$, the partition blocks in $Z$ and $Z'$ have 
entries $1$ throughout, and $Z^\ell X$ and $X (Z')^\ell$ have entries 
$1$ in all positions relating elements from matching partition sets.
\eae
\eL

\bO
\label{inducedpartboolobs}
For a symmetric boolean matrix $Z \in \B^{n,n}$
with $Z_{dd} = 1$ for all $d \in [n]$,
the characteristic vectors $\mathbf{d}_i$ of the partition 
$[n] = \dot{\bigcup}_{i\in I} D_i$ 
induced by $Z$ satisfy the following `eigenvector' equation in 
terms of boolean arithmetic:
\[
Z \mathbf{d}_i = \mathbf{d}_i \quad \mbox{\rm (boolean), \; for all $i \in I$.}
\]
\eO

\subsection{Eigenvalues and -vectors}

\bL
\label{evallem}
If $Z \in \R^{n,n}$ is doubly stochastic, then it has eigenvalue $1$.
If $Z$ is doubly stochastic and irreducible with strictly 
positive diagonal entries, then the eigenspace for eigenvalue $1$ 
has dimension~$1$ and is spanned by the vector 
$\mathbf{d} := (1,\ldots,1)^t$. 
\eL

\prf
Clearly $Z \mathbf{d} = \mathbf{d}$ for any stochastic matrix $Z$.

If $Z$ is moreover irreducible with positive diagonal entries, then 
by Observation~\ref{posobs},
$\hat{Z} := Z^{n-1}$ has strictly positive entries and,  
being doubly stochastic, therefore entries 
strictly between $0$ and $1$. 

If $\mathbf{v}$ is an eigenvector 
for eigenvalue $1$ of $Z$, then it also is an eigenvector 
for eigenvalue $1$ of $\hat{Z}$.%
If $\mathbf{v} = (v_1,\ldots,v_n)$, this is equivalent to
\[
\textstyle
v_i = \sum_j 
\hat{Z}_{ij} v_j\quad \mbox{ for all } i \in [n].
\] 

Looking at an index $i$ for which $v_j \leq v_i$ for all $j$, we see
that the maximal $v_i$ is a convex combination of the 
$v_j$ to which every $v_j$ contributes. 
This implies that all $v_j = v_i$, so that 
$\mathbf{v}$ is a scalar multiple of $\mathbf{d}$ as claimed.
\eprf

\bC
\label{evalcor} 
\bae
\item
Let $Z \in \R^{n,n}$ be doubly stochastic 
with positive diagonal,  
and $[n] = \dot{\bigcup}_i D_i$ a partition with
$Z_{D_iD_j} = 0$ for $i \not= j$ and such that the minors  
$Z_{D_iD_i}$ are irreducible for all $i$. Then the eigenspace
for eigenvalue~$1$ of $Z$ is the direct sum of the 
$1$-dimensional subspaces spanned by the characteristic 
vectors $\mathbf{d}_i$ of the partition sets $D_i$. 
\item
If $Z = X^tX \in \R^{n,n}$ 
for some doubly stochastic matrix $X$, then 
the eigenspace for eigenvalue~$1$ is the direct sum of 
the spans of the characteristic vectors $\mathbf{d}_i$ 
from the unique partition 
$[n] = \dot{\bigcup}_i D_i$ of $[n]$ induced by 
$Z$ according to Lemma~\ref{symdecomplem}.
\eae
\eC

\prf
Towards~(a), it is clear that $Z \mathbf{d}_i= \mathbf{d}_i$, 
so that each $\mathbf{d}_i$ is an eigenvector with eigenvalue~$1$. 
Let $V_i := \mathrm{span}(\mathbf{e}_d \colon d
\in D_i)$; then $\R^n = \bigoplus_i V_i$ is a direct sum
decomposition, and 
$Z_{D_jD_i} = 0$ for $j \not= i$ implies that $Z$ maps $V_i$ to itself. 
Therefore any eigenvector $\mathbf{v}$ with eigenvalue~$1$ decomposes 
as $\mathbf{v} = \sum_i \mathbf{v}_i$, where 
$\mathbf{v}_i \in V_i$, in such manner that $Z \mathbf{v}_i =
\mathbf{v}_i$. Since the restriction of $Z$ to $V_i$ is 
irreducible with positive diagonal, $\mathbf{v}_i \in
\mathrm{span}(\mathbf{d}_i)$ by Lemma~\ref{evallem}, as claimed.

Statement~(b) is s direct consequence, since $Z$ is symmetric with
positive diagonal.  
\eprf

\subsection{Stable partitions}

\bD
\label {stablepartdef}
Let $A \in \R^{n,n}$, $[n] = \dot{\bigcup}_{i\in I} D_i$ be a partition.
We call this partition a \emph{stable partition for $A$} if 
there are numbers $(s_{ij})_{i,j \in I}$ and $(t_{ij})_{i,j \in I}$
such that for all $i,j \in I$:
\[
d \in D_i \quad \Rightarrow \quad \sum_{d' \in D_j} A_{dd'} = s_{ij}
\quad \mbox{ and }\quad
\sum_{d' \in D_j} A_{d'\!d} = t_{ij}.
\] 

If there are $s_{ij}$ such that $\sum_{d' \in D_j} A_{dd'} = s_{ij}$
for all $d \in D_i$, we call the partition \emph{row-stable};
similarly, for $t_{ij}$ such that $\sum_{d' \in D_j} A_{d'\!d} = t_{ij}$
for all $d \in D_i$, \emph{column-stable}. 
\eD

For symmetric $A$, column- and row-stability are equivalent
(with $t_{ij} = s_{ij}$).

Note that the row and column sums in the definition are
the $D_i$-components of $A \mathbf{d}_j$ and 
of $\mathbf{d}_j^t A = (A^t \mathbf{d}_j)^t$, respectively. 
So, for instance, row stability precisely says that 
\[
A \mathbf{d}_j 
= \sum_i s_{ij} \mathbf{d}_i \in \bigoplus_i \mathrm{span}(\mathbf{d}_i).
\]

\bL
\label{commstablelem}
Let $A \in \R^{n,n}$ commute with some symmetric 
matrix of the form $Z = XX^t \in \R^{n,n}$
for some doubly stochastic $X \in \R^{n,n}$. 
Then the partition $[n] = \dot{\bigcup}_i D_i$
of $[n]$ induced by $Z$ according to Lemma~\ref{symdecomplem} is stable for $A$. 
\eL

\prf
We use the characteristic vectors 
$\mathbf{d}_i$ of the partition sets. 
By Corollary~\ref{evalcor},
the eigenspace for eigenvalue $1$
of $Z$ is the direct sum of the 
spans of the vectors $\mathbf{d}_i$. 

Now 
$Z A \mathbf{d}_i = AZ \mathbf{d}_i = A \mathbf{d}_i$ shows that 
$A \mathbf{d}_i$ 
is an eigenvector of $Z$ 
with eigenvalue $1$, whence\[
A \mathbf{d}_i 
\in \bigoplus_i \mathrm{span}(\mathbf{d}_i).
\]

It follows that the partition $[n] = \dot{\bigcup}_i D_i$ is row-stable.

Note again that $(A \mathbf{d}_j)_d = \sum_{d' \in D_j} A_{d d'}$
and $A \mathbf{d}_j \in \bigoplus_i \mathrm{span}(\mathbf{d}_i)$
precisely means that this value $(A \mathbf{d}_j)_d$ only
depends on the partition set $D_i$ to which $d$ belongs.
I.e., $\sum_{d'\in D_j} A_{d d'} = s_{ij}$ 
for all $d \in D_i$. 

As $Z = XX^t = Z^t$, $A^t$ commutes with $Z$ if $A$ does:
$A^t Z = A^t Z^t = (ZA)^t = (AZ)^t = Z^t A^t = Z A^t$. 
The above reasoning therefore shows that the partition into the
$D_i$ is row-stable for $A^t$ as well, hence column stable for
$A$. Hence it is stable for $A$.
\eprf

NB: symmetry of $A$ is not required here. It is essential 
for deriving commutation of $A$ (and $A^t$) with $Z = XX^t$
from an equation of the form $AX = XB$, as we shall see below.
But first a corollary from the  argument just given.

\bC
\label{innersymcor}
Let $A$ commute with $Z = XX^t$ 
and $B$ commute with $Z' = X^t X$, where 
$X$ is doubly stochastic (cf.~Lemma~\ref{commstablelem}). 
Then the partitions induced by 
$Z$ and $Z'$, which are $X$-related by Lemma~\ref{newdecomplem}, 
are stable for $A$ and $B$, respectively.
\eC

\subsubsection*{Aside: boolean arithmetic}

We give a separate elementary proof of the analogue of 
Lemma~\ref{commstablelem} 
for boolean arithmetic.  
Here the definition of a 
\emph{boolean} stable partition is this natural analogue
of Definition~\ref{stablepartdef}.

\bD
\label{stablepartbooldef}
A partition $[n] = \dot{\bigcup}_{i \in I} D_i$ is \emph{boolean stable}
for $A \in \B^{n,n}$ if, in the sense of boolean arithmetic,
$\sum_{d' \in D_j} A_{dd'}$ and $\sum_{d' \in D_j} A_{d'd}$
only depend on the partition set $i$ for which $d \in D_i$. 
\eD

Note that boolean stability implies that, for the characteristic
vectors $\mathbf{d}_i$ of the partition, 
$(A \mathbf{d}_j)_d = \sum_{d' \in D_j} A_{dd'}$ is the same for all
$d \in D_i$, so that also here $A \mathbf{d}_j$ is a boolean linear 
combination of the characteristic vectors $\mathbf{d}_i$.

\bL
\label{commstableboollem}
Let $A \in \B^{n,n}$ commute, in the sense of boolean arithmetic,  
with some symmetric 
matrix of the form $Z = XX^t \in \B^{n,n}$
with entries $Z_{dd} = 1$ for all $d \in [n]$.
Then the partition $[n] = \dot{\bigcup}_i D_i$
induced by $Z$ according to Lemma~\ref{summaryboollem}
is boolean stable for $A$. 
\eL

\prf
Recall from Observation~\ref{inducedpartboolobs} that the 
characteristic vectors $\mathbf{d}_i$ of the induced partition 
behave like eigenvectors with eigenvalue $1$ for boolean arithmetic:
$Z \mathbf{d}_i = \mathbf{d}_i$. Moreover, we may assume that 
$Z_{dd'} = 1$ iff $d$ and $d'$ are in the same partition set
(after passage to $Z^{n-1}$ if necessary). Let us write 
$\brck{\ell \in D_j}$ for the boolean truth value 
of the assertion $\ell \in D_j$. Then, for $d \in D_i$,
\[
\barr{rcl} 
\sum_{d' \in D_j} A_{dd'} &=&  (A \mathbf{d}_j)_d 
=  (A Z \mathbf{d}_j)_d  
\\
&=& (Z A \mathbf{d}_j)_d = 
\sum_{k,\ell} Z_{dk} A_{k\ell} \, \brck{\ell \in D_j}
=
\sum_{k \in D_i,\ell \in D_j} A_{k\ell}
\earr
\]
does indeed not depend on $d \in D_i$, whence 
the partition is boolean row-stable. Column-stability  
again follows from similar considerations based on 
commutation of $Z = Z^t$ with $A^t$.
\eprf

\section{Fractional isomorphism}

\subsection{$\LC^2$-equivalence and linear equations} 
\label{fracisorealsec}
The \emph{adjacency matrix} of graph $\str A$ is the square matrix $A$ with
rows and columns indexed by vertices of $\mathcal A$ and entries $A_{aa'}=1$ if
$aa'$ is an edge of $\str A$ and $A_{aa'}=0$ otherwise. 
By our assumption that graphs are undirected and simple, $A$ is a 
symmetric square matrix with null diagonal. It will
be convenient to assume that our graphs always have an initial segment
$[n]$ of the positive integers as their vertex set. Then the adjacency
matrices are in $\mathbb B^{n,n}\subseteq\mathbb R^{n,n}$. Throughout this
subsection, we assume that $\str A$ and $\str B$ are graphs with
vertex set $[n]$ and with adjacency matrices $A,B$,
respectively. It will be notationally suggestive to
 denote typical indices of matrices $a,a',\ldots \in [n]$ 
when they are to be interpreted as vertices of $\str{A}$, and 
$b,b',\ldots \in [n]$ when they are to be interpreted as 
vertices of $\str{B}$.

Recall (from the discussion in the introduction) that two graphs $\str A,\str
B$ are isomorphic if, and only if, there is a permutation
matrix $X$ such that $AX=XB$. We can rewrite this as the following integer
linear program in the variables $X_{ab}$ for $a,b\in[n]$.
\begin{center}
\framebox{\begin{minipage}{\textwidth-3em-4mm}
$\ISO$
\begin{align*}
  \displaystyle\sum_{b'\in[n]}X_{ab'}&=\sum_{a'\in[n]}X_{a'b}=1,\\
  \displaystyle\sum_{a'\in[n]}A_{aa'}X_{a'b}&=\sum_{b'\in[n]}X_{ab'}B_{b'b},\\
  X_{ab}&\ge 0&\text{for all }a,b\in[n].
  \end{align*}
 \end{minipage}} 
\end{center}
Then $\str A$ and $\str B$ are isomorphic if, and
only if, $\ISO$ has an integer solution.

\begin{ddef}
  Two graphs $\str A,\str B$ are \emph{fractionally isomorphic}, $\str
  A\approx\str B$,
  if, and only if, the system $\ISO$ has a real solution.
\end{ddef}

Observe that graphs are fractionally isomorphic if, and only if, there is a doubly
stochastic matrix $X$ such that $AX=XA$.

Note that fractionally isomorphic graphs necessarily have the same 
number of vertices (this will be different for the boolean analogue, which
cannot count). 

The established theorem on fractional isomorphism, by
Tinhofer~\cite{tin86,tin91} and Ramana, Scheinerman
and Ullman from~\cite{RamanaScheinermanUllman,ScheinermanUllman},
relates fractional isomorphis to the colour refinement algorithm (`iterated
degree sequences' in~\cite{ScheinermanUllman}) introduced in Section~\ref{sec:logic}
and stable partitions (`equitable partitions' in~\cite{ScheinermanUllman}). 

A \emph{stable partition} of the vertex set of an undirected graph is a stable
partition $[n] = \dot{\bigcup}_{i \in I} D_i$ for its adjacency matrix in the
sense of Definition~\ref{stablepartdef}. Reading that definition for the
(symmetric) adjacency matrix $A$ of a graph on $[n]$, and thinking of the
partition sets $D_i$ as vertex colours, stability means that the colour of any
vertex determines the number of its neighbours in every one of the colours.
This is stability in the sense of colour refinement; it means that the colour
refinement algorithm produces the coarsest stable partition.

The characteristic %
parameters for a stable partition $[n] = \dot{\bigcup}_{i \in I} D_i$
for $A$ are 
the numbers $s_{ij}=s_{ij}^A$ such that $s_{ij}=\sum_{d' \in D_j} A_{dd'}$ for
all $d\in D_i$. (As $A$ is symmetric, the parameters $t_{ij}$ of
Definition~\ref{stablepartdef} are equal to the $s_{ij}$.) We call two stable
partitions $\dot{\bigcup}_{i \in I} D_i$ for a matrix $A$ and $\dot{\bigcup}_{i
  \in J} D_i'$ for a matrix $B$ \emph{equivalent} if $I=J$ and $|D_i|=|D_i'|$
for all $i\in I$ and 
$s_{ij}^A=s_{ij}^B$ and for all $i,j\in I$.

\begin{llem}\label{lem:c2stable}
  $\str A$ and $\str B$ are $\LC^2$-equivalent if, and only if, there
  are equivalent stable partitions $\dot{\bigcup}_{i \in I} D_i$ for $A$ and
  $\dot{\bigcup}_{i \in I} D_i'$ for $B$.
\end{llem}

\begin{proof}
  The forward direction follows from Theorem~\ref{theo:immlan}, because the
  colour refinement algorithm computes equivalent stable partitions of
  $\mathcal A$ and $\mathcal B$. 

  To establish the converse implication, we use the bijective $2$-pebble game, which
  characterises $\LC^2$-equivalence by 
  Theorem~\ref{theo:hel}. Suppose we have equivalent stable partitions $\dot{\bigcup}_{i \in I} D_i$ of $A$ and
  $\dot{\bigcup}_{i \in J} D_i'$ of $B$. Then it is a winning strategy for
  player \PII\ to maintain the following invariant for every position $p$ of
  the game: $p$ is a local isomorphism (that is, if $\dom(p)=\{a,a'\}$ then
  $a=a'$ if, and only if, $p(a)=p(a')$, and $a$ and $a'$ are adjacent in $\str A$ if, and only if, $p(a)$ and $p(a')$ are
  adjacent in $\str B$), and if $a\in\dom(p)\cap D_i$ then $p(a)\in D_i'$. It
  follows easily from the definition of stable partitions that player \PII\
  can indeed maintain this invariant.
\end{proof}

\bT[Tinhofer] 
\label{simplegraphthm}
Two graphs are $\LC^2$-equivalent if, and only if, they
are fractionally isomorphic.
\eT 

\prf
In view of Lemma~\ref{lem:c2stable}, it suffices to prove that $\str A$ and
$\str B$ have equivalent stable partitions if, and only if, they are
fractionally isomorphic.

For the forward direction, suppose that we have equivalent stable partitions
$\dot{\bigcup}_{i \in I} D_i$ for $A$ and $\dot{\bigcup}_{i \in J} D_i'$ for
$B$. 
For all $a\in D_i,b\in D_j'$ we let 
\[
X_{ab}:=\delta(i,j)/n_i,
\] 
where
$n_i:=|D_i|=|D_i'|$. (Here and elsewhere we use Kronecker's $\delta$
function defined by $\delta(i,j)=1$ if $i=j$ and $\delta(i,j)=0$
otherwise.) 
An easy calculation
shows that this defines a doubly stochastic matrix $X$ with $AX=XB$, that is,
a solution for $\ISO$.

For the converse direction, suppose that $X$ is a doubly stochastic matrix
such that $AX = XB$.
Since $A$ and $B$ are symmetric, also
$X^t A = B X^t$, and 
\[
A XX^t = XB X^t = XX^t A \quad \mbox{ and } \quad 
B X^tX = X^tA X = X^t X B,
\]
show that $A$ commutes with $Z := XX^t$ and $B$ with $Z' := X^tX$.

From Lemma~\ref{commstablelem} and Corollary~\ref{innersymcor},
the partitions 
$[n] = \dot{\bigcup}_{i \in I} D_i$ and 
$[n] = \dot{\bigcup}_{i \in I} D_i'$ that are
induced by the symmetric matrices $Z$ and $Z'$ 
are $X$-related and stable for 
$A$ and for $B$, respectively. 
We need to show that $|D_i| = |D_i'|$ and that  
the partitions also agree w.r.t.\ the parameters $s_{ij}$.

By Lemma~\ref{stochrelatedpartlem} we have $|D_i| = |D_i'|$ and
\begin{equation}
\label{parttransleqn}
\mathbf{d}_i = X \mathbf{d}_i'
\quad \mbox{ and } \quad
\mathbf{d}_i' = X^t \mathbf{d}_i,
\end{equation}
where $\vec d_i$ and $\vec d_i'$ for $i\in I$ are the characteristic vectors
of the two partitions.
Thus for all $i,j\in I$,
\[
(\mathbf{d}_i')^t B \mathbf{d}_j' 
= 
(X^t \mathbf{d}_i)^t B X^t \mathbf{d}_j
= 
\mathbf{d}_i^t X B X^t \mathbf{d}_j
=
\mathbf{d}_i^t A X X^t \mathbf{d}_j
=
\mathbf{d}_i^t A Z \mathbf{d}_j
=
\mathbf{d}_i^t A \mathbf{d}_j,
\]
where the last equality follows from the fact that $\vec d_j$ is an
eigenvector of $Z$ with eigenvalue $1$ by 
Corollary~\ref{evalcor}.

Note that $\mathbf{d}_i^t A \mathbf{d}_j$ is the number of edges of $\str A$
from $D_i$ to $D_j$. By stability of the partition, we have
$s_{ij}^A=\mathbf{d}_i^t A \mathbf{d}_j/|D_i|$ and similarly
$s_{ij}^B=(\mathbf{d}_i')^t B \mathbf{d}_j'/|D_i'|$, so that $s_{ij}^A=s_{ij}^B$.
\eprf

\subsection{$\LL^2$-equivalence and boolean linear equations}
\label{boolfracisosec}
W.r.t.\ an adjacency matrix $A \in \B^{n,n}$, 
a boolean stable partition $[n] = \dot{\bigcup}_{i \in I} D_i$ 
has as parameters just the boolean values 
\[
\iota_{ij}^A  = \left\{
\barr{ll}
0 & \mbox{if } A_{D_iD_j} = 0, 
\\
\hnt
1 & \mbox{else.} 
\earr
\right.
\]
Boolean (row-)stability of the partition for $A$ implies that
$\iota_{ij}^A = 1$ if, and only if, for each individual $d \in D_i$ 
there is at least one $d' \in D_j$ such that $A_{dd'} = 1$, and
similarly for column stability. 

To capture the situation of $2$-pebble game equivalence, though, 
we now need to work with similar partitions that are stable both
w.r.t.\ $A$ and w.r.t.\ to the adjacency matrix $A^c$ of the
complement of the graph with adjacency matrix $A$. Here the complement of
a graph $\str A$ is the graph $\str A^c$ with the same vertex set as
$\str A$ obtained by replacing edges by non-edges
and vice versa. Hence $A^c_{aa'}=1$ if $A_{aa'}=0$ and $a\neq a'$,
and $A^c_{aa'}=0$ otherwise. 
While a partition in the sense 
of real arithmetic is stable for $A$ if, and only if, it is stable for
$A^c$, this is no longer the case for boolean arithmetic.
Let us call a partition that is boolean stable for both $A$ and $A^c$,
\emph{boolean bi-stable} for $A$.

Then the following captures 
the situation of two graphs that are $2$-pebble game
equivalent. We note that $2$-pebble equivalence is a very coarse notion
of equivalence, if we look at just simple undirected graphs -- but the
concepts explored here 
form the basis for the analysis of $k$-pebble 
equivalence, which is non-trivial even for simple undirected graphs. 

$\LL^2$-equivalence of two graphs does not imply that the graphs have
the same size. In the following, we always assume that $\mathcal
A,\mathcal B$ are graphs with vertex sets $[m],[n]$ respectively and
that $A\in\mathbb B^{m,m}$ and $b\in\mathbb B^{n,n}$ are their
adjacency matrices. 
We call two bi-stable
partitions $[m]=\dot{\bigcup}_{i \in I} D_i$ for $A$ (and $A^c$) and $[n]=\dot{\bigcup}_{i
  \in J} D_i'$ for $B$ (and $B^c$) \emph{b-equivalent} if $I=J$ and 
$\iota_{ij}^A=\iota_{ij}^B$ and $\iota_{ij}^{A^c}=\iota_{ij}^{B^c}$
for all $i,j\in I$. Note that b-equivalence does not imply that
$|D_i|=|D_i'|$.

\begin{llem}\label{lem:l2stable}
  $\str A$ and $\str B$ are $\LL^2$-equivalent if, and only if, there
  are b-equivalent bi-stable partitions $[m]=\dot{\bigcup}_{i \in I}
  D_i$ for $A$ and $[n]=\dot{\bigcup}_{i \in J} D_i'$ for $B$.
\end{llem}

\begin{proof}
  The proof is analogous to the proof of Lemma~\ref{lem:l2stable}. For
  the backward direction, we need bi-stability to guarantee that player
  \PII\ can maintain positions $p$ that preserve adjacency,
  non-adjacency, and (in)equality. Stability alone would only enable
  her to maintain adjacency and equality.  
\end{proof}

\bD
\label{fracisobooldef}
$\str{A}$ and $\str{B}$ are \emph{boolean
  isomorphic},
$\str{A} \fracsimeqbool \str{B}$, if 
there is some boolean matrix $X$ without null rows or columns such that 
$A X = X B$ and $A^c X = X B^c$. 
\eD

That is, $\str{A}$ and $\str{B}$ are boolean isomorphic if they
satisfy the following system of linear equation over the boolean
semiring. 

\begin{center}
\framebox{\begin{minipage}{\textwidth-3em-4mm}
$\BISO$
\begin{align*}
  \displaystyle\sum_{b'\in[n]}X_{ab'}&=\sum_{a'\in[n]}X_{a'b}=1,\\
  \displaystyle\sum_{a'\in[n]}A_{aa'}X_{a'b}&=\sum_{b'\in[n]}X_{ab'}B_{b'b},\\
  \displaystyle\sum_{a'\in[n]}A^c_{aa'}X_{a'b}&=\sum_{b'\in[n]}X_{ab'}B^c_{b'b} &\text{for all }a,b\in[n].
  \end{align*}
 \end{minipage}} 
\end{center}

\bT
\label{simplegraphboolthm}
Two graphs are $\LL^2$-equivalent if, and only if, they are boolean
isomorphic. 
\eT

\prf For the forward direction, suppose that $A \equiv_\LL^2 B$, and
let $[m] = \dot{\bigcup}_{1 \leq i \leq s} D_i$ and $[n] =
\dot{\bigcup}_{1 \leq i \leq s} D_i'$ be the similar boolean
bi-stable partitions.  For all $a\in D_i, b \in D_j'$ we let
$X_{ab}:=\delta(i,j)$. This defines a
boolean matrix $X\in \B^{m,n}$ without null rows or columns.  One
checks that $AX = XB$, in boolean arithmetic: for $a \in D_i$ and
$b \in D_j'$, and for the characteristic vectors $\mathbf{d}_i$
and $\mathbf{d}_j'$ for the partitions,
\begin{align*}
(AX)_{ab} &=
\sum_{k} A_{ak}X_{kb} 
= 
(A \mathbf{d}_j)_a
= \iota_{ij}^A\\
&= \iota_{ij}^B 
= 
((\mathbf{d}_i')^t B)_{b}
= 
\sum_{k} X_{ak} B_{kb} = (XB)_{ab}.
\end{align*}
The argument for $A^c X = X B^c$ is completely analogous.

For the converse, suppose that $A \fracsimeqbool B$, and let
$X$ be a boolean
matrix without null rows or columns such that $AX = XB$ and $A^c X = X B^c$. 
Since $A$ and $B$ are symmetric, also
$X^t A = B X^t$  $X^t A^c  = B^c X^t$, and 
\[
A XX^t = XB X^t = XX^t A \quad \mbox{ and } \quad 
B X^tX = X^tA X = X^t X B,
\]
together with the analogues for the complements,
show that both $A$ and $A^c$ commute with $Z := XX^t$ and both $B$ and $B^c$ 
commute with $Z' := X^tX$. Moreover, the matrices $Z$ and $Z'$ 
have entries $1$ on the diagonal.

From Lemma~\ref{commstableboollem} and the straightforward analogue of
Corollary~\ref{innersymcor},
the partitions 
$[m] = \dot{\bigcup}_{i \in I} D_i$ and 
$[n] = \dot{\bigcup}_{i \in I} D_i'$ 
induced by the symmetric matrices $Z$ and $Z'$ 
are $X$-related and boolean bi-stable for 
$A$ and for $B$, respectively. 
We need to show that these partitions also agree 
w.r.t.\ the characteristic $\iota_{ij}$.
By Lemma~\ref{summaryboollem}, 
the characteristic vectors $\mathbf{d}_i'$
and $\mathbf{d}_i'$
of the partitions are related by
$\mathbf{d}_i = X \mathbf{d}_i'$ and 
$\mathbf{d}_i' = X^t \mathbf{d}_i$ in the sense of boolean arithmetic.

Since $AX=XB$ and as the $\mathbf{d}_j$ are boolean eigenvectors of 
$Z = XX^t$ with eigenvalue $1$ by Observation~\ref{inducedpartboolobs},
\[
\barr{rcl}
\iota_{ij}^B &=& 
(\mathbf{d}_i')^t B \mathbf{d}_j' 
= 
(X^t \mathbf{d}_i)^t B X^t \mathbf{d}_j
= 
\mathbf{d}_i^t X B X^t \mathbf{d}_j 
\\
&=&
\mathbf{d}_i^t A X X^t \mathbf{d}_j
=
\mathbf{d}_i^t A Z \mathbf{d}_j
=
\mathbf{d}_i^t A \mathbf{d}_j
\;=\; \iota^A_{ij}.
\earr
\]

The argument for $\iota_{ij}^{B^c} = \iota_{ij}^{A^c}$ is strictly analogous.
\eprf

\subsection{Good solutions for fractional isomorphism}
\label{goodsolisosec}
We conclude the analysis of fractional isomorphism
with an account that will be useful towards generalisations 
in higher dimensions. 
Fractional isomorphism as well as its boolean analogue
are based on solutions of linear matrix equations
\[
\comp[A,B]: \quad
AX = XB, %
\]
which express a \emph{compatibility} condition.
These equations may be read in the sense of real arithmetic or
in the sense of boolean arithmetic. 
For fractional isomorphism between graphs with adjacency matrices $A$ and $B$
we are interested in doubly stochastic real solutions of this single equation; 
while the boolean analogue involves simultaneous boolean solutions
without null rows or columns of the pair of equations 
$\comp[A,B]$ and $\comp[A^c,B^c]$. 
We isolate the properties of \emph{good solutions} 
of equations of this type as follows;
compare Definition~\ref{goodsymmatrixdef} for \emph{good matrices}.  

\bD
\label{goodsoldef}
For symmetric matrices $A \in \B^{m,m}$ and $B \in \B^{n,n}$,
a solution $X$ to the linear matrix equation
$\comp[A,B]$ in the sense of real arithmetic 
(boolean arithmetic) is a \emph{good solution} if
$X$ is a doubly stochastic real matrix (a boolean matrix without null
rows or columns) and the matrices $Z=XX^t$ and $Z'=X^tX$ induce 
$X$-related partitions $[m] = \dot{\bigcup}_i D_i$ 
and $[n] = \dot{\bigcup}_i D_i'$ 
such that 
\bre
\item
these partitions are equivalent (boolean equivalent) 
and stable (boolean stable) w.r.t.\ $A$ and $B$, respectively;
\item
$X_{D_iD_i'} > 0$ for all $i$;
\item
$X_{D_iD_j'}$ = 0 for $i \not= j$.
\ere
\eD

Recall from Corollary~\ref{goodlinkcor}, and from Lemma~\ref{summaryboollem} 
for the boolean case, 
that good solutions are always obtained from 
given solutions $X$ through passage to $X' := Z^\ell X$ for
sufficiently large $\ell$ and 
$Z := XX^t$, which is a symmetric matrix with strictly positive diagonal. 
Clearly this transformation of solutions into good solutions works
for simultaneous solutions to several equations
of type %
$\comp[\cdot,\cdot]$. For simultaneous (boolean) solutions w.r.t.\ 
$A/B$ and $A^c/B^c$ we obtain simultaneous good solutions 
w.r.t.\ $A/B$ and $A^c/B^c$, which are therefore (boolean) 
bi-stable w.r.t.\ $A$ and $B$.
We thus find the following summary account for the results 
presented in Sections~\ref{fracisorealsec} and~\ref{boolfracisosec} above.

\bL
\label{goodsolfraclem}
For graphs $\str{A}$ and $\str{B}$ with 
adjacency matrices $A \in \B^{m,m}, B \in \B^{n,n}$, 
let $X$ be a good real solution to
$\comp[A,B]$ (i.e., a real solution to $\ISO$)
with induced partitions $[m] = \dot{\bigcup}_i D_i$
and $[n] = \dot{\bigcup}_i D_i'$. 
Then player \PII\ has a strategy 
in the corresponding bijective $2$-pebble game
to maintain pebble configurations 
$p \colon a_1a_2 \mapsto b_1 b_2$ for which 
\bre
\item
$p$ is a local isomorphism:
$a_1=a_2$ iff $b_1=b_2$ and $A_{a_1a_2} =1$ iff 
$B_{b_1b_2} =1$; 
\item
$p$ respects the induced partitions: 
corresponding pebbles are from matching partition sets, i.e., 
$X_{a_1b_1}, X_{a_2b_2} >0$.
\ere
Similarly, if $X$ is a simultaneous good boolean solution to
$\comp[A,B]$ and $\comp[A^c,B^c]$,
then the above condition can be maintained in the plain $2$-pebble
game (without counting).
\eL

\prf
We discuss the real version in relation to the bijective $2$-pebble
game. Consider a position $p \colon a_1a_2 \mapsto b_1b_2$ as
described by conditions~(i) and~(ii), and assume that 
player \PI\ starts a round involving relocation of the second pebble
pair, so that just $ab := a_1 b_1$ of the current position needs to be
respected. By the rules of the game, player 
\PII\ needs to propose a bijection $\rho \colon [n]
\rightarrow [n]$ between the vertex sets of $\str{A}$ and
$\str{B}$ (recall that $n=m$ as $X$ is doubly stochastic)
with $\rho(a) = b$ and 
such that $A_{aa'} = 1$ iff $B_{b\rho(a')} = 1$ for all $a'$;  
in order to maintain conditions~(i) and~(ii) we additionally
want $X_{a'\rho(a')} > 0$ for all $a'$. The existence of such a 
bijection follows from the $X$-equivalence of the partitions 
$[n] = \dot{\bigcup}_i D_i$ and $[n] = \dot{\bigcup}_i D_i'$. By 
pre-condition~(ii), $a =a_1 \in D_i$ and $b=b_1 \in D_i'$ 
for the same index $i$, whence the numbers of $A/B$-adjacent
vertices in $D_j$ in $\str{A}$/in $D_j'$ in $\str{B}$ 
agree for each index $j$, and the desired bijection can be pieced
together from bijections between $D_j$ and $D_j'$ that respect
adjacency with $a$ and $b$.
\eprf

\bC
\label{goodsolutionfracisocor}
\bae
\item
The following are equivalent for graphs
$\str{A}$ and $\str{B}$ with adjacency matrices $A \in \B^{m,m}$,
$B \in \B^{n,n}$:
\bre
\item $\str{A}$ and $\str{B}$ are fractionally isomorphic, i.e.\ the
  system $\ISO$ has a real solution;
\item the matrix equation %
$\comp[A,B]$ admits a good solution $X$;
\item $\PII$ has a winning strategy in the bijective 
$2$-pebble on $\str{A}$ and $\str{B}$. 
\ere
\item
Similarly the following are equivalent:
\bre
\item $\str{A}$ and $\str{B}$ are boolean fractionally isomorphic;
i.e.\ the system $\BISO$ has a solution in the sense of boolean arithmetic;
\item
the matrix equations %
$\comp[A,B]$ and $\comp[A^c,B^c]$ 
possess a simultaneous good boolean solution $X$;
\item $\PII$ has a winning strategy in the plain
$2$-pebble on $\str{A}$ and $\str{B}$. 
\ere
\eae
\eC

\section{Relaxations in the style of Sherali--Adams}  
In this section we refine the connection between the Sherali--Adams
hierarchy of LP relaxations of the integer linear program $\ISO$
to equivalence in the finite variable counting logics or the higher-dimensional Lehman--Weisfeiler equivalence. 

\medskip
NB: our parameter $k \geq 2$ is the number of pebbles, or the  
variables available in the $k$-variable logics $\LC^k$ or $\LL^k$.

\medskip
As before, $\str A$ and $\str B$ are graphs with vertex sets
$[m]$ and $[n]$, respectively, and $A$ and $B$ are their adjacency
matrices.
We denote typical elements and tuples of elements 
from $\str{A}$ and $\str{B}$
as $\abar = (a_1,\ldots, a_r)$ or 
$\bbar = (b_1,\ldots, b_r)$, for $0 \leq r \leq k$; 
correspondingly, we typically denote entries of the adjacency 
matrices as, e.g., $A_{aa'}$. This device will help in an
intuitive consistency check also in matrix compositions like
$AX$ with entries $(AX)_{ab}$ if $A$ is an $m\!\times\!m$ matrix 
over $[m]$ and 
$X$, as an $m\!\times\!n$ matrix, 
relates $[m]$ and $[n]$
through entries $X_{ab}$: $(AX)_{ab} = \sum_{a'} A_{aa'}X_{a'b}$ 
(which rightly suggests paths of length two in a suitable 
composition of graphs
$\str{A}$ and $G(X)$).

\paragraph*{Types.}
Let $\etp(\abar)$ denote the equality type of tuple $\abar$ 
in $\str{A}$,
$\atp(\abar)$ its quan\-ti\-fier-free type, and 
$\tp(\abar)$ its complete type in the logic $\LC^k$, that is, the set
of all $\LC^k$-formulae $\phi(\vec x)$ such that $\str A$ satisfies
$\phi(\vec a)$. 
Note that $\abar \mapsto \bbar$ constitutes a local bijection
if, and only if, $\etp(\abar) = \etp(\bbar)$, and a 
local isomorphism if, and only
if, $\atp(\abar) = \atp(\bbar)$. 

It will sometimes be useful to view some of the elements whose type we
consider as ``parameters'' and only the remaining as
``variables''. Formally, we define the \emph{type of
  $a$ w.r.t.\ to the
parameters $\abar$} simply to be the type of $\abar a$, that is,
$\tp_\abar(a):=\tp(\abar a)$. The distinction between plain types and types with
parameters is one of semantic intention rather than
syntactic. It is suggestive when it comes to
counting realisations. For example, we let 
\[
\#_\xbar^\str{A} (\tp(\xbar) \!=\! \tp(\abar)) 
\]
denote the 
number of tuples $\abar'$ in $\str{A}$ that realise the type of
$\abar$, i.e., those $\abar'$ for which $\str{A},\abar' \equiv_\LC^k
\str{A},\abar$. If the structure in which realisations are counted is
obvious or does not matter because of $\LC^k$-equivalence, we drop the 
superscript and write e.g.\ just $\#_\xbar$ instead of  $\#_\xbar^\str{A}$
or $\#_\xbar^\str{B}$.

The number of realisations of the $1$-type of $a$ 
over parameters $\abar$ (in $\str{A}$) is denoted by
\[
\#_{x}(\tp_\abar(x)  \!=\!\tp_\abar(a))
\quad\bigl(= \#_{x}(\tp(\abar x)  \!=\! \tp(\abar a))\bigr). 
\]

Regarding the counting of realisations we note that generally  
\begin{equation}
\label{quotienteqn}
\#_{\xbar x} (\tp(\xbar x)  \!=\! \tp(\abar a))
\quad=\quad 
\#_\xbar (\tp(\xbar)  \!=\! \tp(\abar))
\cdot  \#_x (\tp_\abar(x)  \!=\! \tp_\abar(a)). 
\end{equation}

\paragraph*{The equations.}

In the following we use variables $X_p$ indexed by subsets 
$p \subset[m]\times[n]$ of size up to $k-1$;
we may think of such $p$ as being specified by two tuples 
$\abar$ and $\bbar$ of length $|p|$ that enumerate the 
first and second components of the pairs in $p$ in any coherent 
order. In this sense we write $p = \abar\bbar$. 
As remarked before, $p$ is a local bijection between $\str A$
and 
$\str B$ iff
$\etp(\abar) = \etp(\bbar)$, and a local isomorphism 
iff $\atp(\abar) = \atp(\bbar)$ (and neither of these conditions
depends on the chosen enumeration of tuples in $p$, which gives rise to
the order of components in both $\abar$ and $\bbar$).

The augmentation 
of $p \subset [m]\times[n]$  
by some pair $ab \in [m]\times[n]$  
is simply denoted $p \,\widehat{\ }\,ab$.
It is crucial that the notation $p \,\widehat{\ }\,ab$ does \emph{not}
refer to a \emph{tuple} of pairs but to a \emph{set} of pairs, in
which the pair $(a,b)$ is not distinguished. In particular, 
if the pair $ab$ is in $p$, then $p\,\widehat{\ }\,ab = p$. Correspondingly, 
$p \setminus ab$ stands for the set of pairs $a'b' \in p$ that are
distinct from $ab$.

\medskip
For further reference we isolate and name equation types as follows.
For given $n,m \geq 1$ and matrices $A \in \B^{n,n}$ and $B \in \B^{m,m}$: 
\[
\barr{l}
\,X_\emptyset = 1
\hfill \CONT0
\\
\\
\left.%
\barr{@{}l@{}}
X_p = \sum_{b'} X_{p\,\widehat{\ }\, ab'} = \sum_{a'} X_{p\,\widehat{\ }\, a'b}
\\
\hnt
\mbox{for } |p| = \ell-1, a \in [m], b \in [n]
\earr
\qquad\;
\right\} 
\hfill \CONT{\ell}
\\
\\
\left.%
\barr{@{}l@{}}
\sum_{a'} A_{aa'}  X_{p\,\widehat{\ }\, a'b} = 
\sum_{b'} X_{p\,\widehat{\ }\, ab'} B_{b'b} 
\\
\hnt
\mbox{for } |p| = \ell-1, a \in [m], b \in [n]
\earr
\quad\;\,
\right\} 
\qquad \COMP{\ell}
\earr
\]

\bigskip
Here \emph{level $\ell$} refers to $\ell$ as the size of the pairings 
$\abar \bbar$ in the typical variables $X_{\abar \bbar}$ involved; 
note that the size of $p$ mentioned in $X_{p\,\widehat{\ }\,ab}$ 
therefore remains one below this $\ell$. In the generic formats 
for  $\CONT{\ell}$ and $\COMP{\ell}$
above, we assume $\ell \geq 1$. Note that the combination of
$\CONT{0}$, $\CONT{1}$ and $\COMP{1}$ 
precisely corresponds to the equations for fractional 
isomorphism; in particular the level~$1$ equation $\COMP{1}$ is the same 
equation that was denoted 
$\comp[A,B]$ with specific reference to the constituents $A/B$ 
in Section~\ref{goodsolisosec}.

If we think of the matrix entries 
$(X_{\abar\bbar\,\widehat{\ }\,ab})_{a \in [n], b \in [m]}$ as specifying 
extensions of $\abar \mapsto \bbar$ in the form of a distribution
on possible pairings $a \mapsto b$, then equations~$\CONT{\ell}$
may be seen as \emph{continuity conditions}, while 
equations~$\COMP{\ell}$ specify \emph{compatibility conditions}
with the edge relations encoded in $A$ and $B$.
Variants of the compatibility conditions can be expressed for matrices other
than the adjacency matrices $A$ and $B$ that we primarily think
of. We saw one such variation in the discussion on boolean 
isomorphisms above, where $\COMP1$
was postulated for both $A,B$ and $A^c,B^c$. Further 
variants will play a role in Section~\ref{soltogamesec}.

\subsection{Sherali--Adams of level $k-1$}
\label{sec:sa}

For $k \geq 2$, 
the \emph{level $k-1$ Sherali--Adams relaxation} of the integer linear
program $\ISO$ consists of 
the collection of the equations 
$\CONT{\ell}$ and $\COMP{\ell}$ 
for $\ell < k$. 
Note that $\ISO[1]$ (i.e., level~$1$ or $\ISO[k-1]$ for $k=2$) 
is fractional isomorphism.

\begin{center}
\framebox{\begin{minipage}{\textwidth-3em-4mm}
$\ISO[k-1]$
\[
\barr{l}
\left.
\barr{@{}l@{}} 
X_\emptyset = 1
\quad\mbox{and}
\\
\hnt
X_p  = 
\sum_{b'} X_{p\,\widehat{\ }\, ab'} %
= \sum_{a'} X_{p\,\widehat{\ }\, a'b}
\\
\hnt
\mbox{for } |p| < k-1, a \in [m], b \in [n]
\earr
\qquad
\right\} \quad
\CONT{\ell} \mbox{ for } \ell < k 
\\
\\
\left.
\barr{@{}l@{}} 
\sum_{a'} A_{aa'}  X_{p\,\widehat{\ }\, a'b} = 
\sum_{b'} X_{p\,\widehat{\ }\, ab'} B_{b'b}
\\
\hnt
\mbox{for } |p| < k-1,  a \in [m], b \in [n]
\earr
\quad\;
\right\} \quad
\COMP{\ell} \mbox{ for } \ell < k
\\
\\
X_p\ge 0\mbox{ for }|p|\le k-1 
\earr
\]
\end{minipage}}
\end{center}

\subsubsection{From $\LC^k$-equivalence to solutions}

Assume that $\str{A} \equiv_\LC^k \str{B}$. This implies that 
$\str{A}$ and $\str{B}$ realise exactly the same
types of $r$-tuples for $r \leq k$, and with the 
same number of realisations: 
\begin{equation}
\#_\xbar^\str{A} (\tp(\xbar)  \!=\! \tp(\abar))
= 
\#_\xbar^\str{B} (\tp(\xbar)  \!=\! \tp(\abar))
\end{equation}
and similarly for all types $\tp(\bbar)$ of $r$-tuples in $\str{B}$
for $r \leq k$. In particular $m=|\str{A}| = |\str{B}|=n$
so that both structures have domain $[n]$. 

If $\tp(\abar) = \tp(\bbar)$, where 
$\abar$ and $\bbar$ are $r$-tuples for $r \leq k-1$, then, for 
any $a \in [n]$, there is a $\hat{b} \in [n]$ such that  
$\tp(\bbar \hat{b}) = \tp(\abar a)$; and 
for any such choice of $\hat{b}$ we find 
(cf.~equation~(\ref{quotienteqn})):
\begin{equation}
\label{quotienteqntwo}
\barr{r@{\;=\;}l}
\#_x^\str{A}(\tp_\abar(x)  \!=\! \tp_\abar(a))  & 
\#_{\xbar x}^\str{A}(\tp (\xbar x)  \!=\! \tp(\abar a))
\bigm/
\#_\xbar^\str{A} (\tp(\xbar)  \!=\! \tp(\abar)) 
\\
\hnt
&
\#_{\xbar x}^\str{B}(\tp(\xbar x)  \!=\! \tp(\abar a))
\bigm/
\#_\xbar^\str{B} (\tp(\xbar)  \!=\! \tp(\abar)) 
\\
\hnt
&
\#_{\xbar x}^\str{B}(\tp(\xbar x)  \!=\! \tp(\bbar \hat{b}))
\bigm/
\#_\xbar^\str{B} (\tp(\xbar)  \!=\! \tp(\bbar)) 
\\
\hnt
& 
\#_x^\str{B} (\tp_\bbar(x)  \!=\! \tp_\bbar(\hat{b})).
\earr
\end{equation}

Similarly, for $r$-tuples $\abar$ and $\bbar$ 
such that $\tp(\abar) = \tp(\bbar)$, where $r \leq k-2$ (!), 
and for any 
$a$ and $b$, 
there are $\hat{a}$ and $\hat{b}$ 
such that 
$\tp(\bbar \hat{b} b) = \tp(\abar a \hat{a})$
and 
\begin{equation}
\label{edgecount}
\#_{xy}^\str{A} (\tp_\abar(xy)  \!=\! \tp_\abar(a\hat{a})) =
\#_{xy}^\str{B} (\tp_\bbar(xy)  \!=\! \tp_\bbar(\hat{b}b)).
\end{equation}

For the desired solution put
\medskip
\begin{equation}
\label{soldef}
\barr{l}
X_\emptyset := 1, 
\\
X_p := \delta(\tp(\abar),\tp(\bbar)) \bigm/ \#_\xbar 
(\tp(\xbar)  \!=\! \tp(\abar)), 
\\
\nt\hfill
\mbox{ for $p = \abar\bbar$, $0 < |p| <k$.}
\earr
\end{equation}

For the denominator note that
$\#_\xbar (\tp(\xbar)  \!=\! \tp(\abar)) 
= 
\#_\xbar (\tp(\xbar)  \!=\! \tp(\bbar))$ whenever 
$\tp(\abar) = \tp(\bbar)$.
Clearly $X_p \geq 0$. 
Note that $X_p \not= 0$ implies 
$\tp(\abar) =\tp(\bbar)$, which implies that 
$\atp(\abar) =\atp(\bbar)$ whereby $\abar \mapsto \bbar$ is a 
local isomorphism.

We check that the given assignment to the variables $X_p$ 
satisfies all instances of the
equations~$\CONT{\ell}$ and~$\COMP{\ell}$ 
in Sherali--Adams of level $k-1$. In fact, 
we shall prove that
the $X_p$ as specified by~(\ref{soldef}),
satisfy all instances of the  continuity 
equations~$\CONT{\ell}$ of levels $\ell \leq k$ (!) 
and all instances of the compatibility 
equations~$\COMP{\ell}$
of levels~$\ell < k$, while level $k-1$ Sherali--Adams,
i.e.~$\ISO(k-1)$,
just requires both equation types for levels $\ell < k$.
This combination of equations, which seem to be of `mixed' or
intermediate levels in relation to the established Sherali--Adams
hierarchy, will be important in our analysis.
We denote it as $\ISO(k-1/2)$:

\begin{center}
\framebox{\begin{minipage}{\textwidth-3em-4mm}
$\ISO[k-1/2]$
\[
\barr{l}
\left.
\barr{@{}l@{}} 
X_\emptyset = 1
\quad\mbox{and}
\\
\hnt
X_p  = 
\sum_{b'} X_{p\,\widehat{\ }\, ab'} %
= \sum_{a'} X_{p\,\widehat{\ }\, a'b}
\\
\hnt
\mbox{for } |p| < k, a \in [n], b \in [m]
\earr
\qquad
\right\} \quad
\CONT{\ell} \mbox{ for } \ell \leq k 
\\
\\
\left.
\barr{@{}l@{}} 
\sum_{a'} A_{aa'}  X_{p\,\widehat{\ }\, a'b} = 
\sum_{b'} X_{p\,\widehat{\ }\, ab'} B_{b'b}
\\
\hnt
\mbox{for } |p| < k-1,  a \in [n], b \in [m]
\earr
\quad\,
\right\} \quad
\COMP{\ell} \mbox{ for } \ell < k 
\\
\\
X_p\ge 0\mbox{ for }|p|\le k
\earr
\]
\end{minipage}}
\end{center}

\bL
\label{mixedlevellem}
If $\str{A} \equiv_\LC^k \str{B}$, then $\ISO(k-1/2)$
admits a solution.
\eL

\prf 
We show that the solution $(X_p)$ proposed as~\ref{soldef}
satisfies all relevant instances of $\CONT{\ell}$ and $\COMP{\ell}$.
 
Consider an instance of $\CONT{\ell}$ 
for the solution proposed in 
of level $\ell \leq k$, i.e., for $|p| < k$, 
with $p = \abar \bbar$, $a \in [n]$. 
If $\tp(\abar) \not= \tp(\bbar)$, then both sides of the equation 
are zero. In case $\tp(\abar) = \tp(\bbar)$, let $\hat{b} \in [n]$ be such
that $\tp(\bbar \hat{b}) = \tp (\abar a)$ (such $\hat{b}$ exist as 
$\tp(\abar) = \tp(\bbar)$ and since $|p| < k$).  
Then
\[
\barr{r@{\;=\;}l}
\sum_{b'} X_{p\,\widehat{\ }\, ab'} & 
\sum_{b'} \; \delta(\tp(\abar a),\tp(\bbar b'))
\bigm/
\#_{\xbar x}(\tp(\xbar x)  \!=\! \tp(\abar a))
\\
\vhnt
&
\#_{x} (\tp_\bbar(x)  \!=\! \tp_\abar (a)) 
\bigm/
\#_{\xbar x}(\tp(\xbar x)  \!=\! \tp(\abar a))
\\
\vhnt
&
\#_{x} (\tp_\bbar(x)  \!=\! \tp_\bbar (\hat{b})) 
\bigm/
\#_{\xbar x}(\tp(\xbar x)  \!=\! \tp(\bbar \hat{b}))
\\
\vhnt
& 
\#_\xbar (\tp(\xbar)  \!=\! \tp(\bbar))^{-1} =  X_p,
\earr
\]
where the crucial equality leading to the last line is from
equation~(\ref{quotienteqn}). 

Consider now an instance of equation~$\COMP{\ell}$
of level $\ell < k$, i.e., with $|p| < k - 1$, 
with $p = \abar \bbar$, $a \in [n]$, $b \in [n]$. Again, the case 
of $\tp(\abar) \not= \tp(\bbar)$ is trivial.
So we are left with the case of $p = \abar \bbar$
with $\tp(\abar) = \tp(\bbar)$ and $|p| \leq k-2$.
These imply that there are $\hat{a}$ and 
$\hat{b}$ such that $\tp(\abar a \hat{a}) = \tp (\bbar \hat{b} b)$.
Then
\begin{equation}
\label{evallhs}
\barr{rl}
&
\sum_{a'} A_{aa'}  X_{p\,\widehat{\ }\, a'b} 
\\
=&
\vhnt
\sum_{a'} A_{aa'} \, \delta(\tp(\abar a'),\tp(\bbar b)) 
\bigm/ \#_{\xbar x}(\tp(\xbar x)  \!=\! \tp(\bbar b))
\\
=&
\vvhnt
\frac{\displaystyle
\#_{y}
\bigl( \mathrm{edge}(ay) \wedge 
\tp_\abar(y) \!=\!   \tp_\abar(\hat{a}) \bigr)}%
{\displaystyle 
\#_{\xbar x}(\tp(\xbar x)  \!=\! \tp(\bbar b))}
\\
=&
\vvhnt
\frac{\displaystyle
\#_{xy}
\bigl( \mathrm{edge}(xy) \wedge \tp_\abar(x)  \!=\! \tp_\abar(a) \wedge 
\tp_\abar(y) = \tp_\abar(\hat{a}) \bigr)}%
{\displaystyle 
\#_{\xbar x}(\tp(\xbar x)  \!=\! \tp(\bbar b))
\cdot \#_x(\tp_\abar(x)  \!=\! \tp_\abar(a))}
\\
=&
\vvhnt
\frac{\displaystyle
\#_{xy}
\bigl( \mathrm{edge}(xy)  \wedge \tp_\abar(x)  \!=\! \tp_\abar(a) \wedge 
\tp_\abar(y)  \!=\! \tp_\abar(\hat{a}) \bigr)}%
{\displaystyle
\#_{\xbar}(\tp(\xbar)  \!=\! \tp(\bbar))
\cdot 
\#_{x}(\tp_\bbar(x)  \!=\! \tp_\bbar(b))
\cdot 
\#_x(\tp_\abar(x)  \!=\! \tp_\abar(a))}\;,
\earr
\end{equation}
where we use instances of equation~(\ref{quotienteqn}), and, 
in the passage from the third to the fourth line, 
artificially count over all realisations 
of $\tp_\abar(a)$ instead of just the fixed parameter $a$,  
and compensate for that in the denominator.

The counting term in the enumerator of this expression,
\[
\#_{xy}
\bigl( \mathrm{edge}(xy) \wedge \tp_\abar(x)  \!=\! \tp_\abar(a)\wedge 
\tp_\abar(y)  \!=\! \tp_\abar(\hat{a})\bigr),
\]
is the sum of the number of realisations of 
all those types $(\tp_\abar(a'',a'))$ 
that simultaneously extend 
$\tp_\abar(a)$, $\tp_\abar(\hat{a})$ and contain 
the formula $\mathrm{edge}(xy)$. 
Each one of these types has exactly the same number
of realisations in $\str{A}$ as the corresponding type 
that simultaneously extends  
$\tp_\bbar(\hat{b})$, $\tp_\bbar(b)$ and contains 
the formula $\mathrm{edge}(xy)$. By symmetry of the graphs under 
consideration, $\mathrm{edge}(xy) $ is equivalent with
$\mathrm{edge}(yx)$ 
and what we obtained in~(\ref{evallhs}) coincides 
with the corresponding evaluation 
of the right-hand side of this instance 
of equation~$\COMP{\ell}$ as desired:

\begin{equation}
\label{evalrhs}
\barr{rl}
&\sum_{b'} B_{b'b}  X_{p\,\widehat{\ }\, ab'} 
\\
\vhnt
=& 
\sum_{b'} B_{b'b} \, \delta(\tp(\abar a),\tp(\bbar b')) 
\bigm/ \#_{\xbar x}(\tp(\xbar x)  \!=\! \tp(\abar a))
\\
\vhnt
=&
\frac{\displaystyle
\#_{xy}
\bigl( \mathrm{edge}(yx)  \wedge \tp_\bbar(x)  \!=\! \tp_\bbar(b) \wedge 
\tp_\abar(y)  \!=\! \tp_\abar(a) \bigr)}%
{\displaystyle
\#_{\xbar x}(\tp(\xbar x)  \!=\! \tp(\abar a))
\cdot \#_x(\tp_\bbar(x)  \!=\! \tp_\bbar(b))}
\\
\vhnt
=&
\frac{\displaystyle
\#_{xy}
\bigl( \mathrm{edge}(yx)  \wedge \tp_\bbar(x)  \!=\! \tp_\bbar(b) \wedge 
\tp_\bbar(y)  \!=\! \tp_\bbar(\hat{b}) \bigr)}%
{\displaystyle
\#_{\xbar}(\tp(\xbar)  \!=\! \tp(\abar))
\cdot 
\#_{x}(\tp_\abar(x)  \!=\! \tp_\abar(a))
\cdot 
\#_x(\tp_\bbar(x)  \!=\! \tp_\bbar(b))}\;.
\earr
\end{equation}
\eprf

We shall see in Theorem~\ref{SAkminus1/2prop} 
that $\equiv_\LC^k$ precisely corresponds to 
$\ISO[k-1/2]$.

\subsubsection{From solutions to pebble game equivalence}
\label{soltogamesec}
In the following we discuss what it means that some admissible real or rational 
non-negative assignment to the variables $X_p$ for all $|p| < k$ satisfies the 
equations~$\CONT{\ell}$ and~$\COMP{\ell}$ for $\ell < k$, i.e., $\ISO[k-1]$.

\bD
\label{supportdef}
A matrix $X = (X_{p})_{|p| \leq k}$, with non-negative entries 
indexed by \emph{sets} $p \subset [m] \times [n]$ of size up to $k$
is \emph{supported by local bijections} if $X_p > 0$ only for 
$p$ that are the graphs of local bijections between $[n]$ and $[m]$.%
We say that $X$ is \emph{supported by local isomorphisms} (w.r.t.\
to edge relations $A$ and $B$ of graphs $\mathcal{A}$ and 
$\mathcal{B}$ on $[n]$ and $[m]$, respectively) 
if $X_p >0$ moreover implies that 
the partial bijection $p$ is a local isomorphism. 
\eD

\bL
\label{locisolem}
\bae
\item
If $(X_p)_{|p| < k}$ is a solution to $\CONT{\ell}$ for $\ell < k$,
then it is supported by local bijections.
\item
For $k \geq 3$, if $(X_p)_{|p| < k}$ is a solution to $\CONT{\ell}$ for 
$\ell < k$ and of $\COMP{\ell}$ for $\ell = 2$, then 
it is supported by local isomorphisms.
\eae
It follows that any solution to $\ISO[k-1]$
is supported by local isomorphisms.
\eL

\prf
Suppose that $p = \abar\bbar$ is not a local bijection, w.l.o.g.\ 
(by symmetry) assume that there are $a$ and $b_1 \not= b_2$ such that 
$(a,b_1), (a,b_2) \in p$. For $p_0 := p \setminus (a,b_2)$ we clearly
have $\ell := |p_0| < k-1$, and 
looking at the instance of~$\CONT{\ell}$ 
for this $p_0$ and $a$, we find that
the two summands for $b' = b_i$, $i = 1,2$, both contribute 
to the left-hand side. So the equation and non-negativity of all
$X$-assignments imply that $X_{p_0} + X_p \leq X_{p_0}$, whence $X_p =0$. 

\medskip
For~(b) we use (a) and instances of the equations~$\COMP{2}$.
Note that for $p = \emptyset$ as well as for $|p|=1$, $p = ab$, cannot fail 
to be a local isomorphism (in undirected, loop-free graphs). Note that 
$p$ is a local isomorphism if all restrictions $p' \subset p$ of 
$p$ of size $|p'| = 2$ are local isomorphisms, 
and that $X_{p} > 0$ implies $X_{p'} >0$ for all
$p' \subset p$ by instances of $\CONT{\ell}$. So it remains 
to argue that~$\COMP{2}$ enforces that $X_{a_1a_2b_1b_2} > 0$ only for 
local isomorphisms $a_1a_2 \mapsto b_1b_2$. As $a_1a_2 \mapsto b_1b_2$ must be
a local bijection by~(a), it remains to check that
\bre
\item
$A_{a_1a_2} = 1 \;\Rightarrow\; B_{b_1b_2} \not=0$, and 
\item
$B_{b_1b_2} = 1 \;\Rightarrow\; A_{a_1a_2} \not=0$.
\ere

For~(i) we use the instance 
$
\sum_{a'} A_{a_1a'} X_{a_1b_1\,\widehat{\ }\,a'b_2} = 
\sum_{b'}  X_{a_1b_1\,\widehat{\ }\,a_1 b'} B_{b'b_2}
$
of equation~$\COMP{2}$, whose right-hand side reduces to the single term 
$X_{a_1b_1} B_{b_1b_2}$ because $X$ is supported by local bijections.
As the left-hand side is positive if $A_{a_1 a_2} = 1$ and 
$X_{a_1b_2 \,\widehat{\ }\, a_2 b_2} > 0$, (i) follows.
The case for~(ii) is strictly analogous.
\eprf

\subsubsection*{From solutions to good solutions} 
Recall Definition~\ref{goodsoldef} for good solutions to individual
equations \[
\comp[A,B] \colon\; AX = XB, 
\]
which is motivated by the 
useful properties these solutions have for the analysis of fractional
isomorphism in relation to $2$-pebble equivalence. 

We now want to analyse the solution spaces of $\ISO[k-1]$ and 
$\ISO[k-1/2]$, and very specifically of the relevant 
continuity equations with a view to strengthening the analogy with 
fractional isomorphism. 
Solutions to the continuity equations up to
arity level~$\ell$ can in fact be understood as fractional isomorphisms 
between graphs of $\ell$-tuples that govern 
the combinatorics of the plain $\ell$-pebble game
and, by extension, of the $k$-pebble game with counting.
This will in particular also allow us to resort to \emph{good solutions} with
similar benefits as in the analysis of plain fractional isomorphism
(at level $1$). 

To this end we switch to a modified view of 
a solution $(X_p)_{|p|\leq \ell}$ that 
labels entries not by local isomorphisms $p$ as \emph{sets of pairs}
$p \subset [m]\times [n]$ of size $|p|$, 
but instead, for $|p|> 0$, by pairs of \emph{full $\ell$-tuples}  
$(\abar,\bbar) \in [m]^\ell \times [n]^\ell$ such that $p(\abar) = \bbar$,
where $\abar$ is any enumeration of $\mathrm{dom}(p)$ and $\bbar$
the matching tuple enumerating $\mathrm{image}(p)$. 
With a solution $X = (X_p)_{|p|\leq \ell}$ we want to 
associate the single 
matrix 
\[
\check{X} = (\check{X}_{\abar,\bbar})_{\abar \in [m]^\ell, \bbar \in [n]^\ell}\colon
\qquad
\check{X}_{\abar,\bbar} := X_p \mbox{ for }
p = \{ (a_i,b_i) \colon i \in [\ell] \}. 
\]

It is easy to see 
that the continuity equations $\CONT{\ell'}$ for all
$\ell' \leq \ell$ for $(X_p)$ imply 
that the associated matrix $\check{X}$ is doubly stochastic.
But the continuity equations for $(X_p)$ also manifest themselves 
in the fact that $\check{X}$ is 
a fractional isomorphism  w.r.t.\ a graph representation 
of the combinatorial layout of the plain $k$-pebble game over 
$[m]^\ell,[n]^\ell$, as follows. 
For $i \in [\ell]$, the matrix $\mathbb{I}^{\ssc (i)} = 
\mathbb{I}_{m,\ell}^{\ssc (i)}$ with entries
\[
\mathbb{I}^{\ssc (i)}_{\abar,\abar'} := 
\prod_{j \not= i} \delta(a_j,a_j') \; \;
\mbox{ for $\abar,\abar' \in [m]^\ell$} 
\] 
may be regarded as the adjacency
matrix of the reflexive graph  
\[
\str{I}^{\ssc (i)}(m,\ell) 
 := ([m]^\ell, \mathbb{I}^{\ssc (i)})
\]
associated with legal moves of the 
$i$-th pebble in the plain $\ell$-pebble game over universe $[m]$. 
Its reflexive and symmetric edge relation 
links two tuples in $[m]^\ell$ if they disagree at most 
in the $i$-th component.
If $X$ solves $\CONT{\ell'}$ for $\ell' \leq \ell$, then 
$\check{X}$ describes a fractional isomorphism between
$\str{I}^{\ssc (i)}(m,\ell)$ and $\str{I}^{\ssc (i)}(n,\ell)$, 
for each $i \in [\ell]$: $\check{X}$ is doubly stochastic and a
solution to the matrix equations 
\[
\comp[\mathbb{I}_{m,\ell}^{\ssc (i)},\mathbb{I}_{n,\ell}^{\ssc (i)}]: \quad
\mathbb{I}_{m,\ell}^{\ssc (i)} \check{X} =
\check{X} \mathbb{I}_{n,\ell}^{\ssc (i)}
\qquad \mbox{ for each $i \in [\ell]$,}
\]
as is shown in Lemma~\ref{translationcontlem} below. 
We may modify this solution 
$\check{X}$ to obtain a \emph{good
 solution} in the sense of Definition~\ref{goodsoldef},
as in Corollary~\ref{goodlinkcor};
and this good solution $(\check{X}'_{\abar,\bbar})$ 
then translates back into a solution $(X'_p)$ of the continuity
equations up to $\CONT{\ell}$ with similar homogeneity benefits.

\bD
\bae
\item
Let $X = (X_p)_{|p| \leq \ell}$ be non-negative, labelled by
\emph{sets} $p \subset [m] \times [n]$ of sizes up to $\ell$. 
Define its \emph{lifting} to tuple co-ordinates to be the 
array $\check{X} := 
(\check{X}_{\abar,\bbar})_{\abar \in [m]^\ell,\bbar \in[n]^\ell}$ 
where $\check{X}_{\abar,\bbar} := X_p$ for 
$p := \{ (a_i,b_i) \colon i \in [\ell] \}$. 
\item
A matrix $\check{X} = (\check{X}_{\abar,\bbar})$ 
labelled by \emph{pairs of tuples}
$(\abar,\bbar) \in [m]^\ell \times [n]^\ell$ is said to be 
\emph{consistent}
if $\check{X}_{\abar,\bbar} = \check{X}_{\abar',\bbar'}$ 
whenever $\{ (a_i,b_i) \colon i \in [\ell] \} 
= \{ (a_i',b_i') \colon i \in [\ell] \}$. 
\item
For any consistent matrix
$\check{X} = (\check{X}_{\abar,\bbar})$ labelled by 
$(\abar,\bbar) \in [m]^\ell \times [n]^\ell$, define 
its \emph{projection} to set co-ordinates to be the array
$X := (X_p)_{|p|\leq \ell}$ with $X_\emptyset := 1$ and 
$X_p = \check{X}_{\abar,\bbar}$ for any 
$\abar,\bbar$ such that $p = \{ (a_i,b_i) \colon i \in [\ell] \}$. 
\eae
\eD

Note that the natural notions of 
being supported by local bijections, or by local
isomorphisms, in the sense of Definition~\ref{supportdef},
are preserved in the passage from $X$ to
$\check{X}$ according to~(a), and in the passage from
consistent $\check{X}$ to $X$ according to~(b).
Note also that consistency of a matrix
$\check{X} = (\check{X}_{\abar,\bbar})$ %
implies that its co-ordinatisation
is \emph{permutation invariant} in the sense that 
$\check{X}_{\abar,\bbar} = \check{X}_{\pi(\abar),\pi(\bbar)}$ for all
permutations $\pi \in S_\ell$. 
Permutation invariance guarantees in particular that 
the fractional isomorphism conditions
$ \mathbb{I}^{\ssc (i)} \check{X} = \check{X} \mathbb{I}^{\ssc (i)}$
are equivalent for any two choices of $i \in [\ell]$.

\bO
\label{conssupppresobs}
The conditions of being doubly stochastic and of permutation-invariance,
as well as of being supported by local bijections or by local isomorphisms,
are compatible with transposition and matrix products:
e.g., if $\check{X}, \check{Y}$ are consistent then 
so are $\check{X}^t$ and $\check{X}\check{Y}$.
For matrices that are supported by local bijections, consistency 
is also preserved in transposition and matrix products.
\eO

\prf
The proofs are straightforward and we just sketch the argument for
consistency. Symmetry of the condition immediately implies
compatibility with transposition. For products consider 
matrices $\check{X} = (\check{X}_{\abar,\bbar})$ and 
$\check{X} = (\check{X}_{\bbar,\cbar})$ that are both consistent
and supported by local bijections. Let $X = (X_p)$ and 
$Y = (Y_p)$ be their projections to set co-ordinates, which are 
therefore also supported by local bijections. 
Let $\check{Z}:= \check{X}\check{Y}$
and consider index pairs $(\abar,\cbar)$ and $(\abar',\cbar')$ such that 
$\{ (a_i,c_i) \colon i \in [\ell]\} = \{ (a_i',c_i') \colon i \in
[\ell]\}$.
Let $s$ be the size of these sets. 
Then 
\[
\check{Z}_{\abar,\cbar} =
\sum_{\bbar} 
\check{X}_{\abar,\bbar} \check{Y}_{\bbar,\cbar} 
=
1/s!  \!\!\! \sum_{q \circ p \colon \abar \mapsto \cbar} \!\!\!
X_p Y_q 
=
1/s!  \!\!\! \sum_{q \circ p \colon \abar' \mapsto \cbar'}  \!\!\!
X_p Y_q 
=
\check{Z}_{\abar',\cbar'}.
\]
Note that the $p,q$-sums are over all
pairs of local bijections whose compositions 
are represented by the graph 
$\{ (a_i,c_i) \colon i \in [\ell]\} = \{ (a_i',c_i') \colon i \in [\ell]\}$.
That just local bijections need to be considered, is due to the support
of $\check{X}$ and $\check{Y}$ by local bijections: all other terms
vanish; the uniform translation of these terms into set co-ordinates 
relies on consistency of $\check{X}$ and $\check{Y}$.
\eprf

\bL 
\label{translationcontlem}
\bae
\item
Let $X= (X_p)_{|p|\leq \ell}$ be 
labelled by sets $p \subset [m] \times [n]$ of sizes up to $\ell$
and let $\check{X}$ be its consistent lifting. 
If $X$ solves the continuity equations up to level $\ell$, then 
$\check{X}$
is doubly stochastic (which implies $n=m$) 
and commutes with the $\mathbb{I}^{\ssc (i)} = \mathbb{I}^{\ssc (i)}_{n,\ell}$:
\[
\mathbb{I}^{\ssc (i)} \check{X}  = \check{X} \mathbb{I}^{\ssc (i)} 
\; \mbox{ for } i \in [\ell]. 
\]
In other words, 
$\check{X}$ is a fractional automorphism
of the (reflexive) graph $\str{I}^{\ssc (i)}(n,k)$, for each $i \in [\ell]$.
\item
Conversely, let $\check{X} = 
(\check{X}_{\abar,\bbar})$, labelled 
by $(\abar,\bbar) \in [m]^\ell \times [n]^\ell$, 
be consistent 
and let $X = (X_p)_{|p| \leq \ell}$ be its projection to
set co-ordinates. 
If $\check{X}$ is a fractional isomorphism between 
$\str{I}^{\ssc (i)}(m,\ell)$ and $\str{I}^{\ssc (i)}(n,\ell)$
for some $i$ (and hence for each $i$), then $n=m$ and 
$X$ satisfies the continuity equations up to level $\ell$.
\eae
\eL

\prf%
For part~(a), assume first that $X$ solves the continuity equations
up to level $k$. By induction on the number $|[\abar]|$ 
of distinct components in $\abar \in [n]^\ell$, we show that 
the row sum $\sum_{\bbar} \check{X}_{\abar,\bbar} = 1$. For
$|[\abar]| = 1$, $\abar = a^\ell$ and  
$\sum_{\bbar} \check{X}_{\abar,\bbar} = 
\sum_{b} \check{X}_{a^\ell,b^\ell} = \sum_b X_{\emptyset\,\widehat{\ }\,{ab}}
= X_\emptyset = 1$ by the continuity equations for $X$. 
The induction step is treated analogously: if $|[\abar]|$ 
is such that, e.g., $a_\ell \not\in \{ a_1,\ldots,a_{\ell-1}\}$, we 
associate with any $\bbar \in [n]^\ell$ such that
$\check{X}_{\abar,\bbar} > 0$
the local bijection $p = \{ (a_i,b_i) \colon i < \ell \}$ and 
use the identity $\check{X}_{\abar,\bbar} = X_{p\,\widehat{\ }\,a_\ell
  b_\ell}$
to first rewrite 
\[
\sum_{b_\ell} \check{X}_{\abar,\bbar} = \sum_{b_\ell} X_{p\,\widehat{\
  }\,a_\ell b_\ell}
= X_p = 
\check{X}_{\abar\frac{a_1}{\ell},\bbar\frac{b_1}{\ell}}
\]
using continuity equations for $X$. Therefore, 
\[
\sum_{\bbar} \check{X}_{\abar,\bbar} = \sum_{b_1,\ldots,b_{\ell-1}}
\sum_{b_\ell} X_{\abar\bbar}
=
\sum_{b_1,\ldots,b_{\ell-1}}
\check{X}_{\abar\frac{a_1}{\ell},\bbar\frac{b_1}{\ell}}
=
\sum_{\bbar}
\check{X}_{\abar\frac{a_1}{\ell},\bbar} = 1
\]
by inductive hypothesis. For the commutation condition
$\mathbb{I}^{\ssc (i)} \check{X} = \check{X} \mathbb{I}^{\ssc (i)}$
we check that, for any $\abar,\bbar \in [n]^\ell$, and for
$p := \{ (a_j,b_j) \colon j \not= i \}$: 
\[
(\mathbb{I}^{\ssc (i)} \check{X} )_{\abar,\bbar} = 
\sum_{\abar'}
\mathbb{I}^{\ssc (i)}_{\abar,\abar'} \check{X}_{\abar',\bbar} 
= 
\sum_{a'} 
\check{X}_{\abar\frac{a'}{i},\bbar}
= 
\sum_{a'} X_{p\,\widehat{\ }\,a'b_i} 
= 
X_p
\]
and similarly 
\[
(\check{X} \mathbb{I}^{\ssc (i)})_{\abar,\bbar} = 
\sum_{\bbar'}
\check{X}_{\abar,\bbar'} 
\mathbb{I}^{\ssc (i)}_{\bbar',\bbar} 
= 
\sum_{b'} 
\check{X}_{\abar,\bbar\frac{b'}{i}}
= 
\sum_{b'} X_{p\,\widehat{\ }\,a_ib'} 
= 
X_p. 
\]

Towards the converse, (b), the last two equations also show 
that the fractional isomorphism conditions for $\check{X}$ w.r.t.\ 
$\mathbf{I}^{\ssc (i)}$ imply that the sums
\[
\sum_{a'} X_{p\,\widehat{\ }\,a'b_i}  
= 
\sum_{b'} X_{p\,\widehat{\ }\,a_ib'}
\] 
must be independent 
of the choice of $a_i$ and $b_i$. 
Their value, therefore, must be equal to 
$\sum_{b'} X_{p\,\widehat{\ }\,a_1 b'} =  X_{p\,\widehat{\ }\,a_1 b_1}
= X_p$, whence the continuity equations 
at all levels $2 \leq \ell' \leq \ell$ are established. For 
level $1$, $1 = X_\emptyset = \sum_a X_{ab} = \sum_b X_{ab}$
follows from the doubly stochastic nature
of $\check{X}$, which implies, e.g., that $1 = 
\sum_{\bbar} \check{X}_{a^\ell \bbar} =
\sum_b \check{X}_{a^\ell b^\ell} 
= \sum_b X_{ab}$. 
\eprf

\bD
\label{goodtranslationdef}
A solution $X = (X_p)_{|p| \leq \ell}$ 
to $\CONT{\ell'}$ for $\ell' \leq \ell$ 
is \emph{good} if its lifting 
$\check{X} = (\check{X}_{\abar,\bbar})$
to tuple co-ordinates is a good solution (in the sense of 
Definition~\ref{goodsoldef}) to 
$\comp[\mathbb{I}^{\ssc (i)}(m,\ell),
\mathbb{I}^{\ssc (i)}(n,\ell)]$ (simultaneously for
every $i \in [\ell]$).
\eD

Recall from Definition~\ref{goodsoldef} that a good solution,
$\check{X}$ in tuple co-ordinates, induces $\check{X}$-related 
partitions of the vertex 
sets $[m]^\ell = \dot{\bigcup}_s D_s$ and $[n]^\ell =
\dot{\bigcup}_s D_s'$ of the 
$\mathbb{I}^{\ssc (i)}$ such that
\bre
\item
these partitions are equivalent and stable w.r.t.\ the edge relations of the  
$\mathbb{I}^{\ssc (i)}$, for every $i \in [\ell]$;
\item
$\check{X}_{D_sD_s'} > 0$ for all $s$; 
\item
$\check{X}_{D_sD_t'} = 0$ for $s \not= t$.
\ere

Note that, if 
$(X_p)$ is a good solution that is 
supported by local isomorphisms, then 
the $\check{X}$-related partitions $(D_s)$ and $D_s')$ 
are such that 
$\abar \mapsto \bbar$ is a local isomorphism between $\str{A}$ and $\str{B}$
whenever $\check{X}_{\abar,\bbar} >0$, i.e., whenever 
$\abar$ and $\bbar$ are from matching partition sets $D_s$ and $D_s'$.

\medskip
Suppose that $X = (X_p)_{|p| \leq \ell}$ is a solution to 
$\CONT{\ell'}$ for $\ell' \leq \ell$ that is supported by local
isomorphisms. Note that by Lemma~\ref{locisolem} 
this is the case for any solution to $\ISO(k-1)$ and for $\ell = k-1$. 
Then by Lemma~\ref{translationcontlem}~(a) 
the lifting $\check{X}$ of $X$ to tuple co-ordinates 
is a fractional isomorphism between 
$\mathbb{I}^{\ssc (i)}(m,\ell)$ and $\mathbb{I}^{\ssc (i)}(n,\ell)$,
which is also supported by local isomorphisms 
between $\str{A}$ and $\str{B}$ in the sense that 
$\check{X}_{\abar,\bbar} > 0$ implies that
$\abar \mapsto \bbar$ is a local isomorphism
between $\str{A}$ and $\str{B}$. 
The associated matrices $\check{Z} := \check{X}^t \check{X}$ and 
$\check{Z}' := \check{X} \check{X}^t$ are 
doubly stochastic, symmetric with strictly positive diagonal, and 
also supported by local isomorphisms between $\str{A}$ and $\str{B}$ 
in the above sense. As in Corollary~3.9., we thus 
obtain a good solution $\check{X}'$ for 
the fractional isomorphism between 
$\mathbb{I}^{\ssc (i)}(m,\ell)$ and $\mathbb{I}^{\ssc (i)}(n,\ell)$. 
This good solution is 
of the form $\check{X}' = \check{Z}^n \check{X}$.
Note that, by Lemma~\ref{translationcontlem}~(b), 
the projection of $\check{X}'$ to set co-ordinates, 
$X' = (X'_p)_{|p| \leq \ell}$, 
is again a solution to $\CONT{\ell'}$ for $\ell' \leq \ell$,
and it is supported by local isomorphisms. Regarding the relationship
between $(X_p)$ and $(X_p')$, also note that 
$X_p >0$ implies $X_p' >0$  as $\check{Z}$ and its powers have
strictly positive diagonal.

We summarise these observations regarding solutions to
$\CONT{\ell}$ that are supported by local isomorphisms 
in a corollary before turning to the status of the 
induced solutions w.r.t.\ levels of $\COMP{\ell}$.

\bC
\label{translatetogoodcor} 
If the equations 
$\CONT{\ell'}$ for $\ell' \leq \ell$ for $\str{A}$ 
on $[m]$ and $\str{B}$ on $[n]$
admit any solution $X = (X_p)_{|p| \leq \ell}$ 
that is supported by local isomorphisms between 
$\str{A}$ and $\str{B}$, then they also admit a solution 
$X' = (X'_p)_{|p| \leq \ell}$ that is \emph{good} in the sense of
Definition~\ref{goodtranslationdef} and supported by local
isomorphisms. Moreover, the natural good solution
$X'$ associated with a given solution $X$ is 
strictly positive where $X$ is.
\eC

We are now ready to prove the converse of Lemma~\ref{mixedlevellem},
thus matching $\LC^k$-equivalence to $\ISO(k-1/2)$. Recall
$\ISO(k-1/2)$ (cf.~the table given for Lemma~\ref{mixedlevellem} above):
these intermediate levels Sherali--Adams levels combine 
the equations~$\COMP{\ell}$, concerning compatibility with the edge relations,  
of level $\ell < k$ with the continuity equations~$\CONT{\ell}$ of 
levels $\ell \leq k$.

Let $\str A$, $\str B$ be graphs with vertex sets
$[m],[n]$, respectively, $A,B$ be their symmetric adjacency
matrices. 
We know from %
Lemma~\ref{mixedlevellem} that $\LC^k$-equivalence of 
$\str{A}$ and $\str{B}$ implies 
the existence of a solution for exactly this combination of equations. 
The following theorem says that the converse is also true. 

\bT
\label{SAkminus1/2prop}
$\ISO[k-1/2]$ has a solution if, and only if, $\str A\equiv_\LC^{k}\str B$.
\eT

\prf
It remains to argue for the implication from solvability of 
$\ISO[k-1/2]$ to $\LC^k$-equivalence. In fact, if 
$(X_p)_{|p| \leq k}$ is any solution to $\ISO[k-1/2]$, then 
\[
X_{\abar\bbar} > 0 \quad\Longrightarrow\quad \str{A},\abar \equiv_\LC^k
\str{B},\bbar.
\]

We know from Lemma~\ref{locisolem} that any solution $X = (X_p)$
to $\ISO[k-1/2]$ is supported by local isomorphisms. By 
Corollary~\ref{translatetogoodcor} we may therefore assume without 
loss of generality that the given solution 
is good in the sense of Definition~\ref{goodtranslationdef}.
So its lifting $\check{X} = (\check{X}_{\abar,\bbar})$ to tuple
co-ordinates induces $\check{X}$-related partitions 
$[n]^k = \dot{\bigcup}_{s} D_s$ and $[n]^k = \dot{\bigcup}_{s} D_s'$
that are equivalent stable partitions 
w.r.t.\ $\mathbb{I}^{\ssc (j)}(n,k)$ for each $j \in
[k]$, and $\check{X}_{\abar,\bbar} > 0$ precisely if the tuples 
$\abar$ and $\bbar$ are from matching partition sets. It suffices to 
provide strategy a for $\PII$ to maintain this condition. For a round
played in component $j$, 
\[
\textstyle 
\#_a ( \abar\frac{a}{j} \in D_t ) = 
\#_b ( \bbar\frac{b}{j} \in D_t') 
\]
follows from equivalence and stability of the partitions 
w.r.t.\ the edge relation of $\mathbb{I}^{\ssc (j)}$.
Player \PII\ can offer a bijection $a \mapsto b$ 
that matches partition indices for $\abar\frac{a}{j}$ 
and $\bbar\frac{b}{j}$, and thus maintains strict positivity 
of $\check{X}$ and $X$.
\eprf

Note that, of the compatibility equations in $\ISO(k-1/2)$,
we only had to use the level~$2$ equation $\COMP{2}$, which is 
sufficient to establish support by local isomorphisms
(cf.\ Lemma~\ref{locisolem}~(b)). The rest of the compatibility 
levels are actually redundant, and this only changes when 
the levels of compatibility and continuity equations 
are properly matched, as they are in the regular Sherali--Adams 
relaxation levels of the isomorphism problem. 
To understand the nature of $\ISO(k-1)$ 
we need to take the higher levels of the compatibility equations
into account. To this end we turn to the interpretation of the compatibility
equations $\COMP{\ell}$ in terms of liftings to tuple co-ordinates.
 
By Lemma~\ref{locisolem}, even the level~$2$ equations
$\COMP{2}$ guarantee that solutions are supported by local isomorphisms.
But if $X = (X_p)_{|p| \leq \ell}$ satisfies 
$\COMP{\ell'}$ for all levels $\ell' \leq \ell$, 
then its lifting $\check{X} = (\check{X}_{\abar,\bbar})$ 
to tuple co-ordinates also satisfies the following equations 
for all $\abar \in [m]^\ell, \bbar \in
[n]^\ell, a \in [m], b \in [n]$ and each $i \in [\ell]$:
\[
\sum_{a'} A_{aa'} \check{X}_{\abar\frac{a'}{i},\bbar\frac{b}{i}}
=
\sum_{b'}  \check{X}_{\abar\frac{a}{i},\bbar\frac{b'}{i}} B_{b'b}.
\]

For this it suffices to note that the set projection of 
$(\abar\frac{a}{i},\bbar\frac{b}{i})$ is $p\,\widehat{\ }\,ab$
for $p = \{ a_jb_j \colon j \not= i \}$. In order to view these
equations as commutativity conditions in the style of 
fractional isomorphism, we artificially lift the adjacency matrices 
$A$ and $B$ to tuple co-ordinates with entries $A_{\abar,\abar'}$ and 
$B_{\bbar,\bbar'}$ for pairs of tuples $\abar,\abar' \in [m]^\ell$ and 
$\bbar,\bbar' \in [n]^\ell$ by putting, for instance in the case of
$A$, for $i \in [\ell]$:
\[
A_{\abar,\abar'}^{\ssc (i)}:= \prod_{j \not= i} \delta(a_j,a_j') A_{a_ia_i'}.
\]

Note that these matrices are symmetric if $A$ is.
We obtain an equivalence between $\COMP{\ell}$ for $X$ and 
$\comp[A^{\ssc (i)},B^{\ssc (i)}]$ for $\check{X}$. 

\bL
\label{translationcompeqnslem}
For any matrix $X = (X_p)_{|p| \leq \ell}$ labelled by sets
$p \subset [m] \times [n]$ of size up to $\ell$ and its lifting to 
tuple co-ordinates $\check{X} = (\check{X}_{\abar,\bbar})$ labelled by
tuples $\abar \in [m]^\ell, \bbar \in [n]^\ell$, the following are
equivalent:
\bre
\item 
$X= (X_p)$ satisfies the equations $\COMP{\ell'}$  for $\ell' \leq
\ell$;
\item
$\check{X} = (\check{X}_{\abar,\bbar})$  satisfies 
$\comp[A^{\ssc (i)},B^{\ssc (i)}]$ for some (and hence any) $i \in
[\ell]$.
\ere 
\eL

\prf
The equation $\comp[A^{\ssc (i)},B^{\ssc (i)}]$ for $\check{X}$ is
\[
\sum_{\abar' \in [m]^\ell} 
A_{\abar,\abar'}^{\ssc (i)} \check{X}_{\abar',\bbar}
=
\sum_{\bbar' \in [n]^\ell}  
\check{X}_{\abar,\bbar'} B^{\ssc (i)}_{\bbar',\bbar},
\]
and its equivalence with the equations $\COMP{\ell'}$ for 
levels $\ell'$ up to $\ell$ is immediate from the definition of 
$A^{\ssc (i)}/B^{\ssc (i)}$ in terms of $A/B$  and of 
$\check{X}$ in terms of $X$. 
\eprf

The following summarise the relevant translations of continuity 
and compatibility equations, which are expressed in terms of set
co-ordinates $X_p$, to their liftings 
to tuple co-ordinates $\check{X}_{\abar,\bbar}$; compare  
Lemmas~\ref{translationcontlem} and~\ref{translationcompeqnslem}.

\[
\mbox{$\displaystyle%
\barr{l}
\left.%
\barr{@{}l@{}}
\sum_{\abar'} 
\check{X}_{\abar'\bbar} =
\sum_{\bbar'} 
\check{X}_{\abar\bbar'} = 1
\\
\\
\sum_{\abar'} \mathbb{I}^{\ssc
  (i)}_{\abar\abar'} \check{X}_{\abar',\bbar} 
=
\sum_{\bbar'} \check{X}_{\abar,\bbar'}\mathbb{I}^{\ssc
  (i)}_{\bbar'\bbar} 
\\
\hnt
\mbox{for all } 
\abar \in [m]^\ell, \bbar \in [n]^\ell, \mbox{ and all }i \in [\ell]
\earr
\qquad\;\,%
\right\} 
\quad \mbox{lifting of }\CONT{\ell'}, \ell' \leq \ell 
\\
\\
\left.%
\barr{@{}l@{}}
\sum_{\abar'}%
A_{\abar,\abar'}^{\ssc (i)} \check{X}_{\abar',\bbar}
=
\sum_{\bbar'}%
\check{X}_{\abar,\bbar'} B^{\ssc (i)}_{\bbar',\bbar}
\\
\hnt
\mbox{for all } 
\abar \in [m]^\ell, \bbar \in [n]^\ell, \mbox{ and all }i \in [\ell]
\earr
\quad\;\, \quad
\right\} 
\quad \mbox{lifting of }\COMP{\ell'},\ell' \leq \ell 
\earr$}%
\]

As in the analysis of solutions for the fractional isomorphism problem,
we see that passage to a good solution 
$\check{X}'$ of the form $\check{X}' = \check{Z}^n \check{X}$
for $\check{Z} = \check{X} \check{X}^t$ is compatible with this type
of commutativity condition.

Starting from a solution $X$ to $\ISO(k-1)$, whose 
lifting to tuple co-ordinates $\check{X}$ satisfies the above 
combination of equations for $\ell = k-1$, we obtain a good solution
$\check{X}'$, which induces $\check{X}'$-related partitions of 
$[m]^{k-1}$ and $[n]^{k-1}$ as in Definition~\ref{goodtranslationdef}.
These partitions are now simultaneously stable w.r.t.\ 
the $\mathbb{I}^{\ssc (i)}$ and w.r.t.\ the liftings of the edge relations 
$A^{\ssc (i)}$ and $B^{\ssc (i)}$, for each $i \in [\ell]$. In other
words, we obtain a simultaneous good solution $\check{X}'$
in the sense of Definition~\ref{goodsoldef} 
for 
\[
\comp[\mathbb{I}^{\ssc (i)}_{m,k-1},\mathbb{I}^{\ssc (i)}_{n,k-1}]\,,\;
\comp[A^{\ssc (i)},B^{\ssc (i)}] \quad\mbox{ for all } i \in [k-1].
\]

We thus obtain the following, by reverse
translation of a good solution $\check{X}'$ into 
set co-ordinates. 

\bC
\label{translatetogoodcompcor} 
Let $k \geq 3$. Any solution $X = (X_p)_{|p| < k}$ 
to $\ISO(k-1)$ induces a \emph{good} solution  
$X' = (X_p')_{|p| < k}$ to $\ISO(k-1)$ 
with $\check{X}'$-related induced partitions of 
$[m]^{k-1} = \dot{\bigcup}_s D_s$ and 
$[n]^{k-1} =  \dot{\bigcup}_s D_s'$ such that 
\bre
\item 
these partitions are equivalent stable partitions 
w.r.t.\ the edge relations of the $\mathbb{I}^{\ssc (i)}$ 
for each $i \in [\ell]$;
\item
these partitions 
are equivalent stable partitions w.r.t.\ the liftings of the edge relations 
$A^{\ssc (i)}$ on $[m]^{k-1}$ and $B^{\ssc (i)}$ on $[n]^{k-1}$, 
for each $i \in [\ell]$;
\item
$\check{X}'_{D_sD_s'} > 0$ for all $s$;
\item
$\check{X}_{D_sD_t}' = 0$ for all $s\not= t$.
\ere
In particular, $\abar \mapsto \bbar$ is a local isomorphism between 
$\str{A}$ and $\str{B}$ for $\abar$ and $\bbar$ from matching
partition sets.
Moreover, $X'$ is strictly positive where the given $X$ is.
\eC

\subsubsection*{\boldmath Level~$k-1$ solutions and 
$\LC^{<k}$-equivalence}
As we shall see in Section~\ref{gapsec}, a solution $X = (X_p)_{|p| < k}$ to
$\ISO(k-1)$ is in fact not strong enough to guarantee 
$\LC^k$-equivalence between $\str{A}$ and $\str{B}$.  
Instead it matches to a slightly lesser level of equivalence,
\[
\str{A}\equiv_\LC^{<k} \str{B},
\]
which we characterise in terms of a modified game, the 
\emph{weak bijective $k$-pebble 
game} over $\str{A},\str{B}$. 
The game is played by two players. If
$m\neq n$, player \PII\ loses immediately. Otherwise, a play of the
game proceeds in a sequence of rounds.
Positions of the game are sets $p\subseteq [m]\times [n]$ of size
$|p|\le k-1$. Normally, the initial position is $\emptyset$, but we
will also consider plays of the game starting from other initial positions. A single round of
the game, starting in position $p$, is played as follows.
\begin{enumerate}%
\item If $|p|=k-1$, player \PI\ selects a pair $ab\in p$. 
\\
If $|p|<k-1$, this step is omitted.
\item Player \PII\ selects a bijection between $[m]$ and
  $[n]$ (recall that $m=n$).
\item Player \PI\ chooses a pair $a'b'$ from this bijection.
\item If $p^+:=p\,\widehat{\ }\,a'b'$ is a local isomorphism then the new
  position is 
  \[
  p':=
  \begin{cases}
   \;\; (p\!\setminus\! ab)\,\widehat{\ }\,a'b'&\text{ if }|p|=k-1,\\
\hnt
    \;\;\, p \,\widehat{\ }\,a'b'&\text{ if }|p|<k-1.
  \end{cases}
  \]
  Otherwise, the
  play ends and player \PII\ loses.
\end{enumerate}
Player \PII\ wins a play if it lasts forever, i.e., if $m=n$ 
and she never loses in step 4 of a round.

By comparison, a round in the ordinary bijective
$(k-1)$-pebble game, which characterises $\equiv^{k-1}_\LC$
according to Theorem~\ref{theo:hel},  
can be described as follows.

\label{page:bij-game2}
\begin{enumerate}%
\item If $|p|=k-1$, player \PI\ selects a pair $ab\in p$. 
\\
If $|p|<k-1$, this step is omitted.
\item Player \PII\ selects a bijection between $[m]$ and
  $[n]$.
\item Player \PI\ chooses a pair $a'b'$ from this bijection.
\item The new position is
  \[
  p':=
  \begin{cases}
   \;\; (p\!\setminus\! ab)\,\widehat{\ }\,a'b'&\text{ if }|p|=k-1,\\
\hnt
    \;\;\, p \,\widehat{\ }\,a'b'&\text{ if }|p|<k-1,
  \end{cases}
  \]
  provided it is a local isomorphism.
  Otherwise, the
  play ends and player \PII\ loses.
\end{enumerate}

Note that the weak bijective $k$-pebble game requires more of the second player
than the bijective $(k-1)$-pebble game, because $p^+$ rather than just 
$p'$ is required to be a local isomorphism. On the other hand, it 
requires less than the  bijective $k$-pebble game: 
the bijective $k$-pebble game precisely 
requires the second player to choose the bijection without 
prior knowledge of the pair $ab$ that will be removed from the
position (cf.\ the alternative presentation of the bijective
$k$-pebble game on page~\pageref{page:alt-game}).
A strategy for player \PII\ in the weak version is good for the usual
version if it is fully symmetric or uniform w.r.t.\ the pebble pair 
that is going to be removed.

However, this is only relevant if $k\ge 3$. The
weak bijective $2$-pebble game and the bijective $2$-pebble game are
essentially the same.

\bD
$\str{A}$ and $\str{B}$ are $\LC^{<k}$-equivalent, 
$\str{A}\equiv_\LC^{<k} \str{B}$, if 
the second player has a winning strategy in the weak bijective
$k$-pebble game on $\str A$, $\str B$.

Furthermore, for tuples $\vec a$ and $\vec b$ of the same length
$\ell<k$ we let $\str{A},\vec a\equiv_\LC^{<k} \str{B},\vec b$ if 
the second player has a winning strategy in the weak bijective
$k$-pebble game on $\str A$, $\str B$ with initial position $\abar
\vec b$.
\eD 

\bO
\label{sandwichobs}
$
\str A\equiv_\LC^2\str B\;\Leftrightarrow\;\str A\equiv_\LC^{<2}\str B,
$
and for all $k \geq 3$:
\[
\str{A}\equiv_\LC^k \str{B}
\quad \Rightarrow \quad 
\str{A}\equiv_\LC^{<k} \str{B}
\quad \Rightarrow \quad 
\str{A}\equiv_\LC^{k-1} \str{B}.
\]
\eO

\bR
\label{edgemoverem}
The weak bijective $k$-pebble game is equivalent to a
bisimulation-like game with $k-1$ pebbles where in each round the
first player may slide a pebble along an edge of one of the graphs
and player \PII\ has to answer by sliding the corresponding pebble
along an edge of the other graph. In this version, the game coresponds
to the $(k-1)$-pebble sliding game introduced by Atserias and
Maneva~\cite{AtseriasManeva}. They prove that equivalence of two graphs with
respect to the $(k-1)$-pebble sliding game implies that $\ISO[k-1]$
has a solution. In view of the equivalence of the sliding game with
our weak bijective $k$-pebble game, this implies the backward
direction of Theorem~\ref{SAkminus1prop} below.
\eR

Let $\wp_{r}$ 
be the set of positions of size $k-1$ of the weak bijective $k$-pebble game over $\str{A},\str{B}$ in which 
the second player has a strategy to survive through ${r}$ rounds. 
Let $\sim^{r}$ stand for the 
equivalence relation induced by $\wp_{r}$, i.e.,  
the symmetric transitive 
closure of the relation that puts $\abar \sim^{r} \bbar$ if 
$p = \abar \bbar \in \wp_{r}$.
Note that 
$\sim^r$ is compatible with permutations
in the sense that, for instance, 
$\abar \sim^r \bbar$ iff 
$\pi(\abar) \sim^r \pi(\bbar)$ for any $\pi \in S_{k-1}$. 
We write $\pi(\abar)$ for the 
application of the permutation $\pi \in S_{k-1}$ to the components of 
$\abar = (a_1,\ldots, a_{k-1})$, which results in 
$\pi(\abar) = (a_{\pi(1)},\ldots, a_{\pi(k-1)})$.

For ${r} = 0$, the set $\wp_0$ consists of all local isomorphisms of size $k-1$.  
We characterise $\wp_{{r}+1}$ and $\sim^{{r}+1}$ in terms of $\wp_{r}$ by means of
back\&forth conditions for a single round:  
$\abar \sim^{{r}+1} \bbar$ ($p = \abar\bbar \in \wp_{{r}+1}$) iff 
position $\abar\bbar$ is good in the following sense:
for $1 \leq j \leq{k-1}$, the second player has a response 
that guarantees a target position in $\wp_{r}$ if the first player 
chooses index $j$. 

I.e., for each $1 \leq j\leq{k-1}$, the second player 
needs to have a bijection
$\rho_j$
between $[m]$ and $[n]$ such that for every $ab \in \rho_j$ 
\[
\atp(\abar a) = \atp(\bbar b)
\quad\mbox{ and }\quad 
\textstyle \abar\frac{a}{j} \bbar\frac{b}{j}  \in \wp_{r}.
\]

The first condition says that $p\,\widehat{\ }\,ab$ is a local
isomorphism, the second that the new position is good for ${r}$ further rounds.

Note that, since $\str A$ is a graph, 
the quantifier-free type $\atp(\abar a)$ is 
fully determined by $\atp(\abar)$ and 
the $\atp(a_i a)$ for $1 \leq i \leq {k-1}$. 
The condition that $\atp(\abar) = \atp(\bbar)$ is a pre-condition 
for the round to be played; 
the condition that $\atp(a_i a) = \atp(b_i b)$ 
for all $i \not= j$, on the other hand,
is part of the post-condition that 
$p\!\setminus\! a_jb_j \,\widehat{\ }\,ab$ 
is a local isomorphism.%

Let $(\alpha_i)_{i\in I}$ be an enumeration of 
the $\sim^{r}$-classes over $\str{A}$ and $\str{B}$.
Then the above conditions on membership
of $p = \abar\bbar$ in $\wp_{{r}+1}$ 
are equivalent to the following: 

\[
\barr{l}
\mbox{%
\btfll for each $1 \leq j \leq {k-1}$,
\\
for every $\sim^{r}$-class $\alpha$, and
\\
for every quantifier-free type $\eta(x,y)$:  
\etfll}
\\
\#_a^\str{A} 
\bigl(
\abar {\textstyle \frac{a}{j}} \in \alpha 
\wedge 
\atp(a_ja)  \!=\! \eta 
\bigr)
\;\;=\;\;
\#_b^\str{B} 
\bigl(
\bbar {\textstyle \frac{b}{j}} \in \alpha 
\wedge
\atp(b_jb)  \!=\! \eta 
\bigr).
\earr
\]
Note, towards the claimed equivalence, that  
these numerical equalities allow the second 
player to piece together a bijection that respects the partition 
of $[m]$ and $[n]$ 
according to different combinations of $\alpha$ and
$\eta$, which in turn guarantees that any pair $ab$ drawn from the
bijection respects this partition and hence 
leads to a position $\abar\frac{a}{j}\bbar\frac{b}{j} \in \wp_{r}$
as required.

Conversely, if one of these equalities were violated, then 
any bijection will have to have at least one pair that does not
respect the partition of $[m]$ and $[n]$ 
w.r.t.\ the $\alpha$ and $\eta$.
If the first player picks such a bad pair $ab$, then the second player 
loses during this round because $\atp(\abar a) \not= \atp(\bbar b)$, or 
because the resulting new position 
$\abar\frac{a}{j}\bbar\frac{b}{j}$ is not in $\wp_{r}$.

For later use we state the condition on full $\LC^{<k}$-equivalence,
corresponding to the stable limit of the above refinement step.
For $\abar \in [m]^{k-1}$ and $\bbar \in [n]^{k-1}$, 
$\str{A},\abar \equiv_\LC^{<k} \str{B},\bbar$ iff
\begin{equation}
\label{equivlimitcond}
\barr{l}
\mbox{%
\btfll for each $1 \leq j \leq k-1$,
\\
for all $\LC^{<k}$-equivalence classes $\alpha$, and
\\
for every quantifier-free type $\eta(x,y)$:  
\etfll}
\\
\#_a^\str{A} 
\bigl(
\abar {\textstyle \frac{a}{j}} \in \alpha 
\wedge 
\atp(a_ja)  \!=\! \eta 
\bigr)
\;\;=\;\;
\#_b^\str{B} 
\bigl(
\bbar {\textstyle \frac{b}{j}} \in \alpha 
\wedge
\atp(b_jb)  \!=\! \eta 
\bigr).
\earr
\end{equation}

\medskip

\bL
\label{solutiontoskequivlem}
For $k \geq 3$, 
if $(X_p)_{|p| < k}$ is a solution to $\ISO[k-1]$ then for all 
$\abar \in [n]^{k-1}$ and $\bbar \in [n]^{k-1}$:
\[
X_{\abar\bbar} > 0 \;
\quad \Longrightarrow \quad 
\str{A},\abar \equiv_\LC^{<k} \str{B},\bbar.
\]
\eL

\prf
By Corollary~\ref{translatetogoodcompcor} we may assume that 
the given solution $X = (X_p)_{|p| < k}$ itself is good in the sense of
Definition~\ref{goodtranslationdef}. It follows form $\CONT{1}$ that 
$\str{A}$ and $\str{B}$ have the same size, and we let $[n]$ be their
vertex set. By $\COMP{2}$, $X$ is supported by local isomorphisms,
cf.\ Lemma~\ref{locisolem}. 
That the solution is good means that its 
lifting to tuple co-ordinates $\check{X}$,
where $\check{X}_{\abar,\bbar} = X_p$ for 
$p = \{ (a_i,b_i) \colon i < k \}$, induces $\check{X}$-related
partitions of the vertex set of $\mathbb{I}^{\ssc (i)}(n,k-1)$ of the
form 
\[
[n]^{k-1} = \dot{\bigcup}_s D_s \; \mbox{ and } \;
[n]^{k-1} = \dot{\bigcup}_s D_s'
\]
such that
\bre
\item
each partition is stable w.r.t.\ 
$\mathbb{I}^{\ssc (i)}$ for $i < k$;
\item
$[n]^{k-1} = \dot{\bigcup}_s D_s$ is stable w.r.t.\
$A^{\ssc (i)}$ and  
$[n]^{k-1} = \dot{\bigcup}_s D_s'$ is stable w.r.t.\
$B^{\ssc (i)}$, for each $i < k$;
\item
$\check{X}_{D_sD_s'} > 0$ so that 
$\abar \mapsto \bbar$ is a local isomorphism 
whenever $\abar$ and $\bbar$ are from matching 
partition sets; 
\item
$\check{X}_{D_sD_t'} = 0$ for all $t \not= s$.
\ere

It suffices to
exhibit a strategy for
\PII\ that maintains the condition $X_p > 0$, 
or equivalently $\check{X}_{\abar,\bbar} >0$.
The argument is completely analogous to that in 
Lemma~\ref{goodsolfraclem}, for good solutions in the context of basic  
$1$-dimensional fractional isomorphism. 
Consider tuples $\abar$ and $\bbar$ of length $k-1$ such that 
$\check{X}_{\abar,\bbar} > 0$. It suffices to show that, for each $j < k$,
\PII\ can choose a bijection $\rho_j$ of $[n]$ (for a round played 
in component~$j$) such that $a_jb_j \in \rho_j$ 
and for all pairs $ab \in \rho_j$
\[
A_{a_ja} = 1 \;\Leftrightarrow\; B_{b_jb} = 1 \quad \mbox{ and } \quad
\check{X}_{\abar\frac{a}{j},\bbar\frac{b}{j}} > 0.
\]

The second condition precisely requires $\abar\frac{a}{j}$ and 
$\bbar\frac{b}{j}$ to be from matching partition sets while the 
first condition is equivalent to
\[
A^{\ssc (j)}_{\abar,\abar\frac{a}{j}} =1  \;\Leftrightarrow\;
B^{\ssc (j)}_{\bbar,\bbar\frac{b}{j}} =1. 
\]

The existence of the desired bijection therefore 
follows directly from the properties of $X$ as a good solution,
which implies that the $\check{X}$-related partitions are equivalent
stable partitions w.r.t.\ $A^{\ssc (j)}$ and $B^{\ssc (j)}$. Thus,
for every partition index $t$,
\[
\barr{r@{}l}
\#_{a}( \abar\frac{a}{j} \in D_t \wedge A_{a_ja}\!=\!1 ) 
& \; =
\#_{\abar'}( \abar' \in D_t \wedge A^{\ssc
  (j)}_{\abar,\abar'}\!=\!1) 
\\
& \; = 
\#_{\bbar'}( \bbar' \in D_t \wedge B^{\ssc
  (j)}_{\bbar,\bbar'}\!=\!1) =
\#_{b}( \bbar\frac{b}{j} \in D_t \wedge B_{b_jb}\!=\!1 )
\earr
\]
so that the desired bijection can be pieced together 
from corresponding bijections between $D_t$ and $D_t'$.
\eprf

The following should be contrasted with Theorem~\ref{SAkminus1/2prop},
which characterises $\LC^k$-equivalence in terms of $\ISO[k-1/2]$.
That the half-step discrepancies constitute a proper gap 
in shown in Section~\ref{gapsec} below.

\bT
\label{SAkminus1prop}
$\ISO[k-1]$ has a solution if, and only if, $\str A\equiv_\LC^{<k}\str B$.
\eT

\prf
The last lemma settles one implication. For the converse implication, 
it remains to argue that $\LC^{<k}$-equivalence suffices in place of
$\LC^k$-equivalence to provide a solution to the Sherali--Adams
relaxation of level $k-1$. 
We now let $\sktp(\abar)$ stand for the
$\LC^{<k}$-type, or the $\LC^{<k}$-equivalence class of 
the tuple $\abar$. 
We may look at just tuples of length
$k-1$, by trivial padding through repetition of the last component say.
Put
\begin{equation}
\barr{l}
X_\emptyset := 1
\\
X_p := \delta(\sktp(\abar),\sktp(\bbar)) 
\bigm/ \#_\xbar (\sktp(\xbar)  \!=\! 
\sktp(\abar))
\\
\nt\hfill
\mbox{ for $p = \abar\bbar$, $0 < |p| \leq k-1$.}
\earr
\end{equation}

We know that $\LC^{<k}$-equivalence refines $\LC^{k-1}$-equivalence, and
that an assignment to $X_p$ according to $\LC^{k-1}$-types of $(k-1)$-tuples  
was shown above to satisfy the continuity equations~$\CONT{\ell}$
of levels $\ell < k$, cf.\ 
Lemma~\ref{mixedlevellem}.
The same argument applies here to show that the refinement used here satisfies these equations.

For satisfaction of equations~$\COMP{\ell}$ of level 
$\ell < k$, however, we need to appeal to something less than the extension
property that boosts $\abar$ and $\bbar$ to $k$-tuples 
$\abar a \hat{a}$ and $\bbar \hat{b} b$ of the same $\LC^k$-type, as we
used in connection with~(\ref{evallhs}) above. 

Here as there, however, we only need to look at 
$p = \abar \bbar$ of size (up to) $k-2$ for which 
$\str{A},\abar \equiv_\LC^{<k} \str{B},\bbar$, because
all other instances of the equation are trivially true with 
$0$ on both sides. We fix such $p$.  

Now, for any combination of 
$\LC^{<k}$-types $\alpha$ and $\beta$ of $(k-1)$-tuples
and quantifier-free type $\eta$ of a pair, 
\begin{equation}
\label{matchnoeqn}
\barr{rl}
&\#_{aa'}^\str{A} 
\bigl( \sktp (\abar a)  \!=\! \alpha \wedge 
\sktp(\abar a')  \!=\! \beta \wedge  \atp(aa')  \!=\! \eta \bigr)
\\
\hnt
=& 
\#_{bb'}^\str{B} 
\bigl( \sktp (\bbar b)  \!=\! \alpha \wedge 
\sktp(\bbar b')  \!=\! \beta \wedge \atp(bb')  \!=\! \eta \bigr).
\earr
\end{equation}

This follows from an analysis of the $\LC^{<k}$-game from position
$p = \abar\bbar$ through two rounds, in which 
the first player first 
gets the last pebble placed on any one of the possible choices for
$a$, with responses $b$
as provided by the second player's bijection (in exactly the same
number); then the first player plays on that last component again, 
and replaces it with any one of the choices he may have for $a'$ 
and its match $b'$
according to the second player's bijection (again, the same number of
positive choices).

For given $a$ and $b$, 
let now $\alpha := \sktp(\abar a)$
and 
$\beta := \sktp(\bbar b)$.
Then 
\[
\barr{rl}
&
\sum_{a'} A_{aa'}  X_{p\,\widehat{\ }\, a'b} 
\\
=&
\vvhnt
\sum_{a'} A_{aa'} \, \delta(\sktp(\abar a'),\sktp(\bbar b)) 
\bigm/ \#_{\xbar x}(\sktp(\xbar x)  \!=\! \sktp(\bbar b))
\\
=&
\vvhnt
\frac{\displaystyle
\#_{a'}^\str{A} 
\bigl( \sktp(\abar a')  \!=\! \beta \wedge  \mathrm{edge}(aa') \bigr)}%
{\displaystyle 
\#_{\xbar a'}(\sktp(\xbar a')  \!=\! \beta)}
\\
=&
\vvhnt
\frac{\displaystyle
\#_{aa'}^\str{A} 
\bigl( \sktp (\abar a)  \!=\! \alpha \wedge 
\sktp(\abar a')  \!=\! \beta \wedge  \mathrm{edge}(aa') \bigr)}%
{\displaystyle 
\#_{\xbar a'}(\sktp(\xbar a')  \!=\! \beta) \cdot
\#_{a}(\sktp(\abar a)  \!=\! \alpha)}
\\
=&
\vvhnt
\frac{\displaystyle
\#_{aa'}^\str{A} 
\bigl( \sktp (\abar a)  \!=\! \alpha \wedge 
\sktp(\abar a')  \!=\! \beta \wedge  \mathrm{edge}(aa')  \bigr)}%
{\displaystyle 
\#_{a'}(\sktp(\abar a')  \!=\! \beta) 
\cdot
\#_{\xbar}(\sktp(\xbar)  \!=\! \sktp(\abar)) 
\cdot
\#_{a}(\sktp(\abar a)  \!=\! \alpha)}.
\earr
\]
We transform this term further, 
using~(\ref{matchnoeqn}), a renaming of dummy variables in counting
terms and the symmetry of the unique 
quantifier-free type $\eta$ determined by
$\mathrm{edge}(xy)$ in simple
undirected graphs. The goal is to show equality with the 
corresponding term obtained for 
$\sum_{b'}  X_{p\,\widehat{\ }\, ab'} B_{b'b}$.
Equality~(\ref{matchnoeqn}) is used in the first step 
of these transformations, starting from the term just obtained:
\[
\barr{rl}
&
\frac{\displaystyle
\#_{aa'}^\str{A} 
\bigl( \sktp (\abar a)  \!=\! \alpha \wedge 
\sktp(\abar a')  \!=\! \beta \wedge  \mathrm{edge}(aa')  \bigr)}%
{\displaystyle 
\#_{a'}(\sktp(\abar a')  \!=\! \beta) 
\cdot
\#_{\xbar}(\sktp(\xbar)  \!=\! \sktp(\abar)) 
\cdot
\#_{a}(\sktp(\abar a)  \!=\! \alpha)}
\\
\vvhnt
=&
\frac{\displaystyle
\#_{bb'}^\str{B} 
\bigl( \sktp (\bbar b)  \!=\! \alpha \wedge 
\sktp(\bbar b')  \!=\! \beta \wedge  \mathrm{edge}(bb')  \bigr)}%
{\displaystyle 
\#_{b'}(\sktp(\bbar b')  \!=\! \beta) 
\cdot
\#_{\xbar}(\sktp(\xbar)  \!=\! \sktp(\bbar)) 
\cdot
\#_{b}(\sktp(\bbar b)  \!=\! \alpha)}
\\
\vvhnt
=&
\frac{\displaystyle
\#_{bb'}^\str{B} 
\bigl( \sktp (\bbar b')  \!=\! \alpha \wedge 
\sktp(\bbar b)  \!=\! \beta \wedge  \mathrm{edge}(bb')  \bigr)}%
{\displaystyle 
\#_{b}(\sktp(\bbar b)  \!=\! \beta) 
\cdot
\#_{\xbar}(\sktp(\xbar)  \!=\! \sktp(\bbar)) 
\cdot
\#_{b'}(\sktp(\bbar b')  \!=\! \alpha)}
\\
\vvhnt
=&
\frac{\displaystyle
\#_{b'}^\str{B} 
\bigl( \sktp (\bbar b')  \!=\! \alpha  
\wedge  \mathrm{edge}(bb')  \bigr)}%
{\displaystyle 
\#_{\xbar}(\sktp(\xbar)  \!=\! \sktp(\bbar)) 
\cdot
\#_{b'}(\sktp(\bbar b')  \!=\! \alpha)}
\\
\vvhnt
=&
\sum_{b'} 
B_{b'b} \;
\delta(\tp^{<k}(\bbar b'),\tp^{<k}(\abar a)) \bigm/
\#_{\xbar x}(\sktp(\xbar x)  \!=\! \tp^{<k}(\abar a))
\\
\vvhnt
=&
\sum_{b'}
X_{p\,\widehat{\ }\,ab'} B_{b'b}.
\earr
\]
\eprf

\subsection{The gap}
\label{gapsec}

The following theorem shows that for every $k\ge 3$ the level of equivalence captured by the Sherali--Adams relaxation 
of fractional graph isomorphism of level $k-1$, i.e., $\equiv_\LC^{<k}$,
is strictly between $\LC^{k-1}$-equivalence and $\LC^k$-equivalence.

\bT\label{theo:gap}
Let $k\ge 3$.
\begin{enumerate}
\item There are graphs $\str{A}$ and $\str{B}$ such that $\str{A} \equiv_\LC^{k-1} \str{B}$
but $\str{A} \not\equiv_\LC^{<k} \str{B}$.
\item There are graphs $\str{A}$ and $\str{B}$ such that $\str{A} \equiv_\LC^{<k} \str{B}$
but $\str{A} \not\equiv_\LC^{k} \str{B}$.
\end{enumerate}
\eT

We will use the bijective pebble game and weak bijective pebble game to prove the
assertions of the lemma. To be able to deal with the two games more
uniformly, we slightly change the presentation of the 
bijective $k$-pebble game in the following way. 
We now regard as positions sets
of pairs of elements from $\CA,\CB$ of size
$|p|\le k-1$ (instead of size $k$ in the original version). A single round of
the game, starting in position $p$, is played as follows.\label{page:alt-game}
\begin{enumerate}%
\item Player \PII\ selects a bijection $f$ between the $\CA$ and $\CB$.
\item If $|p|=k-1$, player \PI\ selects a pair $ab$ from the current
  position $p$ to be removed.
\item Player \PI\ chooses a pair $a'b'$ from the bijection $f$ to be added.
\item If $p^+:=p\conc a'b'$ is a local isomorphism then the new
  position is 
  \[
  p':=
  \begin{cases}
   \;\; (p\!\setminus\! ab)\,\widehat{\ }\,a'b'&\text{ if }|p|=k-1,\\
\hnt
    \;\;\, p \,\widehat{\ }\,a'b'&\text{ if }|p|<k-1.
  \end{cases}
  \]
  Otherwise, the
  play ends and player \PII\ loses.
\end{enumerate}
Player \PII\ wins a play if it lasts forever. 

It is easy to see that
this new version of the bijective $k$-pebble game is equivalent to the
original version introduced on page~\pageref{bijective-game} (also see
page~\pageref{page:bij-game2}), in the
sense that for all positions $p$ of size at most $k-1$, player \PI\ (and
also player \PII) has a
winning strategy for the new game starting in position $p$ if and only
he has a winning strategy for the original game starting in position
$p$. Essentially, we have just shifted the game by ``half a round'':
instead of starting in a position $p$ of size $k$, removing a pair
from $p$ to obtain an intermediate position $p^-$ of size $k-1$,
then choosing a bijection, then adding a pair to return to a position
of size $k$ and check if this position is a local isomorphism, in the new version we start in a position $p$ of size
$k-1$, choose a bijection, add a pair to $p$ to obtain an intermediate
position $p^+$ of size $k$, check if this is a local isomorphism, and
then remove a pair to return to a position
of size $k-1$. 
When we say ``bijective $k$-pebble game'' in the following, we always
refer to the new version of the game.

This new way of looking at the $k$-pebble game highlights the difference
between the game and its weak version: in the weak bijective
$k$-pebble game, the first two steps are swapped, that is, player \PI\
first picks a pair $ab$ and then player \PII\ selects a
bijection. This is the only difference between the two games.

\medskip
The proof of Theorem~\ref{theo:gap} is based on a well-known 
construction due to Cai, F\"urer, and
Immerman~\cite{caifurimm92}. It will be convenient to discribe the
construction for \emph{multigraphs}, that is, graph that may have
several ``parallel'' edges between the same pair of vertices. We
denote the vertex set of a multigraph $\CG$ by $V(\CG)$ and the edge
set by $E(\CG)$. When we write $e=vw$, this merely indicates
that $e$ is an edge incident with $v$ and $w$;
there may be other edges 
linking this pair of vertices.
For every vertex $v\in V(\CG)$, by
$E(v)$ we denote the set of all edges incident with $v$.
For every multigraph $\CG$ with vertex set
$V(\CG)=\{v_1,\ldots,v_n\}$, we construct two structures $\CX(\CG)$
and $\hCX(\CG)$, which we call the \emph{CFI-companions} of $\CG$, as
follows.\footnote{The graph $\hCX(\CG)$ will not only depend on $\CG$,
  but also on the enumeration of its vertices, or rather, just on the
  choice of a \emph{first vertex.} We choose not to highlight this
  dependence notationally.  } 
Both $\CX(\CG)$ and $\widehat{\CX}(\CG)$
are coloured graphs whose vertices are coloured with $n+|E(G)|$ distinct
colours $C_v$ for $v\in V(\CG)$ and $C_e$ for $e\in E(\CG)$.
It will be convenient to call the vertices of the graphs $\CX(\CG)$
and $\hCX(\CG)$ \emph{nodes}, to distinguish them from the
\emph{vertices} of $\CG$. 

\medskip
For every $v\in V(\CG)$, the graph
  $\CX(\CG)$ has nodes $v^S$, where $S$ is a subset of $E(v)$ of
  even cardinality. For every edge $e$ of $\CG$, the graph
  $\CX(\CG)$ has nodes $e^0$, $e^1$. 
The node set of $\hCX(\CG)$ is the same, except that for the
  vertex $v_1$ we take nodes $v_i^S$ for the subsets of $E(v_1)$ of
  odd cardinality. 
  
Nodes of the form $v^S$ are called
\emph{vertex nodes} and nodes $e^i$ \emph{edge nodes}.

\medskip\noindent
--- The set of nodes $\CX(\CG)$ is
  \begin{align*}
  & \{v_i^S\mid i\in[n], S\subseteq E(v_i)\text{ such that
  }|S|\equiv0\text{ mod }2\}\\
  \cup\; &
  \{e^0,e^1\mid e\in E(\CG)\}
  \end{align*}

\medskip\noindent
--- The node set of $\hCX(\CG)$ is
  \begin{align*}
    &\{v_1^S\mid S\subseteq E(v_1)\text{ such that
    }|S|\equiv1\text{ mod }2\}\\
    \cup\;&\{v_i^S\mid i\in[n]\setminus\{1\}, S\subseteq E(v_i)\text{ such that
    }|S|\equiv0\text{ mod }2\}\\
    \cup\;&
    \{e^0,e^1\mid e\in E(\CG)\}
  \end{align*}

\medskip\noindent
--- The edges of both $\CX(\CG)$ and $\hCX(\CG)$ link vertex nodes
$v^S$ to edge nodes $e^i$ for edges $e \in E(v)$ according to 
\[
\{v^S,e^i\} \mbox{ is an edge if }
\left\{ 
\barr{l}
i = 1 \mbox{ and } e \in S 
\\
i = 0 \mbox{ and } e \not\in S 
\earr
\right.
\]
--- In both $\CX(\CG)$ and $\hCX(\CG)$, the nodes of the form
$v^S$ are coloured $C_v$, and the nodes of the form $e^i$ are coloured $C_e$.

\medskip
Local isomorphisms between the coloured graphs $\CX(\CG)$ and
$\hCX(\CG)$ are required to preserve the colours. Thus in the
bijective $k$-pebble game or the weak bijective $k$-pebble game on
$\CX(\CG)$ and $\hCX(\CG)$, player \PII\ has to preserve colours and
is thus forced to make sure that a node $v^S$ is always mapped to a node
$v^{S'}$ for some $S'$ and that a node $e^i$ is always mapped to a
node $e^{i'}$.\footnote{Colours can be eliminated by attaching gadgets
encoding the colours (such as paths of different lengths) to the
nodes.} 

Let $g$ be a local bijection or bijection from $\CX(\CG)$ to
$\hCX(\CG)$. We say that $g$ is \emph{colour-preserving} if for all
vertices $v$ of $\CG$ and all $S$ we have $g(v^S)=v^T$ for some $T$,
and for all edges $e$ of $\CG$ and all $i$ we have $g(e^i)=e^j$ for
some $j$. In the following, suppose that $g$ is colour-preserving.
Slightly abusing terminology, we say that a vertex $v$ of $\CG$ is in
the domain of $g$ if $v^S$ is in the domain of $g$ for some
$S$. Similarly, we say that an edge $e$ of $\CG$ is in the domain of
$G$ if $e^i$ is in the domain of $g$ for some $i$.  

We say that $g$ is
\emph{vertex-consistent} if for all vertices $v$ of $\CG$ and all
$S,S',T,T'$ such that both $v^S$ and $v^{S'}$ are in the domain of $g$
and $g(v^S)=v^T$, $g(v^{S'})=v^{T'}$ we have $S\triangle T=S'\triangle
T'$. If $g$ is vertex-consistent, then for all $v$ in the domain of
$g$ we let $g_v:=S\triangle T$ for some (and hence all) $S,T$ such
that $v^S$ is in the domain of $g$ and $g(v^S)=v^T$. Note that
$g_v$ determines $g(v^S)$ for all $S$, and that a position 
which fails to be vertex consistent would allow player
\PI\ to win in one round played with the help of any pebble
pair to spare.

Note that $g$,
being a colour-preserving (local) bijection, is automatically
\emph{edge-consistent} in the sense that for all edges $e$ of $\CG$,
if $g(e^i)=e^j$ and $g(e^{i'})=e^{j'}$ then $i+j\equiv i'+j'\text{ mod
}2$. In other words, $i=j$ iff $i'=j'$. This is clearly necessary for
a colour-preserving bijection, which must map the colour $C_e = \{ e^0,e^1\}$ 
in $\CX(\CG)$ to colour $C_e = \{ e^0,e^1\}$ in ${\hCX(\CG)}$.
We let $g_e\in\{0,1\}$ such that $g_e\equiv i+j\text{ mod }2$ for
some (and hence for all) $i,j$ with $g(e^i)=e^j$. If $g_e=1$, then we
say that $g$ \emph{flips} edge $e$. 

We say that a (local) bijection
$g$ is \emph{weakly consistent} if it is colour-preserving (and thus
edge consistent),
vertex-consistent, and for all vertices $v$ and edges $e$ in the
domain of $g$ we have $e\in g_v\IFF g_e=1$. Note that if $g$ is 
weakly consistent then it is a (local) isomorphism. 

We say that $g$ is
\emph{strongly consistent} if it is weakly consistent and, in
addition, for all edges $e=vw$ of $\CG$, if both $v$ and $w$ are in
the domain of $g$ then $e\in g_v\IFF e\in g_w$. 

We say that a bijection $h$ \emph{consistently extends} $g$, or is a
\emph{consistent extension} of $g$, if it satisfies the following
conditions \ref{gap:A}--\ref{gap:D}.
\begin{ealph}
\item\label{gap:A} $h$ is colour-preserving and vertex-consistent.
\item\label{gap:B} $h$ is an extension of $g$.
\end{ealph}
Note that if $h$ satisfies \ref{gap:A} and \ref{gap:B}, then $g$ must
be colour preserving and node consistent. Conversely, if $g$ is colour preserving and node consistent, then
condition \ref{gap:B} is equivalent to the condition that for all vertices $v$
and edges $e$ of $\CG$ in the domain of $g$ we have $g_v=h_v$ and
$g_e=h_e$.
\begin{ealph}[resume]
\item\label{gap:C} For all vertices $v$ of $\CG$, if some edge $e$ incident with
  $v$ is in the domain of $g$, then $e\in h_v\IFF
  g_e=1$.
\item\label{gap:D} For all edges $e$ of $\CG$, if some vertex $v$ incident with
  $e$ is in the domain of $g$, then $h_e=1\IFF
  e\in g_v$.  
\end{ealph}
Note that for $g$ to have a consistent extension, it must be strongly
consistent. However, even if $g$ is strongly consistent it does not
necessarily have a consistent extension, because $g$ is only a local
bijection, whereas $h$ is a total bijection.

Note that player \PI\ can directly win the game from any position that fails to
be strongly consistent, or whenever \PII\ proposes a bijection 
that fails to be a consistent extension of the current position.
We may therefore, without loss of generality, restrict \PII\ to 
strongly consistent positions and consistent extensions.

\begin{proof}[Proof of Theorem~\ref{theo:gap}~(1)]
Let $\CK$ be the complete graph on $k$ vertices. We fix some
enumeration $v_1,\ldots,v_{k}$ of the vertex set of $\CK$. 
For all $i\neq j$, we let $e_{ij}$ be the edge between $v_i$ and $v_j$. 

We let $\CA=\CX(\CK)$ and $\CB=\hCX(\CK)$ and show that 
$\CX(\CK)  \equiv_\LC^{k-1} \hCX(\CK)$, but 
$\CX(\CK) \not\equiv_\LC^{<k} \hCX(\CK)$.

To prove that $\CX(\CK) \not\equiv_\LC^{<k} \hCX(\CK)$,
we give a winning strategy
for player $\PI$ in the weak bijective $k$-pebble game
on $\CX(\CK),\hCX(\CK)$. In the first $k-1$ rounds of the game, player 
$\PI$ can reach a position $p$ with domain 
$v_2^\emptyset,\ldots,v_{k}^\emptyset$
and $p(v_i^\emptyset)=v_i^{S_i}$ for some sets $S_i$. That is,
\[
p=\{v_2^\emptyset v_2^{S_2},\ldots,v_k^\emptyset v_k^{S_k}\}.
\] 
Note that $p_{v_i} = \emptyset \triangle S_i = S_i$ for $2\le i\le
k$. For all $i$ and all edges $e$ of $\CK$ we let $\epsilon(i,e)=1$ if $e\in S_i$
and $\epsilon(i,e)=0$ otherwise. In particular, $\epsilon(i,e)=0$ if
$v_i$ is not incident with 
the edge $e$. 
Without loss of generality, 
$p$ is strongly consistent, and thus for $2\le i<j\le k$ we have
$\epsilon(i,e_{ij})=\epsilon(j,e_{ij})$. Moreover, by the construction of
$\hCX(\CK_{k+1})$, all the sets $S_i$ have even
cardinality. Thus
\[
0\equiv \sum_{i=2}^{k}|S_i|\equiv \sum_{i=2}^{k}
\sum_{e}
\epsilon(i,e)\equiv
\sum_{i=2}^{k}\epsilon(i,e_{1i})\text{ mod }2.
\]
Let $T$ be the set of all edges $e_{1i}$ with
$e_{1i}\in S_i$. Then $|T|$ is even. To simplify the notation, we let
$\epsilon_i=\epsilon(i,e_{1i})$ in the following.

In the next round of the game, player \PI\ starts by selecting the
pair $v_2^\emptyset v_2^{S_2}$. Let $f$ be the bijection selected by
player \PII. Without loss of generality, $f$ is a consistent extension of $p$. Thus by
\ref{gap:D}, $f$ flips edge $e_{12}$ if and only if $e_{12}\in p_{v_2}=S_2$.  That is,
$f(e_{12}^0)=e_{12}^{\epsilon_2}$. Player \PI\ selects the pair
$e_{12}^0e_{12}^{\epsilon_2}$, and the new position is
\[
p'=\{e_{12}^0e_{12}^{\epsilon_2},v_3^\emptyset
v_3^{S_3},\ldots,v_{k}^\emptyset v_{k}^{S_{k}}\}.
\]
Note that $p'_{e_{12}} = \epsilon_2$
In the next round of the game, player \PI\ starts by selecting the
pair $e_{12}^0e_{12}^{\epsilon_2}$. Let $f'$ be the bijection selected by
player \PII. Without loss of generality, $f'$ is a consistent extension of $p'$. Let $T'\subseteq E(v_1)$ such that
$f'(v_1^\emptyset)=v_1^{T'}$
and $f'_{v_1} = T'$. 
By \ref{gap:C}, 
\begin{equation}\label{eq:gap1}
e_{12}\in f_{v_1}' =T'\IFF \epsilon_2=1\IFF e_{12}\in T.
\end{equation}
 Player
\PI\ selects the pair $v_1^\emptyset v_1^{T'}$, and the new position
is
\[
p''=\{v_1^\emptyset v_1^{T'} ,v_3^\emptyset
v_3^{S_3},\ldots,v_{k}^\emptyset v_{k}^{S_{k}}\}.
\]

Now \PI\ wins as $p''$ is not strongly consistent, 
which can be shown indirectly as follows.
Suppose for contradiction that $p''$ were strongly consistent. 
Then, for $3\le i\le k$, we would have
$e_{1i}\in T'\IFF e_{1i}\in S_i\IFF e_{1i}\in T$. Combined with
\eqref{eq:gap1}, this implies $T'=T$. However, $|T'|$ is odd by the
construction of $\hCX(\CK)$, whereas $|T|$ is even. 

\medskip
To prove that $\CX(\CK)  \equiv_\LC^{k-1} \hCX(\CK)$,
we give a winning strategy
for player $\PII$ in the variant of the 
bijective $(k-1)$-pebble game on $\CX(\CK), \hCX(\CK)$.
It suffices to show that in every strongly consistent
position of the game,
\PII\ can maintain strong consistency. 
So let $p$ be a strongly consistent position of size
$|p|\le k-2$. To define a consistent
extension $f$ of $p$, it suffices to specify
$f_v$ for all vertices $v$ and $f_e$ for
all edges $e$ of $\CK$.
\begin{enumerate}%
\item We start by letting $f_v=p_v$ for all vertices $v$ in the domain
  of $p$ and $f_e=p_e$ for all edges $e$ in the domain of $p$.
\item For all edges $e$ of $\CK$ that 
are incident with at least one vertex 
$v$ in the domain of $p$, we let $f_e=1$ if $e\in p_v$ and
  $f_e=0$ otherwise. We can do this consistently because $p$ is strongly consistent.
\item For all remaining vertices $v$ of $\CK$, we note that
  there is at least one edge $e=vw$ incident with $v$ such that
  neither $w$ nor $e$ are in the domain of $p$ (and hence $f_e$ has
  not been defined yet). We choose a subset $S\subseteq E(v)$
such that for all edges $e'=vw'\in E(v)$ such that either $e'$ or $w'$
  is in the domain of $p$ we have $e'\in S\IFF f_{e'}=1$. Moreover, we choose such an
  $S$ such that its cardinality is odd if $v=v_1$ and its cardinality
  is even otherwise. We have the freedom to choose the parity in this manner
  because we can add $e$ to $S$ without affecting the property  $e'\in
  S\IFF f_{e'}=1$ for all edges $e'=vw'\in E(v)$ such that either $e'$ or $w'$
  is in the domain of $p$.

  We let $f_v=S$.
\item Finally, for all edges $e$ for which $f_e$ has not yet been
  defined we let $f_e=0$.
\end{enumerate}
In the next round of the game, \PII\ selects
$f$. Suppose that \PI\ selects $ab\in p$ to be removed and $a'b'\in f$
to be added. It is easy to prove
that $p^+=p\conc a'b'$ is a local isomorphism and that the new position $p'=(p\setminus ab)\conc a'b'$
is strongly consistent. The proof is by case distinction along the
cases of the definition of the bijection $f$. It may seem that edges
$e$ for which $f_e$ is defined in (iv) will cause problems, because
for these edges the definition of $f_e$ does not depend on the current
position $p$ at all. However, in the new position $p'$ and even in the
intermediate position $p^+$ such edges will not
be incident with any vertex in the domain, so they will not affect
strong consistency.
\end{proof}

\begin{figure}
  \centering
  \begin{tikzpicture}
    [
    vertex/.style={draw,fill,circle,minimum height=1.5mm,inner
      sep=0mm},
    ]
    \small
    
    \begin{scope}[yshift=4.6cm,x=1.8cm,thick]
    \node[vertex,minimum height=2mm,red] (xv) at (-1,0) {};
    \node[vertex,minimum height=2mm,blue] (xw) at (1,0) {};
    \draw (-2,0) -- (xv) (xw) -- (2,0);
    \draw (xv)  edge[bend left] (xw) edge[bend right] (xw);
    \path (-1.9,0.2) node {$f_{e,v}$} (-1,0.3) node {$x_{e,v}$} (0,0.8) node
    {$f_{e,1}$} (0,-0.8) node  {$f_{e,2}$} (1.9,0.2) node {$f_{e,w}$} (1,0.3) node
    {$x_{e,w}$};

    \path (0,-1.5) node {\normalsize(a)};
    \end{scope}

    \begin{scope}[x=1.8cm,y=1.5cm]
      \node[vertex,red] (xv1) at (-1,1.2) {};
      \node[vertex,red] (xv2) at (-1,0.4) {};
      \node[vertex,red] (xv3) at (-1,-0.4) {};
      \node[vertex,red] (xv4) at (-1,-1.2) {};
      \node[vertex,blue] (xw1) at (1,1.2) {};
      \node[vertex,blue] (xw2) at (1,0.4) {};
      \node[vertex,blue] (xw3) at (1,-0.4) {};
      \node[vertex,blue] (xw4) at (1,-1.2) {};

      \node[vertex] (fv0) at (-2,0.4) {};
      \node[vertex] (fv1) at (-2,-0.4) {};
      \node[vertex] (f10) at (0,1.2) {};
      \node[vertex] (f11) at (0,0.4) {};
      \node[vertex] (f20) at (0,-0.4) {};
      \node[vertex] (f21) at (0,-1.2) {};
      \node[vertex] (fw0) at (2,0.4) {};
      \node[vertex] (fw1) at (2,-0.4) {};

      \draw (xv1) edge (fv0) edge (f10) edge (f20)
                (xv2) edge (fv0) edge (f11) edge (f21)
                (xv3) edge (fv1) edge (f10) edge (f21)
                (xv4) edge (fv1) edge (f11) edge (f20)
                (xw1) edge (fw0) edge (f10) edge (f20)
                (xw2) edge (fw0) edge (f11) edge (f21)
                (xw3) edge (fw1) edge (f10) edge (f21)
                (xw4) edge (fw1) edge (f11) edge (f20)
                ;

       \path (-2.1,0.6) node {$f_{e,v}^0$} (-2.1,-0.6) node {$f_{e,v}^1$}
                  (-1,1.4) node {$x_{e,v}^\emptyset$} (-1.1,0.6) node
                  {$x_{e,v}^{\{f_{e,1},f_{e,2}\}}$} 
                  (-1.2,-0.2) node {$x_{e,v}^{\{f_{e,v},f_{e,2}\}}$}  (-1,-1.4) node
                  {$x_{e,v}^{\{f_{e,v},f_{e,1}\}}$} 
                  (0,1.4) node {$f_{e,1}^0$} (0,0.6) node {$f_{e,1}^1$}
                  (0,-1.45) node {$f_{e,2}^1$} (0,-0.65) node {$f_{e,2}^0$}
                  (2.1,0.6) node {$f_{e,w}^0$} (2.1,-0.65) node {$f_{e,w}^1$}
                  (1,1.4) node {$x_{e,w}^\emptyset$} (1.1,0.6) node
                  {$x_{e,w}^{\{f_{e,1},f_{e,2}\}}$} 
                  (1.2,-0.65) node {$x_{e,w}^{\{f_{e,w},f_{e,2}\}}$}  (1,-1.4) node
                  {$x_{e,w}^{\{f_{e,w},f_{e,1}\}}$} 
       ;
    \path (0,-2) node {\normalsize(b)};

    \end{scope}
  \end{tikzpicture}
  \caption{Threshold gadget $\CT_e$}
  \label{fig:threshold}
\end{figure}

The proof of Theorem~\ref{theo:gap}~(2) requires more
preparation. Essentially, we will also play games on the
CFI-companions of the complete graph $\CK$, but we will insert
certain ``threshold gadgets'' on the edges that require at least two
pebbles to transport the information of whether an edge is flipped or
not from one end of the edge to the other. The gadget is displayed in
Figure~\ref{fig:threshold}(b); the name $\CT_e$ of the gadget and the
names of  the vertices indicate that
the gadget is intended to be inserted for an edge $e=vw$. Observe that
the gadget is a CFI-companion of the multigraph, or rather: fragment of a
multigraph, displayed in Figure~\ref{fig:threshold}(a).  As for all CFI-companions, the nodes $x_{e,v}^S$ are
coloured by a fresh colour, and so are the nodes $x_{e,w}^S$ as well
as the edge-nodes $f_{e,v}^i$, $f_{e,1,}^i$, $f_{e,2}^i$, $f_{e,w}^i$. 
The idea is to
replace each edge $e=vw$ of some graph by the multigraph 
from Figure~\ref{fig:threshold}(a) by connecting
$f_{e,v}$ to $v$ and $f_{e,w}$ to $w$ and then go to the CFI-companion of the
new graph.
The crucial property of the gadget is that player \PI\ needs two
pairs of pebbles to transport information from one end of the gadget to the
other end. To make this precise, we introduce this
terminology: in the (weak) bijective $k$-pebble game on structures $\CA$
and $\CB$, we say that player \PI\ can \emph{reach} position $p'$ from
position $p$ if he has a strategy for the game starting in position
$p$ such that in each play that he plays according to this strategy,
either he wins or a position $p''\supseteq p'$ occurs. 
If \PI\ cannot reach position
$p'$ from position $p$, we say that \PII\ can \emph{avoid} position
$p'$.
As usual, a position $p$ of the game is a \emph{winning position} for player
\PI\ or \PII\ if the respective player has a winning strategy for the
game starting in that position.

\bL\label{lem:threshold}
  \begin{enumerate}
  \item 
For the weak bijective
    $3$-pebble game on $\CT_e,\CT_e$:
\\
any position $\{f_{e,v}^if_{e,v}^j,f_{e,w}^{i'}f_{e,w}^{j'}\}$
such that $i+j\not\equiv i'+j'\text{
      mod }2$  
    is a winning position for player \PI.
  \item For the bijective $2$-pebble game on $\CT_e,\CT_e$: 
\\
let $i,j\in\{0,1\}$ and let $p=\{ab\}$ be a 
vertex-consistent position such that $x_{e,v}$ and 
$f_{e,v}$ are not in the domain of $p$ (that is, $a,b$ are either of the
    form $f_{e,w}^i$ or $x_{e,w}^S$ or $f_{e,j}^i$); then 
player \PII\ can avoid position $\{f_{e,v}^{i}f_{e,v}^{j}\}$
from position $p$.
  \end{enumerate}
\eL

Note that assertion~(1) implies that \PI\ can
reach position $\{f_{e,w}^{i}f_{e,w}^{j}\}$ from position
$\{f_{e,v}^if_{e,v}^j\}$ in the weak bijective $3$-pebble game
    on $\CT_e,\CT_e$. 
This is because, if \PI\ selects $f_{e,w}^{i}$ in the first round of
the game starting in position $\{f_{e,v}^if_{e,v}^j\}$, then \PII\ has to
answer with $f_{e,w}^{j}$; otherwise the position is
$\{f_{e,v}^if_{e,v}^j,f_{e,w}^{i}f_{e,w}^{j'}\}$ such that
$i+j\not\equiv i+j'\text{ mod }2$, and \PII\ loses by~(1).

\begin{proof}[Proof of Lemma~\ref{lem:threshold}]
  To prove~(1), we give a winning strategy for \PI\ for the game
  starting in position $p = \{f_{e,v}^if_{e,v}^j,f_{e,w}^{i'}f_{e,w}^{j'}\}$. In the first round,
  \PI\ first 
  selects the pair $f_{e,v}^if_{e,v}^j$. Suppose that
  \PII\ answers by selecting the bijection $g$. Without loss of
  generality we may assume that $g$ is a consistent extension of $p$. 
  Let $T\subseteq
  E(x_{e,v})=\{f_{e,v},f_{e,1},f_{e,2}\}$ such that $g(x_{e,v}^\emptyset)=x_{e,v}^T$. Then
  $f_{e,v}\in T\IFF i+j\equiv1\text{ mod }2$, because $g$ is a consistent
  extension of $p$. Player \PI\ selects the pair $x_{e,v}^\emptyset
  x_{e,v}^T$, and the new position is
  \[
  p'=\{x_{e,v}^\emptyset
  x_{e,v}^T, f_{e,w}^{i'}f_{e,w}^{j'}\}.
  \]
  In the next round, \PI\ selects the pair $f_{e,w}^{i'}f_{e,w}^{j'}$ in the first step. Suppose that
  \PII\ answers by selecting the bijection $g'$. Let $T'\subseteq
  E(x_{e,w})$ be such that $g(x_{e,w}^\emptyset)=x_{e,w}^{T'}$. Then
  $f_{e,w}\in T'\IFF i'+j'\equiv1\text{ mod }2$, because $g'$ is a consistent
  extension of $p$. Player \PI\ selects the pair $x_{e,w}^\emptyset
  x_{e,w}^{T'}$, and the new position is
  \[
  p'=\{x_{e,v}^\emptyset
  x_{e,v}^T, x_{e,w}^\emptyset
  x_{e,w}^{T'}\}.
  \]
Recall that $i+j\not\equiv i'+j'\text{
      mod }2$. By symmetry, we may
    assume that $i+j\equiv0\text{ mod }2$ and $i'+j'\equiv1\text{ mod
    }2$. Then $f_{e,v}\not\in T$, and as $T$ is even, it follows that
    either $T=\emptyset$ or $T=\{f_{e,1},f_{e,2}\}$. Similarly, either
    $T'=\{f_{e,w},f_{e,1}\}$ or $T'=\{f_{e,w},f_{e,2}\}$. Thus there is a
    $j\in\{1,2\}$ such that $f_{e,j}\in T\triangle T'$. Hence the position
    $p'$ is not strongly consistent, and therefore Player \PI\ wins
    the game.

    To prove (2), we give a strategy for \PII\ that avoids position $\{f_{e,v}^{i}f_{e,v}^{j}\}$
    from position $p$ in the bijective $2$-pebble game on
    $\CT_e,\CT_e$. A position $p'$ is \emph{good} if it is vertex-consistent
    and satisfies the following two conditions.
    \begin{ealph}[resume]
      \item\label{gap:E}
        If $x_{e,v}$ is in the domain of $p'$, then $f_{e,v}\in p_{x_{e,v}}\IFF
        i+j\equiv0\text{ mod }2$.
      \item\label{gap:F}
        If $f_{e,v}$ is in the domain of $p'$, then $p_{f_{e,v}}\not\equiv
        i+j\text{ mod }2$.
    \end{ealph}
    Recall that positions in the $2$-pebble game have size $1$ and
    thus are strongly consistent if and only if they are
    vertex-consistent. 
    Note that the initial position $p$ is good.
    It is easy to prove (by an extensive case analysis) that in any
    good position, \PII\ can play the next round of the game in such a
    way that the position after the round is good again.
\end{proof}

\begin{proof}[Proof of Theorem~\ref{theo:gap}~(2)]
  Let $\CH$ be the multigraph obtained from the complete $k$-vertex graph $\CK$ by
  replacing every edge $e=vw$ by the multigraph displayed in
  Figure~\ref{fig:threshold}(a), where edge $f_{e,v}$ is connected to $v$
  and edge $f_{e,w}$ is connected to $w$. Note that
  \begin{align*}
  V(\CH)&=V(\CK)\cup\{x_{e,v},x_{e,w}\mid e=vw\in E(\CK)\},\\
  E(\CH)&=\{f_{e,v},f_{e,w},f_{e,1},f_{e,2}\mid e = vw \in E(\CK)\}.
  \end{align*}
To define $\hCX(\CH)$, we fix some enumeration of $V(\CH)$ where the
vertex $v_1\in V(\CK)$ comes first. (Again we assume
that $V(\CK)=\{v_1,\ldots,v_k\}$.)

We let $\CA=\CX(\CH)$ and $\CB=\hCX(\CH)$. Note that $\CA$ 
is obtained from $\CX(\CK)$ by replacing, 
for every edge $e\in E(\CK)$, the
vertices $e^0,e^1$ by a threshold gadget $\CT_e$, and $\CB$ is
similarly obtained from $\hX(\CK)$.

We first prove that player \PI\ has a winning strategy for the
bijective $k$-pebble game on $\CA,\CB$. 

Let us call an edge $e=vw\in E(\CK)$ \emph{inconsistent} in a position
$p$ of the game if both $v$ and $w$ are in the domain of $p$
and $f_{e,v}\in p_v\not\Leftrightarrow f_{e,w}\in p_w$. Note that if
some edge $e$ is
inconsistent in a position $p$ then \PI\ can reach position
$\{f_{e,v}^0f_{e,v}^1,f_{e,w}^0f_{e,w}^0\}$ (if $f_{e,v}\in p_v$ and $f_{e,w}\not\in
p_w$) or position $\{f_{e,v}^0f_{e,v}^0,f_{e,w}^0f_{e,w}^1\}$ (if $f_{e,v}\not\in p_v$ and $f_{e,w}\in
p_w$). Then by Lemma~\ref{lem:threshold}~(1), \PI\ wins the game. Thus
player \PII\ needs to avoid inconsistent edges.

The winning strategy for player \PI\ in the
bijective $k$-pebble game on $\CA,\CB$  is as follows.
In the first $k-1$ rounds of the game he
reaches a position $p$ with domain $v_2^\emptyset,\ldots,v_{k}^\emptyset$
and $p(v_i^\emptyset)=v_i^{S_i}$ for some sets $S_i$. That is,
\[
p=\{v_2^\emptyset v_2^{S_2},\ldots,v_k^\emptyset v_k^{S_k}\}.
\] 
For all $i\ge 2$ and all edges $e=v_iw\in E(\CK)$ 
we let $\epsilon(i,e)=1$ if $f_{e,v_i}\in S_i$
and $\epsilon(i,e)=0$ otherwise. For edges $e$ that 
are not incident with $v_i$ 
we let $\epsilon(i,e)=0$. If there is some edge
$e_{ij}$ of $\CK$ such that
$\epsilon(i,e_{ij})\neq\epsilon(j,e_{ij})$, then $e_{ij}$ is
inconsistent, and player \PI\ wins the game.
Moreover, all the sets $S_i$ have even
cardinality. Thus
\[
0 \equiv \sum_{i=2}^{k}|S_i| \equiv \sum_{i=2}^{k}
\sum_{e}
\epsilon(i,e) \equiv \sum_{i=2}^{k}\epsilon(i,e_{1i})
\mod2.
\]
Thus the set of all edges $e_{1i}$ with
$f_{e_{1i},v_i}\in S_i$ is even. 

Let $g$ be the bijection selected by player \PII\ in the next round of
the game, and let $T\subseteq E(v_1)$ such that
$g(v_1^\emptyset)=v_1^T$. 
Then $T$ is odd, and thus there is some
$i\in\{2,\ldots,k\}$ such that for the edge $e:=e_{1i}=v_1v_i$ we have $f_{e,v_1}\in T\Leftrightarrow
f_{e,v_i}\not\in S_i$. Player \PI\ selects some pair $v_j^\emptyset
v_j^{S_j}$ for $j\neq i$ to remove and the pair $v_1^\emptyset
v_1^{T}$ to add. The new position $p'$ contains the pairs $v_1^\emptyset
v_1^{T}$ and $v_i^\emptyset
v_i^{S_i}$. Thus the edge $e_{1i}$ is inconsistent in this position,
and player \PI\ wins.

\bigskip
Let us now prove that player \PII\ has a winning strategy for the weak
bijective $k$-pebble game on $\CA,\CB$. For a vertex-consistent
position $p$ of the game and an edge $e=vw\in E(\CK)$, we let
$p\restriction\CT_e$ be the set $\{ab\in p\mid a,b\in V(\CT_e)\}$
together with the following pairs:
\begin{itemize}
\item $f_{e,v}^0f_{e,v}^0$ if $v$ is in the domain of $p$ and $f_v\not\in
  p_v$,
\item $f_{e,v}^0f_{e,v}^1$ if $v$ is in the domain of $p$ and $f_v\in
  p_v$,
\item $f_{e,w}^0f_{e,w}^0$ if $w$ is in the domain of $p$ and $f_w\not\in
  p_w$,
\item $f_{e,w}^0f_{e,w}^1$ if $w$ is in the domain of $p$ and $f_w\in
  p_w$.
\end{itemize}
We view $p\restriction\CT_e$ as a position of the weak bijective
$k$-pebble game on $\CT_e,\CT_e$. 

We say that edge $e=vw$ is \emph{flipped} in position $p$ if
$p\restriction\CT_e$ is nonempty and \PI\ can reach $\{f_{e,v}^0f_{e,v}^1\}$
from $p\restriction\CT_e$ in the weak bijective $k$-pebble game on
$\CT_e,\CT_e$. The edge $e$ is \emph{straight} in position $p$ if
$p\restriction\CT_e$ is nonempty and \PI\ can reach $\{f_{e,v}^0f_{e,v}^0\}$
from $p\restriction\CT_e$ in the weak bijective $k$-pebble game on
$\CT_e,\CT_e$. Observe that, as discussed in connection with 
Lemma~\ref{lem:threshold}~(1), 
if $e$ is flipped, then \PI\ can reach
$\{f_{e,w}^0f_{e,w}^1\}$ as well, and similarly, if $e$ is straight, then \PI\ can reach
$\{f_{e,w}^0f_{e,w}^0\}$. Moreover, if $p\restriction\CT_e$ is nonempty then
$e$ is either straight or flipped.

We say that a vertex $v\in V(\CK)$ is \emph{trapped} in a
vertex-consistent position $p$ if it satisfies the following
conditions:
\begin{ealph}[start=7]
\item\label{gap:G} for all edges $e=vw\in E(\CK)$
  incident with $v$, the position $p\restriction\CT_e$ is nonempty.
\item\label{gap:H} the set $S$ of all $f_{e,v}$ such that $e$ is flipped in
  position $p$ has the wrong parity, that is, $|S|$ is even if $v=v_1$
  and $|S|$ is odd otherwise.
\item\label{gap:I} there is an edge $e=vw\in E(\CK)$ such that either $f_{e,v}$ or
  $x_{e,v}$ is in the domain of $p$.
\end{ealph}

Observe that if some vertex is trapped in a
position $p$ of the weak bijective $k$-pebble game, then this position
is a winning position for player \PI. (We will not use this
observation and hence we omit a proof.)
Let us call position $p$ \emph{good} if it satisfies the
following conditions:
\begin{ealph}[resume]
  \item\label{a:good1} $p$ is strongly consistent;
  \item\label{a:good2} for all $e\in E(\CK)$, player \PII\ has a winning strategy for
    the weak bijective $k$-pebble game on $\CT_e,\CT_e$ starting in
    position $p\restriction\CT_e$;
  \item\label{a:good2} no vertex $v\in\CK$ is trapped in position $p$.
\end{ealph}
Note that, by Lemma~\ref{lem:threshold}~(1), if $p$ is good, then for
every edge $e=vw\in E(\CK)$, if both $v$ and $w$ are in the domain of
$p$ then $f_{e,v}\in p_v\IFF f_{e,w}\in p_w$.

We claim that in any good position, \PII\ can play the next round of
the game in such a way that the position after the round is good
again.
We have to define a bijection $g$ for player
\PII, which is done in \ref{gap:M}--\ref{gap:R} below. Since the node
sets of $\CA$ and $\CB$ are
the disjoint unions of the sets of vertex nodes $v^S$ 
for $v \in V(\CK)$ and the node sets of the $T_{e}$ for $e \in
E(\CK)$, it
suffices to define $g_v$ for every $v\in V(\CK)$, which determines
$g(v^S)$ for each $S$,  
and the restriction
of $g$ to $\CT_e$ for every $e\in E(\CK)$.

Let $p$ be the given good position. If $|p|<k-1$, let
$p^-=p$, and if $|p|=k-1$, let $ab\in p$ be the pair selected by
player \PI\ in the first step of the next round, and let
$p^-=p\setminus ab$.

\begin{ealph}[resume]
\item\label{gap:M} For every $v\in V(\CK)$ such that $v$ is in the domain of
  $p$ we let $g_v=p_v$.
\end{ealph}
An edge $e=vw\in E(\CK)$ \emph{requires attention at $v$} if $v$ is
not in the domain of $p$ and for all edges 
$e \not= e'$ that are incident with $v$, the restriction
$p^-\restriction\CT_{e'}$ is nonempty. 
Note that there is at most one pair $(e,v)$ such that $e$ requires attention 
at $v$
and that
$|p\restriction\CT_e|\le 1$ if $e$ requires attention. This follows
from the fact that $|p^-|\le k-2$ and the degree of all vertices of
$\CK$ is $k-1$.

Suppose that $e=vw$ requires attention at $v$.  Let $S'$ be the set of
all $f_{e',v}$ such that $e'$ is flipped in $p^-$. If $v\neq v_1$ and
$|S'|$ is odd or $v=v_1$ and $|S'|$ is even, we let
$S=S'\cup\{f_{e,v}\}$; otherwise, we let $S=S'$. The set $S$
determines whether $e$ must be flipped or not. Without loss of
generality, let us assume that $f_{e,v}\in S$, that is, $e$ must be
flipped.  If neither $f_{e,v}$ nor $x_{e,v}$ are in the domain of $p$,
then by Lemma~\ref{lem:threshold}~(2) player \PII\ can avoid position $\{f_{e,v}^0f_{e,v}^0\}$ in the
bijective $2$-pebble game on $\CT_e,\CT_e$ starting in position
$p\restriction\CT_e$. If $f_{e,v}$ is in the domain of $p$ then
$p_{f_v}=1\IFF f_{e,v}\in S$ and if $x_{e,v}$ is in the domain of $p$
then $f_{e,v}\in p_{x_{e,v}}\IFF f_{e,v}\in S$;
otherwise $v$ would be
trapped in position $p$, which contradicts $p$ being a good
position. In particular, this implies that in both cases player \PII\
can avoid position $\{f_{e,v}^0f_{e,v}^0\}$ in the bijective
$2$-pebble game on $\CT_e,\CT_e$ starting in position
$p\restriction\CT_e$. This enables us to define the restriction of the
bijection $g$ to $V(\CT_e)$.
\begin{ealph}[resume]
\item\label{gap:N} If $e=vw\in E(\CK)$ requires attention at $v$, then we
  determine the set $S$ as above. If $f_{e,v}\in S$, we choose the
  restriction of $g$ to $V(\CT_e)$ according to a strategy for player
  \PII\ in the bijective
$2$-pebble game on $\CT_e,\CT_e$ starting in position
$p\restriction\CT_e$ that avoids $\{f_{e,v}^0f_{e,v}^0\}$. If $f_{e,v}\not\in S$, we choose the
  restriction of $g$ to $V(\CT_e)$ according to a strategy for player
  \PII\ in the bijective
$2$-pebble game on $\CT_e,\CT_e$ starting in position
$p\restriction\CT_e$ that avoids $\{f_{e,v}^0f_{e,v}^1\}$. 
\item\label{gap:O} For every $e=vw\in E(\CK)$ such that $e$ does not require
  attention and $p\restriction \CT_e$
  is nonempty, we choose the restriction of $g$ to $V(\CT_e)$
  according to a winning strategy for \PII\ in the weak bijective
  $k$-pebble game $\CT_e,\CT_e$ starting in position
  $p\restriction\CT_e$. If $|p\restriction\CT_e|=k-1$, we assume
  that in the first step \PI\ selects the unique pair in
  $(p\restriction\CT_e)\setminus (p^-\restriction\CT_e)$.
\item\label{gap:P} For every $e=vw\in E(\CK)$ such that $e$ does not require
  attention and $p\restriction \CT_e$
  is empty, we let the restriction of $g$ to $V(\CT_e)$ be the
  identity mapping.
\end{ealph}
It remains to define $g_v\in V(\CK)$ for $v\in V(\CK)$ that are not in
the domain of $p$. 
\begin{ealph}[resume]
\item\label{gap:Q} If $v\in V(\CK)$ is not in the domain of $p$ and there is
  some edge $e=vw$ that requires attention at $v$, then we
  define the set $S$ as above and
  let $g_v=S$.

\item\label{gap:R} If $v\in V(\CK)$ is not in the domain of $p$ and there is
  no edge that requires attention at $v$, there is an edge $e=vw\in
  E(\CK)$ such that $p\restriction\CT_e=\emptyset$. We choose a set
  $S\subseteq E(v)$ such that for all $e'=vw'$ with nonempty
  $p\restriction\CT_{e'}$ we have $f_{e',v}\in S\IFF p$ flips $e'$. By
  adding $f_{e,v}$ if necessary, we can choose $S$ such that $|S|$ has
  the right parity (even if $v\neq v_1$ and odd if $v=v_1$). We let
  $g_v=S$.
\end{ealph}
It is easy to see
that if player \PII\ selects this bijection $g$,
then regardless of which pair $a'b'$ player \PI\ selects, $p\conc
a'b'$ will be a local isomorphism and the new position $p^-\conc a'b'$
will again be good. Again proof consists of a case distinction along the
cases of the definition of the bijection $g$ in
\ref{gap:M}--\ref{gap:R}.
\end{proof}

\subsection{Boolean arithmetic and $\LL^k$-equivalence} 
\label{boolLksec}
We saw in Section~\ref{boolfracisosec} that equations, which are
direct consequences of the basic continuity and 
compatibility equations w.r.t.\ the
adjacency matrices $A$ and $B$, may carry independent weight 
in their boolean interpretation. This is no surprise, because
the boolean reading is much weaker, especially due to the absorptive 
nature of $\vee$, which unlike $+$ does not allow for inversion. 
$AX = XB$ for doubly stochastic 
$X$ and $A,B \in \B^{n,n}$ implies 
$A^c X = X B^c$.
Similarly, we found in part~(a) of Lemma~\ref{locisolem}
that the continuity equations guarantee that solutions 
are supported by local bijections,
under real arithmetic; this also fails for 
boolean arithmetic. 

We now augment the boolean requirements by corresponding
boolean equations that express 
\bae
\item
compatibility also w.r.t.\ $A^c$ and $B^c$, as in boolean 
fractional isomorphism, %
\item
the new constraint $X_p = 0$ 
whenever $p$ is not a local bijection.
\eae

In the presence of the continuity equations, 
which force monotonicity, it suffices for~(b) 
to force $X_{aa'bb'} = 0$ for 
all $a,a'\in [m]$, $b,b' \in [n]$
such that %
$a=a' \nLeftrightarrow b=b'$.
This is captured by the constraint $\MATCH2$ below.
Together with the continuity and compatibility equations,
$\MATCH2$
then implies that $X_p = 0$ unless $p$ is a local isomorphism, 
just as in the proof of part~(b) of   Lemma~\ref{locisolem}, 
also in terms of boolean arithmetic. 

\medskip
So we now use the following boolean version of the Sherali--Adams
hierarchy $\ISO[k-1]$ and its variant $\ISO[k-1/2]$ for $k\ge 2$.

\begin{center}
\framebox{\begin{minipage}{\textwidth-3em-4mm}
$\BISO[k-1]$
\[
\barr{l}
\left.
\barr{@{}l@{}} 
X_\emptyset = 1
\quad\mbox{and}
\\
\hnt
X_p  = 
\sum_{b'} X_{p\,\widehat{\ }\, ab'} %
= \sum_{a'} X_{p\,\widehat{\ }\, a'b}
\\
\hnt
\mbox{for } |p| < k, a \in [m], b \in [n]
\earr
\qquad
\right\} \quad
\CONT{\ell} \mbox{ for } \ell < k 
\\
\\
\left.
\barr{@{}l@{}}
X_{ab \,\widehat{\ }\,ab'} = 0 = X_{ab \,\widehat{\ }\,a'b}  
\\
\hnt
\mbox{for }  a \not= a' \in [m], b \not= b' \in [n] 
\earr \qquad\;\;\;\;\; \right\}
\quad\, \MATCH2
\\
\\
\left.
\barr{@{}l@{}} 
\sum_{a'} A_{aa'}  X_{p\,\widehat{\ }\, a'b} = 
\sum_{b'} X_{p\,\widehat{\ }\, ab'} B_{b'b}
\\
\hnt
\mbox{for } |p| < k-1,  a \in [m], b \in [n]
\earr
\quad\,
\right\} \quad
\COMP{\ell} \mbox{ for } \ell < k 
\\
\\
\left.
\barr{@{}l@{}} 
\sum_{a'} A^c_{aa'}  X_{p\,\widehat{\ }\, a'b} = 
\sum_{b'} X_{p\,\widehat{\ }\, ab'} B^c_{b'b}
\\
\hnt
\mbox{for } |p| < k-1,  a \in [m], b \in [n]
\earr
\quad\,
\right\} \quad
\COMP{\ell}^c \mbox{ for } \ell < k 
\earr
\]
\end{minipage}}
\end{center}

For $\BISO[k-1/2]$ we require 
$\CONT{\ell}$ for all $\ell \leq k$, i.e., also
for $\ell = k$.

\bR
The systems $\BISO[k-1]$ and $\BISO[k-1/2]$ consist of 
linear boolean equations and can be solved in polynomial time.
\eR

For this observe that the systems $\BISO[k-1]$ and $\BISO[k-1/2]$
consist of equations of the following forms:
\begin{align}
  \label{eq:be1}
  \sum_{i\in I}X_i=\sum_{j\in J}X_j,\\
  \label{eq:be2}
  \sum_{i\in I}X_i=0,\\
 \label{eq:be3}
  \sum_{i\in I}X_i=1.
\end{align}
Those of type~\eqref{eq:be2} are actually subsumed by those of 
type~\eqref{eq:be1} with $J=\emptyset$. It is an easy exercise to prove
that such systems of linear boolean equations can be solved in
polynomial time.

\medskip
The \emph{weak} $k$-pebble game is the straightforward 
adaptation of the weak bijective $k$-pebble game to the setting
without counting. A single round of the game is played as follows.
\begin{enumerate}%
\item If $|p|=k-1$, player \PI\ selects a pair $ab\in p$; 
\\
if $|p|<k-1$, this step is omitted.
\item Player \PI\ chooses an element $a'$ of $\str A$ or $b'$ of $\str B$.
\item 
Player \PII\ must respond with an element $b'$ of $\str B$ 
or $a'$ of $\str A$, respectively, such that  
$aa'$ is an edge of $\str{A}$ if, and only if, 
$bb'$ is an edge of $\str{B}$.
\item %
If $|p|<k-1$, then the new position is $p' := p \,\widehat{\ }\,a'b'$;
\\ 
if $|p|=k-1$, then the new position is $p' := 
(p\!\setminus\! ab)\,\widehat{\ }\,a'b'$. 
\end{enumerate}
$\PII$ loses if she cannot respond in step~(3) or if
the resulting position $p'$ fails to
be a local isomorphism.

\medskip
We denote weak $k$-pebble equivalence as in $\str{A}\equiv_\LL^{<k}
\str{B}$ and extend this to $\str{A},\abar \equiv_\LL^{<k} \str{B},\bbar$
for tuples $\vec a,\vec b$ of the same length $<k$.  We sketch a proof
of the following, which is a boolean analogue of the correspondences
between half-step levels of Sherali--Adams and $\LC^k$- and
$\LC^{<k}$-equivalence  established in Section~\ref{sec:sa}.

\bT
\label{boolk-1/2thm}
$\BISO[k-1/2]$ has a solution (w.r.t. boolean arithmetic) if, and only
if, $\str A\equiv_\LL^{k}\str B$.
\eT

\bT
\label{boolk-1thm} 
$\BISO[k-1]$ has a solution (w.r.t. boolean arithmetic) if, and only
if, $\str A\equiv_\LL^{<k}\str B$.
\eT

Towards the proofs of the critical directions,
viz., from solutions to equivalences mediated by pebble games,
we want to pass from given solutions to induced \emph{good}
solutions, from which strategies can be directly extracted. 
These are characterised in the boolean case by conditions 
that are analogous to those of Definition~\ref{goodtranslationdef}
for the real case; good solutions are induced 
by arbitrary solutions in a manner that is analogous to our findings
in Corollaries~\ref{translatetogoodcor} 
and~\ref{translatetogoodcompcor}, essentially through 
reductions via liftings to tuple co-ordinates.

The relevant liftings of equations to tuple co-ordinates 
are also analogous to those in the real case, now including
compatibility equations for the complements of the edge 
relations in $\str{A}$ and $\str{B}$. We leave out $\MATCH2$, 
whose lifting says that $\check{X}$ is 
supported by local bijections.
\[
\mbox{$\displaystyle%
\barr{l}
\left.%
\barr{@{}l@{}}
\sum_{\abar'} 
\check{X}_{\abar'\bbar} =
\sum_{\bbar'} 
\check{X}_{\abar\bbar'} = 1
\\
\\
\sum_{\abar'} \mathbb{I}^{\ssc
  (i)}_{\abar\abar'} \check{X}_{\abar',\bbar} 
=
\sum_{\bbar'} \check{X}_{\abar,\bbar'}\mathbb{I}^{\ssc
  (i)}_{\bbar'\bbar} 
\\
\hnt
\mbox{for all } 
\abar \in [m]^\ell, \bbar \in [n]^\ell, \mbox{ and all }i \in [\ell]
\earr
\qquad\;
\right\} 
\quad \mbox{lifting of }\CONT{\ell'}, \ell' \leq \ell 
\\
\\
\left.%
\barr{@{}l@{}}
\sum_{\abar'}%
A_{\abar,\abar'}^{\ssc (i)} \check{X}_{\abar',\bbar}
=
\sum_{\bbar'}%
\check{X}_{\abar,\bbar'} B^{\ssc (i)}_{\bbar',\bbar}
\\
\hnt
\mbox{for all } 
\abar \in [m]^\ell, \bbar \in [n]^\ell, \mbox{ and all }i \in [\ell]
\earr
\quad\;\, \quad
\right\} 
\quad \mbox{lifting of }\COMP{\ell'}, \ell'  \leq \ell
\\
\\
\left.%
\barr{@{}l@{}}
\sum_{\abar'}%
(A^c)_{\abar,\abar'}^{\ssc (i)} \check{X}_{\abar',\bbar}
=
\sum_{\bbar'}%
\check{X}_{\abar,\bbar'} (B^c)^{\ssc (i)}_{\bbar',\bbar}
\\
\hnt
\mbox{for all } 
\abar \in [m]^\ell, \bbar \in [n]^\ell, \mbox{ and all }i \in [\ell]
\earr
\quad\;\, \quad
\right\} 
\quad \mbox{lifting of }\COMP{\ell'}^c, \ell'  \leq \ell 
\earr$}%
\]

Here the matrices $(A^c)^{\ssc (i)}$ 
are the lifting of $A^c$ to tuple co-ordinates, 
just as the $A^{\ssc (i)}$ are the familiar liftings of $A$, as
introduced for Lemma~\ref{translationcompeqnslem}, 
for $i \in [\ell]$:  
\[
\textstyle
(A^c)_{\abar,\abar'}^{\ssc (i)} =
\prod_{j \not= i} \delta(a_j,a_j') \; A^c_{a_ia'_i}. 
\]

As the relationship between the edge relations 
$(A^c)^{\ssc (i)}$ and $A^{\ssc (i)}$ is not one of complementation
over $[m]^\ell$, the commutativity conditions for $\check{X}$ in 
$\COMP{\ell}$ and $\COMP{\ell}^c$ do not give rise to 
\emph{bi-stable} boolean equivalent partitions as in the boolean variant
of fractional isomorphism. Instead, the lifting to tuple co-ordinates
of a good boolean solution 
for $\COMP{\ell}$ and $\COMP{\ell}^c$ induces partitions that are 
just \emph{boolean stable} and \emph{boolean equivalent}, but 
simultaneously so, for $(A^c)^{\ssc (i)}/(B^c)^{\ssc (i)}$ and 
$A^{\ssc (i)}/B^{\ssc (i)}$.
Similarly the lifting of a boolean solution to the continuity
equations that is supported by local bijections induces good 
solutions with partitions that are just \emph{boolean stable} 
and \emph{boolean equivalent} w.r.t.\ the $\mathbb{I}^{\ssc (i)}$. It turns
out that these are precisely the conditions that support strategies for the second
player in the corresponding $k$-pebble games. 

Let us say that partitions of the vertex sets $[m]^\ell = \dot{\bigcup}_s D_s$ 
and $[n]^\ell = \dot{\bigcup}_s D_s'$ are \emph{equivalent}
and \emph{boolean stable} w.r.t.\ 
edge relations $E$ on $[m]^\ell$ and 
$E'$ on $[n]^\ell$ if for all partition
indices $s,t$, and all $\abar \in D_s, \bbar \in D_s'$,
\[
\textstyle 
\{ \abar' \in D_t \colon (\abar,\abar') \in E \} \not=\emptyset
\;\Leftrightarrow\; 
\{ \bbar' \in D_t' \colon (\bbar,\bbar') \in E' \} \not=\emptyset.
\]

We shall apply this notion to the undirected reflexive edge 
relations of the $\mathbb{I}^{\ssc (i)}$, where the condition 
becomes
\[
\textstyle
\{ a \in [m] \colon \abar\frac{a}{i} \in D_t \} \not=\emptyset
\;\Leftrightarrow\; 
\{ b \in [n] \colon \bbar\frac{b}{i} \in D_t' \} \not=\emptyset;
\]
and to the symmetric and irreflexive edge relations 
$A^{\ssc (i)}/B^{\ssc (i)}$ and $(A^c)^{\ssc (i)}/(B^c)^{\ssc (i)}$,
for which the combination of the conditions
\[
\barr{r@{\;\;\Leftrightarrow\;\;}l}
\textstyle 
\{ a \in [m] \colon A_{a_ia} = 1 \wedge
\abar\frac{a}{i} \in D_t \} \not=\emptyset
&
\{ b \in [n] \colon B_{b_ib} = 1 \wedge \bbar\frac{b}{i} \in D_t' \} \not=\emptyset,
\\
\hnt
\{ a \in [m] \colon A_{a_ia} = 0 \wedge
\abar\frac{a}{i} \in D_t \} \not=\emptyset
&
\{ b \in [n] \colon B_{b_ib} = 0 \wedge \bbar\frac{b}{i} \in D_t' \} \not=\emptyset
\earr
\]
becomes the natural lifting of bi-stability and equivalence
conditions for boolean fractional isomorphism.

\bD
\label{goodtranslationbooldef}
A boolean solution $X = (X_p)_{|p| \leq \ell}$ 
to $\CONT{\ell'}$ for $\ell' \leq \ell$ is \emph{good} if its lifting 
$\check{X} = (\check{X}_{\abar,\bbar})$
to tuple co-ordinates is a good boolean solution 
to the equations $\comp[\mathbb{I}^{\ssc (i)}(m,\ell),
\mathbb{I}^{\ssc (i)}(n,\ell)]$ (simultaneously for all $i \in [\ell]$)
in the sense of Definition~\ref{goodsoldef}.
\eD

We obtain the analogue of Corollary~\ref{translatetogoodcor},
which can also be proved in complete analogy, simply by specialisation
to boolean arithmetic, which can then be used to prove the critical 
direction in Theorem~\ref{boolk-1/2thm}.

\bL
\label{translatebooltogoodlem} 
Any boolean solution $X = (X_p)_{|p| \leq \ell}$ to %
$\CONT{\ell'}$ for $\ell' \leq \ell$ 
that is supported by local isomorphisms between 
$\str{A}$ and $\str{B}$, induces a solution 
$X' = (X'_p)_{|p| \leq \ell}$ that is \emph{good} in the sense of
Definition~\ref{goodtranslationbooldef}, and is non-zero where $X$ is.
\eL

Now we are ready to prove the theorem.

\prf[Proof of Theorem~\ref{boolk-1/2thm}]
We start with the implication from right to left. %
If
$\str{A} \equiv_\LL^k \str{B}$ we put, for $p = \abar \bbar$ of size $|p| \leq k$,
\[
X_p := \left\{ \barr{ll}
1 & \mbox{if } \str{A},\abar \equiv_\LL^k \str{B},\bbar,
\\
\hnt
0 &\mbox{else.}
\earr\right.
\]

Clearly this assignment satisfies $\MATCH2$, and one  
easily checks that it also satisfies 
the boolean continuity equations $\CONT{\ell}$ for $\ell
\leq k$. For the boolean compatibility 
equations $\COMP{\ell}$ and $\COMP{\ell}^c$
for $\ell < k$, let us check, for instance, an equation  
$\COMP{k-1}$. The non-trivial case is that of
$p = \abar\bbar$ where $\abar \in [n]^{k-2}$ and 
$\bbar \in [m]^{k-2}$ are such that $\str{A},\abar \equiv_\LL^k
\str{B},\bbar$ so that $X_p =1$. 
Consider the instance of equation $\CONT{k-1}$ for $a\in[m],b\in[n]$:
\[
\sum_{a'} A_{aa'}  X_{p\,\widehat{\ }\, a'b} = 
\sum_{b'} X_{p\,\widehat{\ }\, ab'} B_{b'b}.
\]
Since $\str{A},\abar \equiv_\LL^k \str{B},\bbar$, 
there is a $\hat{b} \in [m]$ such that
$\str{A}, \abar a  \equiv_\LL^k \str{B},\bbar \hat{b}$. 
Suppose the left-hand side of the equation evaluates
to $1$. This means that there is an edge in $\str{A}$ from 
$a$ to some $a'$ for which 
$X_{p\,\widehat{\ }\,a'b} = 1$, i.e., for which 
$\str{A},\abar a' \equiv_\LL^k \str{B},\bbar b$.

In other words, there is an edge in $\str{A}$ from some 
$a'$ for which $\str{A},\abar a' \equiv_\LL^k \str{B},\bbar b$
to some $a$ for which $\str{A},\abar a \equiv_\LL^k  \str{B},\bbar \hat{b}$. 
So every realisation of the $\LL^k$-type of 
$\str{B},\bbar b$ has an edge between $b$ and some $b'$ where 
$\str{B},\bbar b' \equiv_\LL^k \str{A},\abar a$, which implies that the
right-hand side of the equation evaluates to $1$, too. 

For the implication from left to right in part~(b)
we extract a strategy for player \PII\ in the $k$-pebble game 
from a good solution as provided in Lemma~\ref{translatebooltogoodlem}. 
As discussed above, any
good solution is supported by local isomorphisms, whence it suffices
for \PII\ to  maintain the condition that $\check{X}_{\abar,\bbar} =
1$, or, equivalently, that the pebbled tuples are
in matching partition sets. This can be achieved for rounds played in
the $j$-th component, because the partitions induced by the lifting of 
the good solution are boolean equivalent and stable w.r.t.\ 
the edge relation $\mathbb{I}^{\ssc (j)}$.
\eprf

We turn to $\BISO[k-1]$ and the situation of Theorem~\ref{boolk-1thm}, 
where the higher levels of the compatibility equations matter. 
For the analogue of Corollary~\ref{translatetogoodcompcor}
we also need to reason that any solution to $\BISO(k-1)$ is 
supported by local
isomorphisms. Support by local bijections is clear, because 
that is explicitly demanded by $\MATCH2$. For the strengthening to
local isomorphy, we may reason on the basis of $\COMP{2}$ exactly as for
Lemma~\ref{locisolem}~(b).
The rest of the proof of the lemma is again strictly analogous 
to the argument for Corollary~\ref{translatetogoodcompcor}.

\bL
\label{translatebooltogoodcomplem} 
Let $k \geq 3$. Any boolean solution $X = (X_p)_{|p| < k}$ 
to $\BISO(k-1)$ induces a boolean solution  
$X' = (X_p')_{|p| < k}$ to $\BISO(k-1)$ that is supported by
local isomorphisms 
and is \emph{good} in the sense that its lifting $\check{X}'$ 
to tuple co-ordinates is a good simultaneous boolean solution  
to the following equations for all $i < k$:
\[
\comp[\mathbb{I}^{\ssc (i)}(m,k-1),\mathbb{I}^{\ssc (i)}(n,k-1)],
\comp[A^{\ssc (i)},B^{\ssc (i)}],
\comp[(A^c)^{\ssc (i)},(B^c)^{\ssc (i)}], 
\]
and thus
induces $\check{X}'$-related partitions of 
$[m]^{k-1} = \dot{\bigcup}_s D_s$ and 
$[n]^{k-1} =  \dot{\bigcup}_s D_s'$ such that
\bre
\item 
these partitions are boolean equivalent and stable 
w.r.t.\ the edge relations of the $\mathbb{I}^{\ssc (i)}$;
\item
these partitions are boolean equivalent and stable 
w.r.t.\ the liftings of the edge relations 
$A^{\ssc (i)}/B^{\ssc (i)}$  as well as 
$(A^c)^{\ssc (i)}/(B^c)^{\ssc (i)}$; %
\item
$\check{X}'_{D_sD_s'} = 1$ for all $s$;
\item
$\check{X}_{D_sD_t}' = 0$ for all $s\not= t$.
\ere
In particular, $\abar \mapsto \bbar$ is a local isomorphism between 
$\str{A}$ and $\str{B}$ for $\abar$ and $\bbar$ from matching
partition sets. Moreover, $X'$ is non-zero where the given $X$ is.
\eL

We are now ready to prove Theorem~\ref{boolk-1thm}.

\prf[Proof of Theorem~\ref{boolk-1thm}]
For the direction from right to left, suppose that $\str{A} \equiv_\LL^{<k}
\str{B}$ and let
$X_\emptyset := 1$ and, for $\abar \in [n]^{k-1}$ and $\bbar \in [m]^{k-1}$,
$X_{\abar \bbar} := 1$ iff $\str{A},\abar \equiv_\LL^{<k} \str{B},\bbar$.
It is clear that this assignment satisfies $\MATCH2$ and 
$\CONT{\ell}$ for $\ell < k$.
Consider then an instance of $\COMP{\ell}$ for $\ell < k$,
\begin{equation}
  \label{eq:bcompl}
\sum_{a'} A_{aa'}  X_{p\,\widehat{\ }\, a'b} = 
\sum_{b'} X_{p\,\widehat{\ }\, ab'} B_{b'b}, 
\end{equation}
where $a \in [n], b \in [m]$, $|p| <  k-1$,
$p = \abar \bbar$. Let us assume that $|p|=k-2$; this is the most
difficult case. If $X_p=0$ then $X_{p\,\widehat{\ }\,a'b'}=0$ for all
$a'b'$, and thus equation \eqref{eq:bcompl} is trivially satisfied. So assume
$X_p=1$, that is,
$\str{A},\abar \equiv_\LL^{<k} \str{B},\bbar$. 
Suppose for instance that the left-hand side of equation \eqref{eq:bcompl}
evaluates to $1$, i.e., that there is some $a'$ adjacent to $a$ in
$\str A$ for which  
$\str{A}, \abar a' \equiv_\LL^{<k} \str{B},\bbar b$. Consider the weak
$k$-pebble game in position $\vec aa'\vec bb$. Assume player \PI\
selects the pair $a'b$ in the first step of the next round and selects
$a$ in the second step. Let $b'$ be the answer of \PII\ when she
plays according to her winning strategy. Then $\vec aa'a\mapsto\vec
bbb'$ is a local isomorphism and the new position $\vec aa\vec bb'$
is a winning position for player \PII, that is, $\str A,\vec
aa\equiv_\LL^{<k}\str B,\vec bb'$. Since $\vec aa'a\mapsto\vec
bbb'$ is a local isomorphism and $aa'$ is an edge of $\str A$, the
pair $bb'$ is an edge of $\str B$ and thus $B_{b'b}=1$. Since $\str A,\vec
aa\equiv_\LL^{<k}\str B,\vec bb'$, we have $X_{p\widehat{\
  }ab'}=1$. Thus the right-hand side of 
equation \eqref{eq:bcompl} evaluates to $1$ as well.

For the direction from left to right
we extract a strategy for \PII\ in the weak 
$k$-pebble game from a good solution $X$ and its lifting 
$\check{X}$ as provided in Lemma~\ref{translatebooltogoodcomplem}. Again, the 
strategy for \PII\ is to  maintain the condition that $\check{X}_{\abar,\bbar} =
1$. In addition, in a round played on the $j$-th component, 
the old and new positions of the $j$-th pebble 
must be linked by an edge in $\str{A}$ if, and only if, they are
linked by an edge in $\str{B}$. This can be achieved  
because the partitions induced by the lifting of 
the good solution are simultaneously boolean equivalent and stable w.r.t.\ 
$A^{\ssc (j)}/B^{\ssc (j)}$ and $(A^c)^{\ssc (j)}/(B^c)^{\ssc (j)}$.
\eprf

\bR
\label{seprem}
For all $k \geq 3$, 
$\equiv_\LL^{k-1}$, $\equiv_\LL^{<k}$, $\equiv_\LL^{k}$ form a strictly
increasing hierarchy of discriminating power.
\eR

\prf
The examples for the gaps between $\equiv_\LC^{k-1}$, 
$\equiv_\LC^{<k}$, $\equiv_\LC^{k}$ given in Section~\ref{gapsec}, are in
fact good in the setting without counting. The strategy 
analysis given there does not involve counting in any 
non-trivial manner.
\eprf

\paragraph*{Acknowledgements}
We are most grateful to Albert Atserias for his valuable comments 
on an earlier draft of this exposition. 
\\
The second author gratefully acknowledges the academic hospitality 
in the first author's group at HU Berlin during his sabbatical in 2011/12.



\begin{thebibliography}{10}

\bibitem{AtseriasManeva}
A.~Atserias and E.~Maneva.
\newblock {S}herali--{A}dams relaxations and indistinguishability in counting
  logics.
\newblock In {\em Innovations in Theoretical Computer Science (ITCS)}, 2012.

\bibitem{bar77}
J.~Barwise.
\newblock On {M}oschovakis closure ordinals.
\newblock {\em Journal of Symbolic Logic}, 42:292--296, 1977.

\bibitem{bergro15}
C.~Berkholz and M.~Grohe.
\newblock Limitations of algebraic approaches to graph isomorphism testing.
\newblock {\em ArXiv}, arXiv:1502.05912 [cs.CC], 2015.

\bibitem{bieozb04}
D.~Bienstock and N.~Ozbay.
\newblock Tree-width and the {S}herali-{A}dams operator.
\newblock {\em Discrete Optimization}, 1:13--21, 2004.

\bibitem{burgalhoo+03}
J.~Buresh-Oppenheim, N.~Galesi, S.~Hoory, A.~Magen, and T.~Pitassi.
\newblock Rank bounds and integrality gaps for cutting planes procedures.
\newblock In {\em Proceedings of the 43rd Annual IEEE Symposium on Foundations
  of Computer Science}, pages 318--327, 2003.

\bibitem{caifurimm92}
J.~Cai, M.~F{\"u}rer, and N.~Immerman.
\newblock An optimal lower bound on the number of variables for graph
  identification.
\newblock {\em Combinatorica}, 12:389--410, 1992.

\bibitem{charmakmak09}
M.~Charikar, K.~Makarychev, and Y.~Makarychev.
\newblock Integrality gaps for {S}herali-{A}dams relaxations.
\newblock In {\em Proceedings of the 41st ACM Symposium on Theory of
  Computing}, pages 283--292, 2009.

\bibitem{graott93}
E.~Gr{\"a}del and M.~Otto.
\newblock Inductive definability with counting on finite structures.
\newblock In E.~B{\"o}rger, G.~J{\"a}ger, H.~Kleine B{\"u}ning, S.~Martini, and
  M.M. Richter, editors, {\em Computer Science Logic, 6th Workshop, CSL `92,
  San Miniato 1992, Selected Papers}, volume 702 of {\em Lecture Notes in
  Computer Science}, pages 231--247. Springer-Verlag, 1993.

\bibitem{gro10}
M.~Grohe.
\newblock Fixed-point definability and polynomial time on graphs with excluded
  minors.
\newblock In {\em Proceedings of the 25th IEEE Symposium on Logic in Computer
  Science}, 2010.

\bibitem{hel92}
L.~Hella.
\newblock Logical hierarchies in {P}{T}{I}{M}{E}.
\newblock In {\em Proceedings of the 6th IEEE Symposium on Logic in Computer
  Science}, pages 360--368, 1992.

\bibitem{imm82}
N.~Immerman.
\newblock Upper and lower bounds for first-order expressibility.
\newblock {\em Journal of Computer and System Sciences}, 25:76--98, 1982.

\bibitem{immlan90}
N.~Immerman and E.~Lander.
\newblock Describing graphs: A first-order approach to graph canonization.
\newblock In A.~Selman, editor, {\em Complexity theory retrospective}, pages
  59--81. Springer-Verlag, 1990.

\bibitem{lau10}
B.~Laubner.
\newblock Capturing polynomial time on interval graphs.
\newblock In {\em Proceedings of the 25th IEEE Symposium on Logic in Computer
  Science}, pages 199--208, 2010.

\bibitem{mal14}
P.N. Malkin.
\newblock Sherali–{A}dams relaxations of graph isomorphism polytopes.
\newblock {\em Discrete Optimization}, 12:73--97, 2014.

\bibitem{matsin09}
C.~Mathieu and A.~Sinclair.
\newblock Sherali-{A}dams relaxations of the matching polytope.
\newblock In {\em Proceedings of the 41st ACM Symposium on Theory of
  Computing}, pages 293--302, 2009.

\bibitem{ott97}
M.~Otto.
\newblock {\em Bounded variable logics and counting -- {A} study in finite
  models}, volume~9 of {\em Lecture Notes in Logic}.
\newblock Springer-Verlag, 1997.

\bibitem{RamanaScheinermanUllman}
M.~Ramana, E.~Scheinerman, and D.~Ullman.
\newblock Fractional isomorphism of graphs.
\newblock {\em Discrete Mathematics}, 132:247--265, 1994.

\bibitem{ScheinermanUllman}
E.~Scheinerman and D.~Ullman.
\newblock {\em Fractional Graph Theory}.
\newblock Wiley, 1997.

\bibitem{sch08}
G.~Schoenebeck.
\newblock Linear level {L}asserre lower bounds for certain k-{CSP}s.
\newblock In {\em Proceedings of the 49th Annual IEEE Symposium on Foundations
  of Computer Science}, pages 593--602, 2008.

\bibitem{schtretul07}
G.~Schoenebeck, L.~Trevisan, and M.~Tulsiani.
\newblock Tight integrality gaps for {L}ov{\'a}sz-{S}chrijver {LP} relaxations
  of vertex cover and max cut.
\newblock In {\em Proceedings of the 39th ACM Symposium on Theory of
  Computing}, pages 302--310, 2007.

\bibitem{sheada90}
H.D. Sherali and W.P. Adams.
\newblock A hierarchy of relaxations between the continuous and convex hull
  representations for zero-one programming problems.
\newblock {\em SIAM Journal on Discrete Mathematics}, 3:411, 1990.

\bibitem{tin86}
G.~Tinhofer.
\newblock Graph isomorphism and theorems of {B}irkhoff type.
\newblock {\em Computing}, 36:285--300, 1986.

\bibitem{tin91}
G.~Tinhofer.
\newblock A note on compact graphs.
\newblock {\em Discrete Applied Mathematics}, 30:253--264, 1991.

\bibitem{wei76}
B.~Weisfeiler.
\newblock {\em On Construction and Identification of Graphs}, volume 558 of
  {\em Lecture Notes in Mathematics}.
\newblock Springer-Verlag, 1976.

\end{thebibliography}

\end{document}